\documentclass[]{jair}

\setcopyright{cc}
\copyrightyear{2025}
\acmYear{2025}
\acmDOI{10.1613/jair.1.xxxxx}

\JAIRAE{Insert JAIR AE Name}
\JAIRTrack{Insert JAIR Track Name Here}
\acmVolume{4}
\acmArticle{111}
\acmMonth{8}
\acmYear{2025}

\usepackage{amsthm,amsmath}
\usepackage{mathtools}
\usepackage{complexity}
\usepackage{thm-restate}
\usepackage{cleveref}

\usepackage[disable]{todonotes}

\usepackage{tikz,pgf}
\usetikzlibrary{positioning}
\usetikzlibrary{calc}
\usepackage{graphicx}
\usepackage{subfig}
\usepackage[shortcuts]{extdash} 
\usetikzlibrary{decorations.pathreplacing,calc}

\tikzset{%
  thick middle dotted line/.style={
    decoration={show path construction,
      lineto code={
          \draw[#1,line width=1pt, color=red] (\tikzinputsegmentfirst) --($(\tikzinputsegmentfirst)!.3333!(\tikzinputsegmentlast)$);,
          \draw[dotted,#1,color=red] ($(\tikzinputsegmentfirst)!.3333!(\tikzinputsegmentlast)$)--($(\tikzinputsegmentfirst)!.6666!(\tikzinputsegmentlast)$);,
          \draw[#1,line width=1pt, color=red] ($(\tikzinputsegmentfirst)!.6666!(\tikzinputsegmentlast)$)--(\tikzinputsegmentlast);,
      }
    },
    decorate
  },
}
\tikzset{%
  middle dotted line/.style={
    decoration={show path construction,
      lineto code={
          \draw[#1] (\tikzinputsegmentfirst) --($(\tikzinputsegmentfirst)!.3333!(\tikzinputsegmentlast)$);,
          \draw[dotted,#1] ($(\tikzinputsegmentfirst)!.3333!(\tikzinputsegmentlast)$)--($(\tikzinputsegmentfirst)!.6666!(\tikzinputsegmentlast)$);,
          \draw[#1] ($(\tikzinputsegmentfirst)!.6666!(\tikzinputsegmentlast)$)--(\tikzinputsegmentlast);,
      }
    },
    decorate
  },
}

\newcommand{\MAPF}{\textsc{Multiagent Path Finding}\xspace}
\newcommand{\MAPFShort}{\textsc{MAPF}\xspace}
\newcommand{\CMAPF}{\textsc{Colored Multiagent Path Finding}\xspace}
\newcommand{\CMAPFShort}{\textsc{Colored MAPF}\xspace}
\newcommand{\PAMAPF}{\textsc{Partially Anonymous Multiagent Path Finding}}

\newcommand{\N}{\mathbb{N}}
\newcommand{\Nzero}{\N_0} 
\newcommand{\dc}{\textit{dc}}
\newcommand{\Integers}[1]{[#1]}
\newcommand{\CC}[1]{\textsf{#1}} 
\newcommand{\PN}[1]{\textsc{#1}\xspace} 
\DeclareMathOperator{\dist}{dist}

\newtheorem{lemma}{Lemma}
\newtheorem{observation}{Observation}

\newtheorem{theorem}{Theorem}
\newtheorem{corollary}{Corollary}
\newtheorem{claim}{Claim}
\newenvironment{proofclaim}{\noindent{\em Proof of the claim.}}{\qedclaim}
\newcommand{\qedclaim}{\hfill $\diamond$ \medskip}

\newtheorem{recipe}{Recipe}

\begin{document}

\title[Solving Multiagent Path Finding on Highly Centralized Networks]{Solving Multiagent Path Finding on Highly Centralized Networks}

\author{Foivos Fioravantes}
\email{foivos.fioravantes@fit.cvut.cz}
\orcid{0000-0001-8217-030X}
\affiliation{%
  \institution{Department of Theoretical Computer Science, Faculty of Information Technology, Czech Technical University in Prague}
  \city{Prague}
  \country{Czech Republic}
}

\author{Dušan Knop}
\orcid{0000-0003-2588-5709}
\email{dusan.knop@fit.cvut.cz}
\affiliation{%
  \institution{Department of Theoretical Computer Science, Faculty of Information Technology, Czech Technical University in Prague}
  \city{Prague}
  \country{Czech Republic}
}

\author{Jan Matyáš Křišťan}
\orcid{0000-0001-6657-0020}
\email{kristja6@fit.cvut.cz}
\affiliation{%
  \institution{Department of Theoretical Computer Science, Faculty of Information Technology, Czech Technical University in Prague}
  \city{Prague}
  \country{Czech Republic}
}

\author{Nikolaos Melissinos}
\orcid{0000-0002-0864-9803}
\email{}
\affiliation{%
  \institution{Department of Theoretical Computer Science, Faculty of Information Technology, Czech Technical University in Prague}
  \city{Prague}
  \country{Czech Republic}
}
\affiliation{%
  \institution{Computer Science Institute, Faculty of Mathematics and Physics, Charles University}
  \city{Prague}
  \country{Czech Republic}
}

\author{Michal Opler}
\orcid{0000-0002-4389-5807}
\email{michal.opler@fit.cvut.cz}
\affiliation{%
  \institution{Department of Theoretical Computer Science, Faculty of Information Technology, Czech Technical University in Prague}
  \city{Prague}
  \country{Czech Republic}
}

\author{Tung Anh Vu}
\orcid{0000-0002-8902-5196}
\email{tung@iuuk.mff.cuni.cz}
\affiliation{%
  \institution{Computer Science Institute, Faculty of Mathematics and Physics, Charles University}
  \city{Prague}
  \country{Czech Republic}
}

\begin{abstract}
The \textsc{Mutliagent Path Finding} (MAPF) problem consists of identifying the trajectories that a set of agents should follow inside a given network in order to reach their desired destinations as soon as possible, but without colliding with each other. We aim to minimize the maximum time any agent takes to reach their goal, ensuring optimal path length. In this work, we complement a recent thread of results that aim to systematically study the algorithmic behavior of this problem, through the parameterized complexity point of view.

First, we show that MAPF is NP-hard when the given network has a star-like topology (bounded vertex cover number) or is a tree with $11$ leaves. Both of these results fill important gaps in our understanding of the tractability of this problem that were left untreated in the recent work of Fioravantes et al., Exact Algorithms and Lowerbounds for Multiagent Path Finding: Power of Treelike Topology, presented in AAAI'24. Nevertheless, our main contribution is an exact algorithm that scales well as the input grows (FPT) when the topology of the given network is highly centralized (bounded distance to clique). This parameter is significant as it mirrors real-world networks. In such environments, a bunch of central hubs (e.g., processing areas) are connected to only few peripheral nodes.
\end{abstract}



\received{20 February 2007}
\received[revised]{12 March 2009}
\received[accepted]{5 June 2009}

\maketitle

\section{Introduction }
Collision avoidance and optimal navigation are crucial in scenarios that require a group of agents to move through a network in an autonomous way. The real-world applications for such problems are well documented and range from virtual agents moving through a video game level~\cite{SnapeGBLM12}, to actual robots transporting goods around a warehouse~\cite{WurmanDM08,LiTKDKK21}. Agents can be required to move through simpler topologies that span multiple ``floors''~\cite{VelosoBCR15}, or much more complicated topologies on a single ``floor'', e.g., in Airport surface management~\cite{MorrisPLMMKK16}.

Many different models have been proposed to capture and solve the scenarios described above~\cite{SternSFK0WLA0KB19manysurvey}. This is due in part to the ubiquity of these scenarios, but also to the variety of slightly different requirements that each real-world application demands. For example, there are models where a set of positions must be reached, regardless of which agent reaches each one~\cite{AliY24}, while other models require each agent to reach a specific predefined position. Many models consider the centralized setting~\cite{HSLK19,sharon2015conflict}, while others consider the decentralized setting~\cite{SkrynnikAYP24}.

The version we consider, usually appearing under the name of \MAPF (\MAPFShort for short), consists of finding non-colliding paths for a set of agents that move through a given graph; each agent has to move from its source to its target vertex, both of which are predefined and unique, in an optimal way. That is, the agents move in such a way that the \emph{length} of the schedule is minimized. 
To the best of our knowledge, this problem first appeared more than $30$ years ago~\cite{HSS84,Reif79}.

Interestingly, the \MAPFShort{} problem is computationally hard. In particular, it is \NP-complete~\cite{surynek2010optimization}, and the hardness holds even in planar graphs~\cite{Yu16}. Thus, the research community has mainly focused on heuristic algorithms that lead to efficiently solving the problem in practice. Some first positive results were based on the $A^*$ algorithm~\cite{HartNR68}, while, more recently, many researchers have focused on exploiting highly efficient solvers; SAT-solvers~\cite{SurynekSFB17} and SMT-solvers~\cite{Surynek19,Surynek20}. There is also a multitude of greedy heuristic algorithms, such as the one based on \emph{Large Neighborhood Search~\cite{PhanHDK24} }. See~\cite{SternSFK0WLA0KB19manysurvey} for a more comprehensive list of approaches of this flavor.

We follow a slightly different approach to attack this problem. Our goal is to design exact algorithms that are guaranteed to return an optimum solution. In view of the computational hardness of the \MAPFShort problem, this is not possible in the general case. But what if we consider restricted instances? Can we identify characteristics, henceforth called \emph{parameters}, whose exploitation leads to efficient algorithms? This is exactly the question at the core of the parameterized analysis that has recently been initiated for this problem~\cite{EibenGK23,FKKMO24}.

\subsection*{Our Contribution}

The first two parameters that one would hope could be exploited are the number of agents $k$ and the \emph{makespan}~$\ell$, i.e., the total length of the schedule; these are exactly the two natural parameters of the \MAPFShort problem. Unfortunately, the results of~\citet{FKKMO24} indicate that it is highly unlikely for either of these two parameters to be sufficient on their own.
Thus, we turn towards structural parameters of the problem. Both~\cite{EibenGK23} and~\cite{FKKMO24} provide a plethora of infeasibility results for many topologies which we could expect to come up in practical applications. Indeed, the \MAPFShort problem remains \NP-complete even when the input graph is a tree, or is planar and the makespan is at most $3$.

Our first result addresses a gap in the tractability landscape that started being painted in~\cite{FKKMO24} and it concerns graphs of bounded vertex cover number, a parameter whose exploitation usually leads to efficient algorithms (see, e.g.,~\cite{fellows2009complexity}). Recall that the vertex cover of a graph is a set of vertices $S$ such that all the edges of the graph are incident to at least one vertex of $S$. Note that the rest of the vertices of the graph form an independent set (there are no edges between them). This is the key observation, as it means that in any round of the schedule, no more than $|S|$ agents can move. So, intuitively, we could hope for an algorithm that would only focus on monitoring the agents that move through the vertices of $S$ during every turn. Surprisingly, our first result shows that this approach will not work.

\begin{restatable}{theorem}{thmNpVc}\label{thm:vc-hardness:mapf-vc-hard}
\MAPF remains \NP-hard on trees with nine internal vertices and vertex cover number~7.
\end{restatable}

We stress here that the above result is based on the tree having many leaves.
In fact, if a tree has few internal vertices and leaves, the entire instance is bounded in size, and finding a solution cannot be intractable.
Therefore, we ask the following question: What if the tree has only a few leaves?
A generalization of this is captured by the parameter known as the max leaf number: the maximum number of leaves a spanning tree of a graph can have.
Note that the max leaf number is a parameter that limits the structure of the input in a rather drastic way---the graph is essentially a subdivided tree with few edges (consequently, few internal vertices and leaves).
Graphs admitting this structure appear quite often as transportation networks, in particular, underground trains~\cite{Eppstein15}.
Indeed, many problems are tractable with respect to the max leaf number (see, e.g., \cite{fellows2009complexity}); a notable exception is the \textsc{Graph Motif} problem~\cite{BonnetS17}.
Here, we obtain yet another intractability result for this parameter.

\begin{restatable}{theorem}{thmNpTrees}\label{thm:pancake}
 \MAPF remains \NP-hard on trees with $11$ leaves.
\end{restatable}

Additionally, minor modifications to the previous theorem yield a hardness result in a more general setting involving semi-anonymous agents.
The input of the \CMAPF (\CMAPFShort) problem~\cite{BartakM21,BartakIS21} consists of a graph $G$, and multiple \emph{types} of agents. All agents have their predefined starting positions, and each type of agents has a specific subset of the vertices of $G$ which are the ending positions of that type; we are satisfied as long as every agent of a type ends up on \emph{any} ending position of that type.

\begin{restatable}{theorem}{thmNpColored}
\CMAPFShort remains \NP-hard even if there are only~$6$ groups of agents and $G$ is a tree with $11$~leaves.
\end{restatable}

All the results we have so far seem to imply that having a sparse topology is not helpful. So, what about dense topologies? Our final result, which is the main algorithmic contribution of this paper, is an efficient algorithm that solves \MAPFShort on graphs of bounded \emph{distance to clique}. A graph~$G$ has distance to clique $d$ if there is a set of $d$ vertices $S$, whose removal renders the graph complete. Such topologies can be expected to appear in real-life applications where agents are required to move through a very dense network (the clique) which is connected to few docking stations (the vertices of $S$). We stress here that this is, to the best of our knowledge, the only positive result that exists that does not rely on the number of agents or makespan as parameters. In other words, the algorithm we construct remains useful even in instances where we have a large enough number of agents and/or long enough schedules that would make any previous exact algorithm infeasible in practice.

\begin{restatable}{theorem}{thmFptDC}\label{thm:fpt-distance-to-clique}
\MAPF is in FPT parameterized by the distance to clique.
\end{restatable}

\section{Preliminaries}
We follow standard graph-theoretic notation~\cite{D12}.
For~\(a, b \in \mathbb{N}\), let~\([a, b] = \{c \in \mathbb{Z} \mid a \le c \le b\}\) and~\([a] = [1, a]\).
Let $G=(V,E)$ be a graph.
For a subset of vertices~\(U \subseteq V\) and a vertex $u \in V$ we denote by $N_{U}(u)$ the \textit{neighborhood of $u$ in $U$}, that is, the set of vertices of~\(U\) that are adjacent to~\(u\).
By~\(N_U[u]\) we mean~\(N_U(u) \cup \{u\}\).
As a short-hand, we write~\(N_G = N_{V(G)}\).
When clear from the context, the subscripts will be omitted.
Finally, for any $S\subseteq V$, let $G[S]$ be the subgraph of $G$ that is induced by~$S$, i.e., the graph that remains after deleting the vertices of $V\setminus S$ from $G$ (along with their incident edges).
For an~\(m\)\=/dimensional integer vector~\(\vec{u} \in \mathbb{N}^m\), we write~\(\|\vec{u}\|_p = \big(\sum_{i = 1}^m u_i^p\big)^{1/p}\), i.e.,~its~\(\ell_p\)\=/norm.
By default, \(\|\cdot\|\)~denotes the~\(\ell_1\)\=/norm.
For a function~\(f\colon A \to B\) and~\(C \subseteq A\) we write~\(f(C) = \{f(c) \mid c \in C\}\).

Formally, the input of the \MAPF (\MAPFShort) problem consists of a graph $G=(V,E)$, a set of agents $A$, two functions $s_0\colon A \rightarrow V$ and $t\colon A \rightarrow V$ and a positive integer~$\ell$. For any pair $a , b \in A$ where $a \neq b$, we have that $s_0(a) \neq s_0(b)$ and $t(a) \neq t(b)$.
Initially, each agent $a \in A$ is placed on the vertex $s_0(a)$. At specific times, called \emph{turns} or \emph{rounds} (we will use both interchangeably), the agents are allowed to move to a neighboring vertex, but are not obliged to do so. The agents can make at most one move per turn and each vertex can host at most one agent at a given turn. The position of the agents in the end of turn $i$ (after the agents have moved) is given by an injective function $s_i\colon A \rightarrow V$.

We say that a schedule $s_1,\ldots,s_{m}$ is a \emph{feasible solution} of an instance $\mathcal{I}=\langle G, A, s_0, t, \ell\rangle$ of \MAPFShort if:
\begin{enumerate}
 \item $s_{i}(a) \in N[s_{i - 1}(a)]$ for all $a \in A$ and every $i\in[m]$,
 \item $s_i(a) \neq s_i(b)$ for all $i \in [m]$ and $a \neq b \in A$, and
 \item $s_m = t$.
\end{enumerate}

A feasible solution $s_1,\ldots,s_{m}$ has \emph{makespan}~$m$. A feasible solution of \textit{minimum} makespan will be called \textit{optimum}. Our goal is to decide if there exists a feasible solution of makespan~$m \leq \ell$.

It is worth mentioning here that there are two principal variations of \MAPFShort. The first, more ``generous'' version, allows the agents to share edges if the wish to do so. That is, in this version, two agents $a$ and $b$ such that $s_i(a)=u,s_i(b)=v$, are allowed to move so that $s_{i+1}(a)=v,s_{i+1}(b)=u$ if $u$ and $v$ are connected through a unique edge. The second, more ``restrictive'', version does not allow this behavior. In this work we focus on the restrictive version. 

\paragraph*{Parameterized Complexity.}
Parameterized complexity is a field in algorithm design that takes into consideration additional measures to determine the time complexity. In some sense, it is about the multidimensional analysis of the time complexity of an algorithm, each dimension being dedicated to its own parameter.
Formally, a parameterized problem is a set of instances $(x,k) \in \Sigma^* \times \mathbb{N}$; $k$ is referred to as the \textit{parameter}.
The goal in this paradigm is to design \emph{Fixed-Parameter Tractable} (FPT) algorithms, i.e., an algorithm that solves the problem $f(k)|x|^{O(1)}$ time for any arbitrary computable function $f\colon \mathbb{N}\to\mathbb{N}$. We say that a problem \textit{is in FPT} if it admits an FPT algorithm. We refer the interested reader to classical monographs~\cite{CyganFKLMPPS15,downey2012parameterized} for a more comprehensive introduction to this topic.

\paragraph*{Structural Parameters.}
There are three structural parameters that we consider in this work. Let $G=(V,E)$ be a graph. 

First we have the \emph{vertex cover number} of~$G$. A set $S\subseteq V$ is called a \emph{vertex cover} of~$G$ if for every edge $uv\in E$ we have either $u\in S$ or $v\in S$ (or both). The vertex cover number of $G$ is the size of a minimum vertex cover of~$G$. This is a parameter that can be computed in FPT time, e.g., by~\cite{ChenKX06}.

Next, we will consider the \emph{distance to clique} of $G$. This is defined as the size of a minimum set $S\subseteq V$ such that $G[V\setminus S]$ is a clique (a graph where each vertex is adjacent to all other vertices of the graph); the size of such a set is denoted by $\dc(G)$. The distance to clique of $G$ can be computed in FPT time by computing the vertex cover number of the \textit{complement of} $G$. Recall that the complement of a graph $G$ is the graph $G'=(V,E')$, where $uv\in E'$ if and only if $uv\not\in E$.

Finally, the \emph{max leaf number} of $G$ is defined as the maximum number of leaves (i.e., vertices with only one neighbor) in a \textit{spanning tree} of $G$.
Recall that $T$ is a spanning tree of $G$ if it is a subgraph of $G$ with the same vertex set same as $G$ and~$T$ is a tree (i.e., it is conncted and contains no cycles).
The max leaf number of a given graph can also be computed in FPT time~\cite{FellowsL92}.

\section{Trees With Few Internal Vertices or Leaves}

This section is dedicated to the main infeasibility results we provide. 

\subsection{Trees With Bounded Vertex Cover}

This first subsection is dedicated to the proof of Theorem~\ref{thm:vc-hardness:mapf-vc-hard}. This proof is done by a 
reduction from the \PN{3-Partition} problem, which is known to be \CC{NP}\=/hard~\cite{3-partition-np-hard}:
as input we receive~\(3n\) integers~\(\vec{\beta} = (\beta_1, \ldots, \beta_{3n})\) where~\({\|\vec{\beta}\| = n\varphi}\) for some~\(\varphi \in \N\).
For a subset of indices~\(I \subseteq \Integers{3n}\) we denote~\(\beta(I) = \sum_{i \in I} \beta_i\).
The goal is to partition~\(\Integers{3n}\) into~\(n\) disjoint sets~\(\sigma_1, \ldots, \sigma_n\) of size~3 each so that~\(\beta(\sigma_i) = \varphi\) for every~\({i \in \Integers{n}}\).
Instead of saying that~\(\vec{\beta}\) is a yes-instance we will sometimes say that~\(\vec\beta\) \emph{has a 3\=/partition}.

By adding~\(2\varphi\) to every~\(\beta_i\) we can assume that~\(\varphi/4 < \beta_i < \varphi/2\) for every~\(i \in \Integers{3n}\).
We also set each~\(\beta_i\) to~\(6\beta_i\).
The reason to do this is the following claim.

\begin{claim}
  \label{claim:vc-hardness:3-partition-multiple}
  Suppose each~\(\beta_i\) is a multiple of~6, and that there exist~\(\beta_j, \beta_k, \beta_\ell \in \vec{\beta}\) with~\(i, j, k\) distinct that sum up to~\(\varphi\).
  Then no three distinct elements of~\(\vec{\beta}\) can sum up to a value in~\(\Integers{\varphi - 5, \varphi - 1} \cup \Integers{\varphi + 1, \varphi + 5}\).
\end{claim}
\begin{proofclaim}
  If the~\(\beta_i\)'s are multiples of~6, then the same holds for~\(\beta_j + \beta_k + \beta_\ell = \varphi\).
  Due to the~\(\beta_i\)'s being multiples of~6, the nearest possible value to~\(\varphi\) to which any three integers of~\(\vec\beta\) could sum up to is either~\(\varphi - 6\) or~\(\varphi + 6\).
\end{proofclaim}

We show a reduction from a \PN{3-Partition} instance~$\vec{\beta}$ to a \MAPFShort instance~\(\langle G, A, s_0, t, \ell\rangle\).
Before we give an intuition behind the reduction, as well as the formal proof, we first describe the construction itself.
Then we show that~\(\vec\beta\) has a 3\=/partition if and only if there exists a certain schedule with makespan~\(\ell = n\varphi + 3n\) (Lemma~\ref{lemma:vc-hardness:forward}).

In the reduction, we start with a tree~\(T\) (see Fig.~\ref{fig:vc-hardness:initial-graph}) with nine vertices~\(U = \{u_1, u_2, u_3, u_4, u_5, u_6, u_7, u_8, u_9\}\) and eight edges~\(u_1u_2, u_2u_3, u_3u_4, u_4u_5, u_2u_6, u_2u_8, u_2u_9, u_4u_7\), an empty agent set, empty functions \(s_0\) and~\(t\), and~\(\ell = n\varphi + 3n\).
We will need a few gadgets in the reduction that we describe now.
All of the gadgets are formed by a set of leaves with some agents with specified movements that we attach to vertices in the initial tree~$T$.

\begin{figure}[!t]
\centering

\begin{tikzpicture}[scale=0.6, inner sep=0.7mm]
    \node[draw, circle, line width=1pt, fill=white](u8) at (0,0)[label=above: { $u_8$}] {};
    \node[draw, circle, line width=1pt, fill=white](u1) at (0,3)[label=above: { $u_1$}] {};
    \node[draw, circle, line width=1pt, fill=white](u2) at (3,3)[label=below: { $u_2$}] {};
    \node[draw, circle, line width=1pt, fill=white](u3) at (6,3)[label=above: { $u_3$}] {};
    \node[draw, circle, line width=1pt, fill=white](u4) at (9,3)[label=above: { $u_4$}] {};
    \node[draw, circle, line width=1pt, fill=white](u5) at (12,3)[label=above: { $u_5$}] {};
    \node[draw, circle, line width=1pt, fill=white](u7) at (9,0)[label=right: { $u_7$}] {};
    \node[draw, circle, line width=1pt, fill=white](u6) at (0,6)[label=below: { $u_6$}] {};
    \node[draw, circle, line width=1pt, fill=white](u9) at (6,6)[label=below: { $u_9$}] {};
    
    \draw[-, line width=1pt]  (u8) -- (u2);
    \draw[-, line width=1pt]  (u1) -- (u2);
    \draw[-, line width=1pt]  (u2) -- (u3);
    \draw[-, line width=1pt]  (u3) -- (u4);
    \draw[-, line width=1pt]  (u4) -- (u5);
    \draw[-, line width=1pt]  (u6) -- (u2);
    \draw[-, line width=1pt]  (u9) -- (u2);
    \draw[-, line width=1pt]  (u7) -- (u4);

\end{tikzpicture}
  \caption{The initial tree~\(T\).}
  \label{fig:vc-hardness:initial-graph}
\end{figure}

\paragraph{Gadgets}
The first gadget is \emph{attaching a red edge of size~\(b\)} to a vertex~\(u \in U\) for any~\(b \in \N\).
See Fig.~\ref{fig:vc-hardness:red-edge} for an illustration.
Whenever we attach a red edge of size~\(b\), we add~\(b\)~new vertices~\(v_1, \ldots, v_b\) to~\(V\), we connect each vertex~\(v_i\) to~\(u \in U\) for~\(i \in \Integers{b}\), we add~\(b\)~new agents~\(a_1, \ldots, a_b\), for each~\(i \in [b]\) we set~\(s_0(a_i) = v_i\), for each~\(j \in [b - 1]\) we set~\(t(a_j) = v_{j + 1}\) and~\(t(a_b) = v_1\).
Throughout the reduction, we will attach multiple red edges to vertices of~\(T\) so we stress that the added vertices~\(v_1, \ldots, v_b\) and agents~\(a_1, \ldots, a_b\) are unique to each red edge.
Slightly overloading notation, if~\(D\) is a red edge, then~\(A(D)\) denotes the set of agents introduced by the red edge~\(D\).

\begin{figure}
  \centering
  \begin{tikzpicture}[scale=0.6, inner sep=0.7mm]
    \node[draw, circle, line width=1pt, fill=white](u) at (0,3)[label=above: {$u$}] {};
    \node[](k) at (0,0)[] {\textcolor{red}{$b$}};
    
    \draw[-, line width=1.5pt, color=red]  (u) -- (k);

    \draw[->, line width=1pt] (1,1.5)--(3,1.5);

    \node[draw, circle, line width=1pt, fill=white](u') at (6,3)[label=above: {$u$}] {};
    \node[draw, circle, line width=1pt, fill=white](v1) at (3,0)[] {};
    \node[draw, circle, line width=1pt, fill=white](v2) at (4.5,0)[] {};
    \node[draw, circle, line width=1pt, fill=white](vk) at (9,0)[] {};

    \draw[-, line width=1pt]  (u') -- (v1);
    \draw[-, line width=1pt]  (u') -- (v2);
    \draw[-, line width=1pt]  (u') -- (vk);
    \path (v2) -- (vk) node [black, font=\Large, midway, sloped] {$\dots$};

    \draw[->, line width=1pt, color=red] (v1) to[bend left=45] (v2);
    \draw[->, line width=1pt, color=red] (v2) to[bend left=45] (6,0);
    \draw[->, line width=1pt, color=red] (7.5,0) to[bend left=45] (vk);
    \draw[->, line width=1pt, color=red] (vk) to[bend left=35] (v1);

    \draw [decorate,decoration={brace, amplitude=10pt}] (9,-1) -- (3,-1) node [black,midway,yshift=-0.7cm] {$b$ vertices};
    
  \end{tikzpicture}
  \caption{%
    Attaching a red edge of size~\(b\) to vertex~\(u\).
    The tails and heads of the red arrows indicate the starting and target vertices, respectively, of each agent.
    Each one of the $k$ leaves is the starting and target position of exactly one agent.
  }
  \label{fig:vc-hardness:red-edge}
\end{figure}

Building on red edges, the next gadget we introduce is \emph{attaching a red star with edge sizes~\(\vec{z} \in \N^b\)}, where~\(b \in \N\), to a vertex~\(u \in U\).
See Fig.~\ref{fig:vc-hardness:red-star} for an illustration.
Whenever we attach a red star~\(R\) with edge sizes~\(\vec{z} \in \N^b\) to a vertex~\(u \in U\), for each~\(i \in [b]\) we attach a red edge of size~\(z_i\) to~\(u\).
Additionally, let $a$ be the first agent introduced by adding a red edge of size~\(z_1\), and let $v_1$ be the initial vertex of $a$. We modify this by setting $s_0(a)=u$ and $t(a)=v_1$.
For each~\(i \in [b]\) we denote by~\(A(R_i)\) the set of agents introduced by the red edge whose size corresponds to~\(z_i\).
We denote~\(A(R) = \cup_{i = 1}^b A(R_i)\) where~\(b\) is the number of red edges added by the red star~\(R\).
By \emph{an agent group of~\(A(R)\)} we mean an element of the set~\(\{A(R_i)\}_{i = 1}^b\). 

\begin{figure}
  \centering
  \begin{tikzpicture}[scale=0.6, inner sep=0.7mm]
    \node[draw, circle, line width=1pt, fill=white](u) at (2,4)[label=above: {$u$}] {};
    \node[](k2) at (0,0)[] {\textcolor{red}{$2$}};
    \node[](k22) at (2,0)[] {\textcolor{red}{$2$}};
    \node[](k3) at (4,0)[] {\textcolor{red}{$3$}};
    
    \draw[-, line width=1.5pt, color=red]  (u) -- (k2);
    \draw[-, line width=1.5pt, color=red]  (u) -- (k22);
    \draw[-, line width=1.5pt, color=red]  (u) -- (k3);

    \draw[->, line width=1pt] (5,2)--(6,2);

    \node[draw, circle, line width=1pt, fill=white](u') at (10,4)[label=above: {$u$}] {};
    \node[draw, circle, line width=1pt, fill=white](v1) at (6,0)[] {};
    \node[draw, circle, line width=1pt, fill=white](v2) at (7,0)[] {};
    \node[draw, circle, line width=1pt, fill=white](v3) at (9,0)[] {};
    \node[draw, circle, line width=1pt, fill=white](v4) at (10,0)[] {};
    \node[draw, circle, line width=1pt, fill=white](v5) at (12,0)[] {};
    \node[draw, circle, line width=1pt, fill=white](v6) at (13,0)[] {};
    \node[draw, circle, line width=1pt, fill=white](v7) at (14,0)[] {};

    \draw[-, line width=1pt]  (u') -- (v1);
    \draw[-, line width=1pt]  (u') -- (v2);
    \draw[-, line width=1pt]  (u') -- (v3);
    \draw[-, line width=1pt]  (u') -- (v4);
    \draw[-, line width=1pt]  (u') -- (v5);
    \draw[-, line width=1pt]  (u') -- (v6);
    \draw[-, line width=1pt]  (u') -- (v7);

    \draw[->, line width=1pt, color=red] (v1) to[bend left=35] (v2);
    \draw[->, line width=1pt, color=red] (u') to[bend right=25] (v1);
    \draw[->, line width=1pt, color=red] (v3) to[bend left=35] (v4);
    \draw[->, line width=1pt, color=red] (v4) to[bend left=35] (v3);
    \draw[->, line width=1pt, color=red] (v5) to[bend left=35] (v6);
    \draw[->, line width=1pt, color=red] (v6) to[bend left=35] (v7);
    \draw[->, line width=1pt, color=red] (v7) to[bend left=35] (v5);
    
  \end{tikzpicture}
  \caption{%
    A red star with edge sizes~\((2, 2, 3)\) attached to vertex~\(u\).
    The tails and heads of the red arrows indicate the starting and target vertices, respectively, of each agent.
  }
  \label{fig:vc-hardness:red-star}
\end{figure}

Finally, we have the gadget of \emph{attaching a bow tie of size~\(b\)} to a vertex~\(u \in U\) for any~\(b \in \N\).
See Fig.~\ref{fig:vc-hardness:bow-tie} for an illustration.
Whenever we attach a bow tie of size~\(b\) to~\(u \in U\), we add~\(2b\)~new vertices~\(w_1, \ldots, w_{2k}\) to~\(V\), we connect each vertex~\(w_i\) to~\(u\) for~\(i \in \Integers{2b}\), we add~\(b\)~new agents~\(a_1, \ldots, a_b\) to~\(A\) (which are distinct from the agents from red edges but we use the same variable for simplicity), and for every~\(i \in \Integers{b}\) we set~\(s_0(a_i) = w_i\) and~\(t(a_i) = w_{b + i}\).
If~\(u\) is a vertex to which a bow tie was attached, then we denote by~\(A_T(u)\) the set of agents introduced by that bow tie.

\begin{figure}
  \centering
  \begin{tikzpicture}[scale=0.6, inner sep=0.7mm]
    \node[draw, circle, line width=1pt, fill=white](u) at (0,0)[label=above: {$u$}] {};
    \draw[-,line width=0.5pt,black,fill=lightgray] (u) --+(-2,-1) --+(-2,1) -- (u);
    \draw[-,line width=0.5pt,black,fill=lightgray] (u) --+(2,-1) --+(2,1) -- (u);

    \draw[->, line width=1pt] (3,0)--(4,0);

    \node[draw, circle, line width=1pt, fill=white](w1) at (5,1.5)[] {};
    \node[draw, circle, line width=1pt, fill=white](w2) at (5,0.5)[] {};
    \node[draw, circle, line width=1pt, fill=white](w3) at (5,-1.5)[] {};

    \path (w2) -- (w3) node [black, font=\Large, midway, sloped] {$\dots$};

    \node[draw, circle, line width=1pt, fill=white](u') at (7,0)[label=above: {$u$}] {};

    \node[draw, circle, line width=1pt, fill=white](w4) at (9,1.5)[] {};
    \node[draw, circle, line width=1pt, fill=white](w5) at (9,0.5)[] {};
    \node[draw, circle, line width=1pt, fill=white](w6) at (9,-1.5)[] {};

    \path (w5) -- (w6) node [black, font=\Large, midway, sloped] {$\dots$};

    \draw [decorate,decoration={brace, amplitude=10pt}] (9.5,1.5) -- (9.5,-1.5) node [black,midway,xshift=1.5cm] {$b$ vertices};

    \draw[-, line width=1pt]  (u') -- (w1);
    \draw[-, line width=1pt]  (u') -- (w2);
    \draw[-, line width=1pt]  (u') -- (w3);
    \draw[-, line width=1pt]  (u') -- (w4);
    \draw[-, line width=1pt]  (u') -- (w5);
    \draw[-, line width=1pt]  (u') -- (w6);

    \draw[->, line width=1pt, color=red] (w1) to[bend left=35] (w4);
    \draw[->, line width=1pt, color=red] (w2) to[bend left=35] (w5);
    \draw[->, line width=1pt, color=red] (w3) to[bend right=35] (w6);
  \end{tikzpicture}
  \caption{%
    Attaching a bow tie of size~\(b\) to a vertex~\(u\).
    The tails and heads of the red arrows indicate the starting and target vertices, respectively, of each agent.
  }
  \label{fig:vc-hardness:bow-tie}
\end{figure}

\paragraph{Construction}
With the above gadgets, we now proceed to describe the instance we construct, illustrated in Fig.~\ref{fig:vc-hardness:complete-reduction}.
First, we attach a red star with edge sizes~\(\vec{\beta}\) to~\(u_1\) and we name it~\(R^1\).
Let~\(\vec{z} \in \N^n\) be the vector where~\(z_1 = \varphi\), \(z_n = \varphi + 4\), and ~\(z_i = \varphi + 2\) for each~\(i \in [2, n - 1]\), then attach a red star with edge sizes~\(\vec{z}\) to~\(u_7\) and we name the red star~\(R^7\).
We add~\(2n - 2\)~new vertices~\(x_1, \ldots, x_{2n - 2}\), denoted by~\(V_X^-\) and connect them to~\(u_5\), and another~\(2n - 2\)~new vertices~\(x_{2n - 1}, \ldots, x_{4n - 4}\), denoted by~\(V_X^+\) and connect them to~\(u_6\).
For each~\(i \in [2n - 2]\) we add a new agent~\(a_i\) and set~\(s_0(a_i) = x_i\), and~\(t(a_i) = x_{2n - 2 + i}\), and we denote these agents by~\(A^X\).
We add~\(4n\)~new vertices~\(y_1, \ldots, y_{4n}\), denoted by~\(V_Y^-\) and connect them to~\(u_8\), and another~\(4n\)~new vertices~\(y_{4n + 1}, \ldots, y_{8n}\), denoted by~\(V_Y^+\), and connect them to~\(u_9\).
For each~\(i \in [4n]\) we add a new agent~\(a_i\) and set~\(s_0(a_i) = y_i\) and~\(t(a_i) = y_{4n + i}\), and we denote these agents by~\(A^Y\).
We attach a bow tie of size~\(\ell - 1 - 4n\) to vertices~\(u_8\) and~\(u_9\) and a bow tie of size~\(\ell - 1 - (2n - 2)\) to vertices~\(u_3, u_5\), and~\(u_6\).
This concludes the construction.
We denote the graph constructed this way by $G$ and the resulting instance by~\(\mathcal{I}\).

Immediately, we can prove that~\(G\) satisfies the structural properties we desire.

\begin{observation}
  \label{lemma:vc-hardness:small-vc}
  The graph~\(G\) produced by the reduction has vertex cover number~7 and nine internal vertices.
\end{observation}
  In fact, the vertices of the set~\(U\) are the only internal vertices of~\(G\).
  Furthermore, the set $\{u_1, u_3, u_5, u_6, u_7, u_8, u_9\}$ is a vertex cover of $G$.

\begin{figure}
  \centering
  \begin{tikzpicture}[scale=0.7, inner sep=0.7mm]
    \filldraw[fill=red!10, draw=black] plot[smooth cycle, tension=0.9] coordinates {
        (1,3) (-3,5) (-3,1) 
    };
    \node[] at (-3,4) {\Large $R^1$};

    \filldraw[fill=green!10, draw=black] plot[smooth cycle, tension=0.9] coordinates {
        (13.5,5.5) (13.5,0.5) (16,1.5) (16,4.5) 
    };
    \node[] at (16,3) {\Large $V_X^-$};

    \filldraw[fill=green!10, draw=black] plot[smooth cycle, tension=0.9] coordinates {
        (-2.5,7.5) (-1.5,10) (1,10) (2,7.5) 
    };
    \node[] at (-0.5,9.5) {\Large $V_X^+$};

    \filldraw[fill=yellow!10, draw=black] plot[smooth cycle, tension=0.9] coordinates {
        (3.5,7.5) (4.5,10) (6.5,10) (8,7.5) 
    };
    \node[] at (5.5,9.5) {\Large $V_Y^+$};

    \filldraw[fill=yellow!10, draw=black] plot[smooth cycle, tension=0.9] coordinates {
        (-2.5,-1.5) (-1.5,-3.5) (1,-3.5) (2,-1.5)
    };
    \node[] at (-0.5,-3.5) {\Large $V_Y^-$};
    
    \filldraw[fill=red!10, draw=black] plot[smooth cycle, tension=0.9] coordinates {
        (9,1) (5,-3) (13,-3) 
    };
    \node[] at (12,-3) {\Large $R^7$};
    
    \node[draw, circle, line width=1pt, fill=white](u8) at (0,0)[label=above: { $u_8$}] {};
    \node[draw, circle, line width=1pt, fill=white](u1) at (0,3)[label=above: { $u_1$}] {};
    \node[draw, circle, line width=1pt, fill=white](u2) at (3,3)[label=below: { $u_2$}] {};
    \node[draw, circle, line width=1pt, fill=white](u3) at (6,3)[label=above right: { $u_3$}] {};
    \node[draw, circle, line width=1pt, fill=white](u4) at (9,3)[label=above: { $u_4$}] {};
    \node[draw, circle, line width=1pt, fill=white](u5) at (12,3)[label=above left: { $u_5$}] {};
    \node[draw, circle, line width=1pt, fill=white](u7) at (9,0)[label=right: { $u_7$}] {};
    \node[draw, circle, line width=1pt, fill=white](u6) at (0,6)[label=below: { $u_6$}] {};
    \node[draw, circle, line width=1pt, fill=white](u9) at (6,6)[label=below: { $u_9$}] {};
    
    \draw[-, line width=1pt]  (u8) -- (u2);
    \draw[-, line width=1pt]  (u1) -- (u2);
    \draw[-, line width=1pt]  (u2) -- (u3);
    \draw[-, line width=1pt]  (u3) -- (u4);
    \draw[-, line width=1pt]  (u4) -- (u5);
    \draw[-, line width=1pt]  (u6) -- (u2);
    \draw[-, line width=1pt]  (u9) -- (u2);
    \draw[-, line width=1pt]  (u7) -- (u4);

    \node[](b1) at (-2,4.5)[] {\textcolor{red}{$\beta_1$}};
    \node[](b2) at (-2,3.5)[] {\textcolor{red}{$\beta_2$}};
    \node[](b3n) at (-2,1.5)[] {\textcolor{red}{$\beta_{3n}$}};
    
    \draw[-, line width=1pt, color=red]  (u1) -- (b1);
    \draw[-, line width=1pt, color=red]  (u1) -- (b2);
    \draw[-, line width=1pt, color=red]  (u1) -- (b3n);

    \path (b2) -- (b3n) node [black, font=\Large, midway, sloped] {$\dots$};

    \draw[-,line width=0.5pt,black,fill=blue!30] (u3) --+(-0.5,1) --+(0.5,1) -- (u3);
    \draw[-,line width=0.5pt,black,fill=blue!30] (u3) --+(-0.5,-1) --+(0.5,-1) -- (u3);

    \draw[-,line width=0.5pt,black,fill=blue!30] (u5) --+(-0.5,1) --+(0.5,1) -- (u5);
    \draw[-,line width=0.5pt,black,fill=blue!30] (u5) --+(-0.5,-1) --+(0.5,-1) -- (u5);
    
    \node[draw, circle, line width=1pt, fill=white](x1) at (14,4.5)[label=right: { $x_1$}] {};
    \node[draw, circle, line width=1pt, fill=white](x2) at (14,3.5)[label=right: { $x_2$}] {};
    \node[draw, circle, line width=1pt, fill=white](x2n2) at (14,1.5)[label=right: { $x_{2n-2}$}] {};

    \path (x2) -- (x2n2) node [black, font=\Large, midway, sloped] {$\dots$};

    \draw[-, line width=1pt]  (u5) -- (x1);
    \draw[-, line width=1pt]  (u5) -- (x2);
    \draw[-, line width=1pt]  (u5) -- (x2n2);

    \draw[-,line width=0.5pt,black,fill=blue!30] (u6) --+(-1,-0.5) --+(-1,0.5) -- (u6);
    \draw[-,line width=0.5pt,black,fill=blue!30] (u6) --+(1,-0.5) --+(1,0.5) -- (u6);
    
    \node[draw, circle, line width=1pt, fill=white](x2n1) at (-2,8)[label=above: { $x_{2n-1}$}] {};
    \node[draw, circle, line width=1pt, fill=white](x2n) at (-0.5,8)[label=above: { $x_{2n}$}] {};
    \node[draw, circle, line width=1pt, fill=white](x4n4) at (1.5,8)[label=above: { $x_{4n-4}$}] {};

    \path (x2n) -- (x4n4) node [black, font=\Large, midway, sloped] {$\dots$};

    \draw[-, line width=1pt]  (u6) -- (x2n1);
    \draw[-, line width=1pt]  (u6) -- (x2n);
    \draw[-, line width=1pt]  (u6) -- (x4n4);

    \draw[-,line width=0.5pt,black,fill=red!30] (u9) --+(-1,-0.5) --+(-1,0.5) -- (u9);
    \draw[-,line width=0.5pt,black,fill=red!30] (u9) --+(1,-0.5) --+(1,0.5) -- (u9);
    
    \node[draw, circle, line width=1pt, fill=white](y4n1) at (4,8)[label=above: { $y_{4n+1}$}] {};
    \node[draw, circle, line width=1pt, fill=white](y4n2) at (5.5,8)[label=above: { $y_{4n+2}$}] {};
    \node[draw, circle, line width=1pt, fill=white](y8n) at (7.5,8)[label=above: { $y_{8n}$}] {};

    \path (y4n2) -- (y8n) node [black, font=\Large, midway, sloped] {$\dots$};

    \draw[-, line width=1pt]  (u9) -- (y4n1);
    \draw[-, line width=1pt]  (u9) -- (y4n2);
    \draw[-, line width=1pt]  (u9) -- (y8n);

    \node[](f1) at (6,-2)[] {\textcolor{red}{$\varphi$}};
    \node[](f2) at (7.5,-2)[] {\textcolor{red}{$\varphi+2$}};
    \node[](f3) at (10.5,-2)[] {\textcolor{red}{$\varphi+2$}};
    \node[](f4) at (12,-2)[] {\textcolor{red}{$\varphi+4$}};
    
    \draw[-, line width=1pt, color=red]  (u7) -- (f1);
    \draw[-, line width=1pt, color=red]  (u7) -- (f2);
    \draw[-, line width=1pt, color=red]  (u7) -- (f3);
    \draw[-, line width=1pt, color=red]  (u7) -- (f4);

    \path (f2) -- (f3) node [black, font=\Large, midway, sloped] {$\dots$};

    \draw [decorate,decoration={brace, amplitude=10pt}] (10.5,-2.5) -- (7.5,-2.5) node [black,midway,yshift=-0.5cm] {$n-2$ red edges};

    \draw[-,line width=0.5pt,black,fill=red!30] (u8) --+(-1,-0.5) --+(-1,0.5) -- (u8);
    \draw[-,line width=0.5pt,black,fill=red!30] (u8) --+(1,-0.5) --+(1,0.5) -- (u8);
    
    \node[draw, circle, line width=1pt, fill=white](y1) at (-2,-2)[label=below: { $y_1$}] {};
    \node[draw, circle, line width=1pt, fill=white](y2) at (-0.5,-2)[label=below: { $y_2$}] {};
    \node[draw, circle, line width=1pt, fill=white](y4n) at (1.5,-2)[label=below: { $y_{4n}$}] {};

    \path (y2) -- (y4n) node [black, font=\Large, midway, sloped] {$\dots$};

    \draw[-, line width=1pt]  (u8) -- (y1);
    \draw[-, line width=1pt]  (u8) -- (y2);
    \draw[-, line width=1pt]  (u8) -- (y4n);

\end{tikzpicture}
  \caption{%
    The construction of the graph $G$, described in the proof of Theorem~\ref{thm:vc-hardness:mapf-vc-hard}.
    Light blue bow ties (attached on the vertices $u_4, u_5$ and $u_6$) have size~\(\ell - 1 - (2n - 2)\).
    Light red bow ties (attached on the vertices $u_8$ and $u_9$) have size~\(\ell - 1 - 4n\).
  }
  \label{fig:vc-hardness:complete-reduction}
\end{figure}

\paragraph{High-level idea of the reduction.} 
On a high level, the reduction works as follows.
First, the bow ties play the role of gates, forcing the agents of $A^X$ and $A^Y$ to only move on the (shortest) path from their starting position to their terminal position, and without doing any ``back and forth'' once they begin traveling towards their destination.
Then, the choice of $\ell$ is such that, according to any optimal solution of $\mathcal{I}$, the agents of $R^1$ and $R^7$ must constantly move towards their terminal vertices.
Notice also that, in order for any agent that starts on a red edge to move towards their goal, they have to pass through either $u_1$ or $u_7$, according to which vertex their red edge is attached to.
Also, in order to have enough empty vertices so that the agents (that start on a red edge) can ``rotate'', one of them must occupy $u_2$ or $u_4$.
This forces the occupation of the vertices $u_2$ and $u_4$ by the agents of $R^1$ and $R^7$, respectively, during specific turns. 
This, essentially, sets a global clock according to which the vertices $u_2$ and~$u_4$ are unoccupied only during specific turns. 
Moreover, we are able to prove that in any optimal solution of~$\mathcal{I}$, once the first agent of any red edge from the $\beta$s ($\varphi$s resp.) starts moving towards their terminal, then all agents of the same red edge of the $\vec{\beta}$s ($\varphi$s resp.) have to reach their terminal before another red edge of the same type can begin to be treated.
This, combined with the global clock, ensures that the makespan of the schedule is exactly $\ell$ if and only if there are three red edges from the $\vec\beta$s that are dealt with every time a red edge from the $\varphi$s is completed, giving us the equivalence between the two problems. 

\medskip 

We are now ready to formally prove our result. 

\thmNpVc*

\begin{proof}
Let~\(\vec\beta\) be an instance of the \PN{3-Partition} problem and \(\mathcal{I}\) be the instance of \MAPF constructed as explained above. We will to show that~\(\vec\beta\) admits a~3\=/partition if and only if~\(\mathcal{I}\) admits a schedule of makespan at most~\(\ell\).

We say that a subset of agents~\(B \subseteq A\) is \emph{resolved in turn}~\(\tau\) if for each~\({b \in B}\) there exists~\(i \in [\tau]\) such that~\({s_i(b) = t(b)}\).
By the nature of our reduction, once an agent~\(a \in A\) reaches~\(t(a)\) in some turn, they will remain there until the end of the schedule.
Instead of providing the schedule all at once, we shall define some \emph{recipes} for various subsets of agents which can then be composed to create a schedule.
\emph{Applying a recipe} with \emph{duration~\(d\)} in turn~\(\tau\) on a subset of agents~\(B \subseteq A\) should be viewed as specifying a schedule for agents in~\(B\) in turns~\([\tau, \tau + d]\).
Each recipe is such that once we apply to a set of agents~\(B \subseteq A\), then~\(B\) will be resolved at the end of the recipe.
Thus when we say that we \emph{resolve~\(B\)}, we mean that we apply a recipe to~\(B\) so that all agents of~\(B\) are resolved by the end of the recipe.

\begin{recipe}[Agents of a red edge.]
  \label{recipe:vc-hardness:red-edge}
  Let~\(B = \{a_1, \ldots, a_c\} \subseteq A\) be a subset of agents that were introduced by attaching a red edge of size~\(c\) to a leaf~\(u\) of $T$ and let~\(w\) be the unique vertex in~\(N_T(u)\). We may assume that $c\geq 2$, since the case for $c=1$ is trivial.
  Let~\(\tau\) be the turn in which the recipe is applied.
  See Fig.~\ref{fig:opt-strategy-single-red-edge} for an illustration of this recipe.
  In turn~\(\tau + 1\) we move~\(a_1\) to~\(u\).
  In turn~\(\tau + 2\) we move~\(a_1\) to~\(w\) and~\(a_c\) to~\(u\).
  For~\(i \in [2, c]\) in decreasing order of~\(i\) we move agent~\(a_i\) from~\(u\) to~\(t(a_i)\) and simultaneously move agent~\(a_{i - 1}\) from~\(s_0(a_{i - 1})\) to~\(u\) in turn~\(\tau + (c - i) + 3\).
  In turn~\(\tau + c + 1\) we move~\(a_2\) from~\(u\) to~\(t(a_2)\) and simultaneously~\(a_1\) to~\(u\).
  Finally, in turn~\(\tau + c + 2\) we move~\(a_1\) to~\(t(a_1)\).
  This recipe has duration~\(c + 2\), and occupies~\(w\) in turns~\([\tau + 2, \tau + c]\).
\end{recipe}

With Recipe~\ref{recipe:vc-hardness:red-edge}, we are ready to specify how to resolve agents of a red star.

\begin{figure}
  \centering
  \subfloat[Turn $\tau$.]{
\begin{tikzpicture}[scale=0.6, inner sep=0.7mm]
    \node[draw, circle, line width=1pt, fill=white](u) at (3,3)[label=above: {$u$}] {\textcolor{white}{$a_4$}};
    \node[draw, circle, line width=1pt, fill=white](w) at (5,3)[label=above: {$w$}] {\textcolor{white}{$a_4$}};
    \node[draw, circle, line width=1pt, fill=white](v1) at (0,0)[label=below: {$t_4$}] {$a_1$};
    \node[draw, circle, line width=1pt, fill=white](v2) at (2,0)[label=below: {$t_1$}] {$a_2$};
    \node[draw, circle, line width=1pt, fill=white](v3) at (4,0)[label=below: {$t_2$}] {$a_3$};
    \node[draw, circle, line width=1pt, fill=white](v4) at (6,0)[label=below: {$t_3$}] {$a_4$};

    \draw[-, line width=1pt]  (u) -- (w);
    \draw[-, line width=1pt]  (u) -- (v1);
    \draw[-, line width=1pt]  (u) -- (v2);
    \draw[-, line width=1pt]  (u) -- (v3);
    \draw[-, line width=1pt]  (u) -- (v4);

\end{tikzpicture}
}

\subfloat[Turn $\tau+1$.]{
\begin{tikzpicture}[scale=0.6, inner sep=0.7mm]
    \node[draw, circle, line width=1pt, fill=white](u) at (3,3)[label=above: {$u$}] {$a_1$};
    \node[draw, circle, line width=1pt, fill=white](w) at (5,3)[label=above: {$w$}] {\textcolor{white}{$a_4$}};
    \node[draw, circle, line width=1pt, fill=white](v1) at (0,0)[label=below: {$t_4$}] {\textcolor{white}{$a_4$}};
    \node[draw, circle, line width=1pt, fill=white](v2) at (2,0)[label=below: {$t_1$}] {$a_2$};
    \node[draw, circle, line width=1pt, fill=white](v3) at (4,0)[label=below: {$t_2$}] {$a_3$};
    \node[draw, circle, line width=1pt, fill=white](v4) at (6,0)[label=below: {$t_3$}] {$a_4$};

    \draw[-, line width=1pt]  (u) -- (w);
    \draw[-, line width=1pt]  (u) -- (v1);
    \draw[-, line width=1pt]  (u) -- (v2);
    \draw[-, line width=1pt]  (u) -- (v3);
    \draw[-, line width=1pt]  (u) -- (v4);

\end{tikzpicture}
}\hspace{10pt}
\subfloat[Turn $\tau+2$.]{
\begin{tikzpicture}[scale=0.6, inner sep=0.7mm]
    \node[draw, circle, line width=1pt, fill=white](u) at (3,3)[label=above: {$u$}] {$a_4$};
    \node[draw, circle, line width=1pt, fill=white](w) at (5,3)[label=above: {$w$}] {$a_1$};
    \node[draw, circle, line width=1pt, fill=white](v1) at (0,0)[label=below: {$t_4$}] {\textcolor{white}{$a_4$}};
    \node[draw, circle, line width=1pt, fill=white](v2) at (2,0)[label=below: {$t_1$}] {$a_2$};
    \node[draw, circle, line width=1pt, fill=white](v3) at (4,0)[label=below: {$t_2$}] {$a_3$};
    \node[draw, circle, line width=1pt, fill=white](v4) at (6,0)[label=below: {$t_3$}] {\textcolor{white}{$a_4$}};

    \draw[-, line width=1pt]  (u) -- (w);
    \draw[-, line width=1pt]  (u) -- (v1);
    \draw[-, line width=1pt]  (u) -- (v2);
    \draw[-, line width=1pt]  (u) -- (v3);
    \draw[-, line width=1pt]  (u) -- (v4);

\end{tikzpicture}
}\hspace{10pt}
\subfloat[Turn $\tau+3$.]{
\begin{tikzpicture}[scale=0.6, inner sep=0.7mm]
    \node[draw, circle, line width=1pt, fill=white](u) at (3,3)[label=above: {$u$}] {$a_3$};
    \node[draw, circle, line width=1pt, fill=white](w) at (5,3)[label=above: {$w$}] {$a_1$};
    \node[draw, circle, line width=1pt, fill=white](v1) at (0,0)[label=below: {$t_4$}] {$a_4$};
    \node[draw, circle, line width=1pt, fill=white](v2) at (2,0)[label=below: {$t_1$}] {$a_2$};
    \node[draw, circle, line width=1pt, fill=white](v3) at (4,0)[label=below: {$t_2$}] {\textcolor{white}{$a_4$}};
    \node[draw, circle, line width=1pt, fill=white](v4) at (6,0)[label=below: {$t_3$}] {\textcolor{white}{$a_4$}};

    \draw[-, line width=1pt]  (u) -- (w);
    \draw[-, line width=1pt]  (u) -- (v1);
    \draw[-, line width=1pt]  (u) -- (v2);
    \draw[-, line width=1pt]  (u) -- (v3);
    \draw[-, line width=1pt]  (u) -- (v4);

\end{tikzpicture}
}

\subfloat[Turn $\tau+4$.]{
\begin{tikzpicture}[scale=0.6, inner sep=0.7mm]
    \node[draw, circle, line width=1pt, fill=white](u) at (3,3)[label=above: {$u$}] {$a_2$};
    \node[draw, circle, line width=1pt, fill=white](w) at (5,3)[label=above: {$w$}] {$a_1$};
    \node[draw, circle, line width=1pt, fill=white](v1) at (0,0)[label=below: {$t_4$}] {$a_4$};
    \node[draw, circle, line width=1pt, fill=white](v2) at (2,0)[label=below: {$t_1$}] {\textcolor{white}{$a_4$}};
    \node[draw, circle, line width=1pt, fill=white](v3) at (4,0)[label=below: {$t_2$}] {\textcolor{white}{$a_4$}};
    \node[draw, circle, line width=1pt, fill=white](v4) at (6,0)[label=below: {$t_3$}] {$a_3$};

    \draw[-, line width=1pt]  (u) -- (w);
    \draw[-, line width=1pt]  (u) -- (v1);
    \draw[-, line width=1pt]  (u) -- (v2);
    \draw[-, line width=1pt]  (u) -- (v3);
    \draw[-, line width=1pt]  (u) -- (v4);

\end{tikzpicture}
}\hspace{10pt}
\subfloat[Turn $\tau+5$.]{
\begin{tikzpicture}[scale=0.6, inner sep=0.7mm]
    \node[draw, circle, line width=1pt, fill=white](u) at (3,3)[label=above: {$u$}] {$a_1$};
    \node[draw, circle, line width=1pt, fill=white](w) at (5,3)[label=above: {$w$}] {\textcolor{white}{$a_4$}};
    \node[draw, circle, line width=1pt, fill=white](v1) at (0,0)[label=below: {$t_4$}] {$a_4$};
    \node[draw, circle, line width=1pt, fill=white](v2) at (2,0)[label=below: {$t_1$}] {\textcolor{white}{$a_4$}};
    \node[draw, circle, line width=1pt, fill=white](v3) at (4,0)[label=below: {$t_2$}] {$a_2$};
    \node[draw, circle, line width=1pt, fill=white](v4) at (6,0)[label=below: {$t_3$}] {$a_3$};

    \draw[-, line width=1pt]  (u) -- (w);
    \draw[-, line width=1pt]  (u) -- (v1);
    \draw[-, line width=1pt]  (u) -- (v2);
    \draw[-, line width=1pt]  (u) -- (v3);
    \draw[-, line width=1pt]  (u) -- (v4);

\end{tikzpicture}
}\hspace{10pt}
\subfloat[Turn $\tau+6$.]{
\begin{tikzpicture}[scale=0.6, inner sep=0.7mm]
    \node[draw, circle, line width=1pt, fill=white](u) at (3,3)[label=above: {$u$}] {\textcolor{white}{$a_4$}};
    \node[draw, circle, line width=1pt, fill=white](w) at (5,3)[label=above: {$w$}] {\textcolor{white}{$a_4$}};
    \node[draw, circle, line width=1pt, fill=white](v1) at (0,0)[label=below: {$t_4$}] {$a_4$};
    \node[draw, circle, line width=1pt, fill=white](v2) at (2,0)[label=below: {$t_1$}] {$a_1$};
    \node[draw, circle, line width=1pt, fill=white](v3) at (4,0)[label=below: {$t_2$}] {$a_2$};
    \node[draw, circle, line width=1pt, fill=white](v4) at (6,0)[label=below: {$t_3$}] {$a_3$};

    \draw[-, line width=1pt]  (u) -- (w);
    \draw[-, line width=1pt]  (u) -- (v1);
    \draw[-, line width=1pt]  (u) -- (v2);
    \draw[-, line width=1pt]  (u) -- (v3);
    \draw[-, line width=1pt]  (u) -- (v4);

\end{tikzpicture}
}
  \caption{%
    A duration~6 recipe for a single red edge of size~4 attached to vertex~\(u\) whose sole neighbor is~\(w\).
    The initial positions of agents $a_1,a_2,a_3,a_4$ are as depicted in subfigure (a). The terminal position of agent $a_i$ is the vertex labeled by $t_i$, $i\in[4]$.
  }
  \label{fig:opt-strategy-single-red-edge}
\end{figure}

\begin{recipe}[Agents of a red star]
  \label{recipe:vc-hardness:red-star}
  Let~\(\vec{z} \in \mathbb{N}^b\) where~\(b \in \mathbb{N}\).
  To apply this recipe, we need to also specify a permutation~\(\pi\) of~\([b]\) where~\(\pi_1 = 1\).
  If we do not specify~\(\pi\), then it is the identity permutation by default.
  Let~\(u\) be a leaf of $T$ with a sole neighbor~\(w\) (in $T$).
  Suppose that we attach a red star with edge sizes~\(\vec{z}\) to~\(u\).
  For each~\(i \in [b]\), we denote the red edge corresponding to~\(z_i\) by~\(r_i\), the actual graph vertices added by~\(r_i\) by~\(V_i\) and the agents added by~\(r_i\) by~\(A_i\).
  To resolve~\(\bigcup_{i = 1}^b A_i\), the idea is to successively apply Recipe~\ref{recipe:vc-hardness:red-edge} on~\(A_{\pi_1}, A_{\pi_2}, \ldots, A_{\pi_b}\).
  But there are two optimizations that we can do.
  \begin{enumerate}
    \item The first agent~\(a \in A_{\pi_1 = 1}\) has~\(s_0(a) = u\) (by our construction), meaning that for~\(A_1\) we can skip the first turn of Recipe~\ref{recipe:vc-hardness:red-edge} as it is already ``performed by the construction''.
    \item In the last turn of Recipe~\ref{recipe:vc-hardness:red-edge} applied to some agent group~\(A_{\pi_i}\), we move an agent from~\(u\) to a vertex of~\(V_{\pi_i}\).
      In that same turn we can move the first agent of~\(A_{\pi_{i + 1}}\) to~\(u\) according to Recipe~\ref{recipe:vc-hardness:red-edge} applied to~\(A_{\pi_{i + 1}}\).
      Thus we can overlap the application of Recipe~\ref{recipe:vc-hardness:red-edge} between two successive agent sets~\(A_{\pi_i}\) and~\(A_{\pi_{i + 1}}\).
  \end{enumerate}
  In general, the duration of this recipe is~\(\|\vec{z}\| + b\).
  For every~\(i \in [b]\) let~\(P_i = \sum_{j = 1}^i \beta_{\pi_j}\), i.e.~the prefix sum of~\(\vec{z}\) in order of permutation~\(\pi\).
  For every~\(i \in [b]\) this recipe keeps vertex~\(w\) unoccupied in turns~\(P_i + 1\) and~\(P_i + 2\) by agents~\(\cup_{i = 1}^b A_i\).
\end{recipe}

In our reduction, we will apply recipes in the order of their number.
In particular, it means that it is legal for a Recipe~\(i\) to check what the action of any Recipe~\(j\) with number~\(j < i\) does and which vertices are occupied by the agents on which Recipe~\(j\) is being applied.
We now specify recipes for agents~\(A^X\) and~\(A^Y\).

\begin{recipe}[Agents~\(A^X\)]
  \label{recipe:vc-hardness:Ax}
  To apply this recipe in turn~\(\tau\), we are given an agent~\(a \in A^X\) such that~\(s_\tau(a) = s_0(a)\).
  If none of the agents of~\(A(R^7)\) occupy~\(u_4\) in turn~\(\tau + 2\) and none of the agents of~\(A(R^1)\) occupy~\(u_2\) in turn~\(\tau + 4\), then we move~\(a\) along the path~\(s_0(a), u_5, u_4, u_3, u_2, u_6, t(a)\) in turns~\([\tau, \tau + 6]\).
\end{recipe}

\begin{recipe}[Agents~\(A^Y\)]
  \label{recipe:vc-hardness:Ay}
  To apply this recipe in turn~\(\tau\), we are given an agent~\(a \in A^Y\) such that~\(s_\tau(a) = s_0(a)\).
  If none of agents  of~\(A(R^1)\) or~\(A^X\) occupy~\(u_2\) in turn~\(\tau + 2\), then we move~\(a\) along the path~\(s_0(a), u_8, u_2, u_9, t(a)\) in turns~\([\tau, \tau + 4]\).
\end{recipe}

Finally, we specify a recipe for bow ties.

\begin{recipe}[Agents of a bow tie]
  \label{recipe:vc-hardness:bowtie}
  Let~\(u\) be a vertex of $T$ to which a bow tie of size~\(c\) was attached, and let this recipe be applied in turn~\(\tau\).
  If no agent being resolved by the preceding Recipes occupies~\(u\) in turn~\(\tau + 1\), then select an arbitrary agent~\(a\) of the bow tie that is not resolved yet, and move it to~\(u\) in turn~\(\tau + 1\) and to~\(t(a)\) in turn~\(\tau + 2\).
\end{recipe}

\todo[inline]{DK: The recipes are great; it would be even better if we could have some example of how the recipes are interleaving in a figure.}

\paragraph{Correctness (``if'').}
\begin{lemma}
  \label{lemma:vc-hardness:forward}
  If~\(\vec\beta\) has a 3\=/partition, then~\(\mathcal{I}\) admits a schedule of makespan at most~\(\ell = n\varphi + 3n\).
\end{lemma}
\begin{proof}
  Let~\(\vec{\sigma} = (\sigma_1, \ldots, \sigma_n)\) be a 3\=/partition of~\(\vec\beta\).
  For convenience, we denote by~\(\vec{\iota}\) the vector~\(\vec{\sigma}\) without nesting, i.e.,~for every~\(i \in [3n]\), let \(\iota_i = (\sigma_{\lfloor i / 3 \rfloor + 1})_{1 + (i \bmod 3)}\).
  We will create a schedule of makespan at most~\(\ell\) by applying recipes to~\(\mathcal{I}\).

  We start by showing how to resolve agents~\(A(R^1)\).
  We want to apply Recipe~\ref{recipe:vc-hardness:red-star} to~\(A(R^1)\) in turn $1$.
  By permuting~\(\vec\sigma\) and entries within each~\(\sigma_i \in \vec\sigma\), we can assume that~\(\iota_1 = \sigma_{1, 1}\) so the unique agent~\(b\) that has~\(s_0(b) = u_1\) belongs to a red edge that is resolved first by Recipe~\ref{recipe:vc-hardness:red-edge}.
  The duration bound of Recipe~\ref{recipe:vc-hardness:red-star} shows that~\(A(R^1)\) will be resolved in exactly~\(\ell = \| \vec{\beta} \| + 3n = n\varphi + 3n\)~turns.

  Next, we want to take care of the agents~\(A(R^7)\).
  We can employ Recipe~\ref{recipe:vc-hardness:red-star} in turn $1$ once more.
  None of these agents will collide with agents~\(A(R^1)\) hence we have a valid schedule so far.
  The duration bound of Recipe~\ref{recipe:vc-hardness:red-star} again shows that~\(A(R^7)\) will be resolved in exactly~\(\ell\)~turns.

  For agents~\(A^X\) and~\(A^Y\) we apply Recipes~\ref{recipe:vc-hardness:Ax} and~\ref{recipe:vc-hardness:Ay} respectively whenever possible.
  This means that in every turn, we check if the preconditions of Recipes~\ref{recipe:vc-hardness:Ax} and~\ref{recipe:vc-hardness:Ay} regarding the occupancy of vertices~\(u_2\) and~\(u_4\) are satisfied.
  And if they are, then we apply the recipes on the agent with the least index among~\(A^X\) or~\(A^Y\).
  We want to prove that all agents~\(A^X\) and~\(A^Y\) are resolved by turn~\(\ell\) and that they do not collide with each other and with agents~\(A(R^7)\) and~\(A(R^1)\).
  Let us start with~\(A^X\).
  From Recipe~\ref{recipe:vc-hardness:red-star} we know that~\(u_4\) is unoccupied in turns~\(i(\varphi + 3) - 3\) and~\(i(\varphi + 3) - 2\) for every~\(i \in [2n - 2]\).
  As~\(\langle \vec\beta \rangle\) has a 3\=/partition, we know that if we resolve~\(\A(R^1)\) using Recipe~\ref{recipe:vc-hardness:red-star} with permutation~\(\vec{\iota}\), then~\(u_2\) is unoccupied in turns~\(i(\varphi + 3) - 1\) and~\(i(\varphi + 3)\) for every~\(i \in [n]\).
  So for every~\(i \in [n - 1]\) we can send two agents of~\(A^X\) using Recipe~\ref{recipe:vc-hardness:Ax} in turns~\(i(\varphi + 3) - 5\) and~\(i(\varphi + 3) - 4\) without colliding into any agent of~\(A(R^7)\) or~\(A(R^1)\).
  The last agent of~\(A^X\) arrives into their target in turn~\(\ell - \varphi\); thus, the makespan of the schedule so far is still~\(\ell\).
  Now we want to resolve~\(A^Y\).
  Let~\(P_i = \sum_{j = 1}^i \beta_{\iota_j}\), i.e.~the prefix sum of~\(\beta\) with respect to permutation~\(\iota\).
  By Recipe~\ref{recipe:vc-hardness:red-star} we know that for every~\(i \in [3n]\) vertex~\(u_2\) is unoccupied in turns~\(P_i + i - 1\) and~\(P_i + i\) by~\(A(R^1)\) but due to Recipe~\ref{recipe:vc-hardness:Ax}, vertex~\(u_2\) is now occupied in turns~\(i(\varphi + 3) - 1\) and~\(i(\varphi + 3)\).
  Nevertheless, there are now at least~\(4n + 2\)~turns during which~\(u_2\) is unoccupied.
  However, two of those turns are turns~\(\ell - 1\) and~\(\ell\) so if we want to resolve~\(A^Y\) using Recipe~\ref{recipe:vc-hardness:Ay} in turns~\(\ell - 3\) and~\(\ell - 2\), then they will not reach any vertex~\(V_Y^+\) by turn~\(\ell\).
  Thus there are exactly~\(4n\)~turns during which~\(u_2\) is unoccupied that can be used to resolve~\(A^Y\) using Recipe~\ref{recipe:vc-hardness:Ay}.
  As each input~\(\beta_i\) is at least~6, the last agent of~\(A^Y\) to be resolved by Recipe~\ref{recipe:vc-hardness:Ay} will reach their target vertex by turn~\(\ell\).

  Finally, we need to resolve the agents introduced by attaching bow ties.
  In the schedule so far, there are~\(2n - 2\)~agents that pass through vertices~\(u_3\), \(u_5\) and~\(u_6\) at some point and there are~\(4n\)~agents that pass through vertices~\(u_8\) and~\(u_9\).
  Thus, if we apply Recipe~\ref{recipe:vc-hardness:bowtie} whenever we can as soon as possible, then every agent introduced by a bow tie will be resolved in turn~\(\ell\) at the latest.
\end{proof}

\paragraph{Correctness (``only if'').}
Now we want to prove the backward direction. That is, starting from a schedule of makespan at most~\(\ell=n\varphi+3\), we want to produce a~3\=/partition of~\(\vec\beta\).
For the rest of the proof, let us fix some schedule~\(s\) of makespan at most~\(\ell\).
For any schedule, we consider the initial turn to be part of the schedule as well.
So if we claim that some invariant holds throughout the schedule, then this invariant must also hold in the~0\=/th turn.

Observe that if~\(s\) behaves according to our Recipes, then we can easily retrieve the partition: look at the permutation~\(\pi\) of with which agents~\(\{A(R^1_i)\}_{i = 1}^{3n}\) of the star~\(R^1\) are resolved using Recipe~\ref{recipe:vc-hardness:red-star}, and set every triple~\((\pi_{3j - 2}, \pi_{3j - 1}, \pi_{3j})_{j = 1}^n\) as a partition of~\(\vec\beta\).
To see why this is a partition, we need to inspect agents~\(A^X\) and~\(A(R^7)\). 

We will now show that, actually, $s$ has to follow our Recipes if it is of makespan at most $\ell$.
Let us first show a lower bound on the number of turns needed to resolve any red star.
\begin{claim}
  \label{claim:vc-hardness:red-star-opt}
  Let~\(R \in \{R^1, R^7\}\) be a red star with edge sizes~\(\vec{z} \in \N^m\) that is attached to a leaf of~$T$ that has $w$ as a sole neighbor in $T$. If the agents of~\(A(R)\) never go to vertices~\(N_T(w) \setminus \{u\}\), then it takes at least~\(\|\vec{z}\| + m=\ell\)~turns for all agents of~\(A(R)\) to reach their targets and they occupy~\(w\) for at least~\(\|\vec{z}\| - m\)~turns.
\end{claim}

\begin{proofclaim}
  Consider the agents~\(A' = A(R_i)\) for any~\(i \in [m]\) of any red edge.
  From the construction of a red edge (\(R_i\)), no agent of $A'$ can reach their target unless there is an agent $a'\in A'$ waiting in $w$, or there is only one agent of $A'$ that remains unresolved.
  This gives us two conclusions regarding~\(A'\):
  \begin{enumerate}
      \item the number of turns needed to move each~\(a \in A'\) from~\(s_0(a)\) to~\(t(a)\) is at least~\(|A'| + 1\), the~``\(+1\)''~factor comes from an agent that needs to return to~\(u\) from~\(w\), and
      \item the number of turns in which there is an agent occupying~\(w\) is at least~\(|A'| - 1\).
  \end{enumerate}

  Let us see what happens when there are multiple red edges.
  Let~\(A_1 = A(R_i)\) and~\(A_2 = A(R_j)\) for~\(i, j \in [m]\) distinct. Similarly to before, no agent of $A_1$ ($A_2$ resp.) can reach their target unless there is a \textit{waiting agent} $a'_1\in A_1$ ($a'_2\in A_2$ resp.) located in $w$, or there is only one agent of $A_1$ ($A_2$ resp.) that remains unresolved and occupies $u$. 
  In other words, while there is an agent of~\(A_1\) in~\(\{w,u\}\), none of the agents from~\(A_2\) can reach their targets and \textit{vice versa}. So, there are basically two ways to resolve the red star formed by these two red edges: either the agents of $A_1$ are completely resolved before resolving the agents of $A_2$, or the agents of $A_1$ start getting resolved, but this process is interrupted by the resolution of the agents of $A_2$. Clearly, the second way is suboptimal, as it requires at least two extra turns to make the extra switch between the waiting agents of $A_1$ and $A_2$. 
  Generalizing for stars with arbitrary many red edges, we get that the number of turns needed to resolve all the agents of~\(A(R)\) is at least that of resolving each red edge of $R$ before moving on to the next one, which is~\(|A(R)| + m\). Moreover, the number of turns that~\(w\) is occupied is at least~\(|A(R)| - m\).
\end{proofclaim}

We also obtain the following corollary from \Cref{claim:vc-hardness:red-star-opt}.
\begin{corollary}
  \label{crly:vc-hardness:always-occupied}
  In any schedule of makespan~\(\ell\), vertices~\(u_2\) and~\(u_4\) are occupied in turns~\([2, \ell - 2]\).
\end{corollary}
    Indeed, \Cref{claim:vc-hardness:red-star-opt} shows that~\(u_2\) is occupied in at least~\(n\varphi - 3n\)~turns by agents~\(A(R^1)\).
    Each agent~\(a \in A^X \cup A^Y\), of which there are~\(6n - 2\), needs to visit~\(u_2\) at least once to get to~\(t(a)\).
    This means that there are at least~\(n\varphi -3n + 6n - 2 = n\varphi + 3n - 2 = \ell -2\)~turns during which there is an agent who is occupying~\(u_2\).
    And none of these~\(n\varphi + 3n - 2\)~turns can be turns~1 or~\(\ell - 1\) as~\(\dist(u_2, s_0(A) \cup t(A)) > 1\), and it cannot be turn~\(\ell\) as~\(u_2 \not\in t(A)\). The argument for~\(u_4\) is analogous.
\medskip

Now, we want to prove that the assumptions of \Cref{claim:vc-hardness:red-star-opt} are satisfied.
First let us observe the following about the behavior of the agents of the bow ties in schedule~\(s\). Recall that \(A_T(u)\) is the set of agents introduced by a bow tie attached on a vertex $u\in U$.
\begin{observation}
  \label{obs:vc-hardness:bowtie-bound}
  Consider any bow tie of size $b$ attached to a vertex $u\in U$. 
  There are at most~\(\ell - b - 1\)~turns during which an agent in~\(A \setminus A_T(U)\) can enter~\(u\).
\end{observation}
Indeed, each agent~\(a \in A_T(u)\) has to enter~\(u\) at least once to get to~\(t(a)\). Also, since~\(u\) is not a target of any agent, \(u\) is unoccupied in turn~\(\ell\).

\smallskip

We stress that \Cref{obs:vc-hardness:bowtie-bound} does not prohibit the existence of a bow tie attached to a vertex~\(u \in U\) and an agent~\(a \in A_T(u)\) that goes to~\(t(a)\) via a detour that includes a vertex~\(N_T(u)\) and then returns to~\(u\) in order to get to~\(t(a)\).\todo{DK: I am not sure a fully understand this sentence.}
But that will not happen in our case.
Let~\(u \in \{u_3, u_5, u_6\}\).
For a bow tie attached to~\(u\), \Cref{obs:vc-hardness:bowtie-bound}~gives~\(2n - 1\)~turns during which $u$ is not occupied by an agent of~\(A_T(u)\); let us call these the \textit{available turns}.
One of the available turns is the last turn of~\(s\), so there are actually~\(2n - 2\)~available turns during which an agent of~\(A \setminus A_T(u)\) can stay in~\(u\).
Each agent of~\(A^X\) has to pass through~\(u\) to get to their target, meaning that $u$ should be occupied by an agent of $A^X$ during each of the available turns. 
A similar argument applies for~\(u_8\)~and~\(u_9\) and agents~\(A^Y\).
Thus, none of the agents~\(A(R^1)\) can enter~\(N_T(u_2) \setminus \{u_1\}\) and none of the agents~\(A(R^7)\) can enter~\(N_T(u_4) \setminus \{u_7\}\), and this shows that the assumptions of \Cref{claim:vc-hardness:red-star-opt} are satisfied.

As agents of~\(A^X\) and~\(A^Y\) have to pass through~\(u_2\) to get to their targets and~\(|A^X| + |A^Y| = 6n - 2\), we now know that agents of~\(A(R^1)\) can occupy~\(u_2\) for at most~\(\ell - (6n - 2) = n\varphi - 3n + 2\)~turns.
Moreover, as~\(\dist(u_2, t(A)) > 1\), then the upper bound can be decreased to~\(n\varphi - 3n\) since in the last two turns of~\(s\) there cannot be an agent in~\(u_2\).
An analogous argument shows that~\(u_4\) can be occupied by~\(A(R^7)\) in at most~\(n\varphi + n\)~turns.
This will be helpful to show that both~\(A(R^1)\) and~\(A(R^7)\) have to be resolved by~\(s\) using Recipe~\ref{recipe:vc-hardness:red-star}.
\begin{claim}
  \label{claim:vc-hardness:continuous-agent-resolution}
  The only way for~\(s\) (of makespan~\(\ell\)) to resolve the agents of~\(A(R^1)\), assuming that they can only enter vertices~\(u_1\) and~\(u_2\) of~\(U\), and that they can spend at most~\(n\varphi - 3n\)~turns in~\(u_2\), is if they follow Recipe~\ref{recipe:vc-hardness:red-star} from the first turn of $s$.
\end{claim}
\begin{proofclaim}
  Let~\(A' = A(R^1)\) and~\(A'_i = A(R^1_i)\) for every~\(i \in [3n]\).
  Due to the lower bound from \Cref{claim:vc-hardness:red-star-opt} on the number of turns needed to resolve~\(A'\), we know that~\(s\) has to resolve~\(A'\) from the first turn.

  The lower bound of~\(n\varphi + 3n\) on the occupancy of~\(u_1\) from \Cref{claim:vc-hardness:red-star-opt} also shows that once an agent arrives to their target they will not return to~\(u_1\) as doing so would extend the makespan beyond~\(\ell\).
  Note that in each turn there can only be one agent of~\(a \in A'\) that arrives to~\(t(a)\) since~\(\cup_{v \in s_0(A')} N(v) = u_1\).
  Hence we can define~\(\tau(a)\) for every~\(a \in A'\) to be the turn in which~\(a\) arrives in~\(t(a)\) and we know for distinct~\(a', a'' \in A'\) values~\(\tau(a')\) and~\(\tau(a'')\) are also distinct.
  For~\(i \in [3n]\) let~\(\tau_i\) be the~\(i\)\=/th agent to arrive to their target.

  We want to prove that~\(s\) has to resolve the agents~\(A'_i\) one by one (i.e., consecutively with no breaks).
  Suppose not, and assume that there exists an~\(A'_i\) such that there exist distinct~\(p, q \in [|A'|]\) with~\(p < q - 1\) where~\(\tau_p, \tau_q \in A'_i\) but for each~\(r \in [p + 1, q - 1]\) we have that  \(\tau_r \not\in A'_i\).
  Let~\(p'\) be the number of agents of~\(A'_i\) that are resolved by turn~\(p\).
  Whenever there is a turn~\(\rho\) such that there exists an~\(a \in A'_i\) whose path to~\(t(a)\) is not blocked, there must exist an~\(a' \in A'_i\) with~\(s_{\rho}(a') = u_2\).
  From the proof of \Cref{claim:vc-hardness:red-star-opt} it follows that~\(u_1\) is occupied by agents of~\(A'_i\) for at least~\(p' + 1\)~turns by turn~\(\tau(\tau_p)\).
  And from turn~\(\tau(\tau_q)\) on, vertex~\(u_1\) must again be occupied by agents of~\(A'_i\) for at least~\(|A'_i| - p' + 1\)~turns.
  But these two amounts together imply that agents of~\(A'_i\) occupy~\(u_1\) for at least~\(|A'_i| + 2\)~turns.
  The existence of such an~\(A'_i\) in combination with the lower bound of \Cref{claim:vc-hardness:red-star-opt} shows that the makespan of~\(s\) is strictly more than~\(\ell\).

  Now we know that the agents~\(A'_i\) are resolved one by one.
  We also need to show that~\(A'_1\) is the first one to be resolved.
  If another~\(A'_j\) with~\(j \neq 1\) is resolved first, then this means that the agent~\(a \in A'_1\) with~\(s_0(a) = u_1\) has to be moved elsewhere.
  But then when it is the turn of the agents of~\(A'_1\) to be resolved, then they will occupy~\(u_1\) for at least~\(|A'_1| + 2\)~turns where~\(|A'_1| +1\) turns are spent by following \Cref{recipe:vc-hardness:red-star} and one more turn for the~0\=/th~turn.
  This in combination with \Cref{claim:vc-hardness:red-star-opt} contradicts that the makespan of~\(s\) is~\(\ell\).
\end{proofclaim}

We can also prove an analogue of \Cref{claim:vc-hardness:continuous-agent-resolution} for agents~\(A(R^7)\).
We will additionally show that the last agent group~\(A(R^7_n)\) is the last one resolved among~\(A(R^7)\).
We need to make an addition for~\(R^7\) because we want to show that the red edge~\(R^7_n\) of size~\(\varphi + 4\) is the last one to be resolved.
\begin{claim}
  The only way for~\(s\) (of makespan~\(\ell\)) to resolve the agents of~\(A(R^7)\), assuming that they can only enter the vertices~\(u_7\) and~\(u_4\) of~\(U\), and that they spend at most~\(n\varphi + n\)~turns in~\(u_4\), is if they follow \Cref{recipe:vc-hardness:red-star} where~\(A(R^7_n)\) is the last agent group resolved among the agents of~\(A(R^7)\).
\end{claim}
\begin{proofclaim}
  We can prove that the schedule obeys \Cref{recipe:vc-hardness:red-star} analogously to the proof of \Cref{claim:vc-hardness:continuous-agent-resolution}.
  We want to additionally prove that the agents of~\(A(R^7_n)\) are the last agent group to be resolved among the agents of~\(A(R^7)\).
  Suppose not and assume that~\(A(R^7_n)\) are resolved as the \(i\)\=/th agent group among the agents of~\(A(R^7)\).
  By applying \Cref{recipe:vc-hardness:red-star}, this would mean that~\(A(R^7_n)\) are resolved in turn~\(\varphi + 1 + (i - 2)(\varphi + 3) + \varphi + 5 = i\varphi + 3i\) and the agent group resolved before~\(A(R^7_n)\) is resolved in turn~\(i\varphi + 4i - (\varphi + 5) = i\varphi + 3i - \varphi - 5\).
  During the~\(\varphi + 5\) turns between the resolutions of~\(A(R^7_n)\) and the preceding agent group, the maximum number of agents that can be resolved is~\(\varphi + 4\) by \Cref{recipe:vc-hardness:red-edge}.
  Due to~\(\varphi / 4 < \beta_i < \varphi / 2\) for each~\(i \in [3n]\) we know that there are exactly three agent groups that are resolved between the resolution of the agents of~\(A(R^7_n)\) and the preceding agent group.
  And due to \Cref{claim:vc-hardness:3-partition-multiple} there are no three agent groups whose total size can be~\(\varphi + 4\).
  Therefore a schedule which does not resolve~\(A(R^7_n)\) the last among the agent groups of~\(A(R^7)\) necessarily wastes some turns during which it could be resolving agents, which in combination with \Cref{crly:vc-hardness:always-occupied} and~\Cref{claim:vc-hardness:red-star-opt} (depending on whether in the wasted turns the agent returning from~\(u_4\) goes straight to a vertex of~\(V(R^7)\) or stays a few turns in~\(u_7\)) contradicts that the makespan of~\(s\) at most~\(\ell\).
\end{proofclaim}

By \Cref{recipe:vc-hardness:red-star} we know that~\(u_4\) is unoccupied by the agents of~\(A(R^7)\) in turns~\(i(\varphi + 3) - 3\) and~\(i(\varphi + 3) - 2\) for every~\(i \in [n]\) and these agents must occupy~\(u_4\) in these turns due to \Cref{claim:vc-hardness:red-star-opt}. This shows that \Cref{recipe:vc-hardness:red-star} has to be applied from the beginning.
This allows~\(2n\)~turns for the agents of~\(A^X\) to get to their targets.
But~2 of these~\(2n\)~turns are actually turns~\(\ell - 1\) and~\(\ell\).
If an agent of~\(A^X\) tries to enter~\(u_4\) in either of these turns, then they will not make it to their targets by turn~\(\ell\).
Thus there are~\(2n - 2\)~turns during which the agents of~\(A^X\) can enter~\(u_4\).
If they do so, then they will reach~\(u_2\) in turns~\(i(\varphi + 3) - 1\) and~\(i(\varphi + 3)\) for every~\(i \in [n]\).
This in particular means that the agents of~\(A(R^1)\) are not in~\(u_2\) in the turns~\(i(\varphi + 3) - 1\) and~\(i(\varphi + 3)\).
It then follows by \Cref{claim:vc-hardness:continuous-agent-resolution} and \Cref{recipe:vc-hardness:red-star} that the total number agents of~\(A(R^1)\) that are resolved between turns~\((i - 1)(\varphi + 3)\) and~\(i(\varphi + 3)\) is~\(\varphi\) (not~\(\varphi + 3\) as \Cref{recipe:vc-hardness:red-star} requires~\(|A(R^1_i)| + 1\)~turns to resolve the agents of~\(A(R^1_i)\) for any~\(i \in [3n]\)).

To sum up, we obtain that according to any schedule $s$ of makespan at most $\ell$, each agent group we have defined has to be resolved following one of our Recipes. This, in combination with \Cref{lemma:vc-hardness:forward} and \Cref{lemma:vc-hardness:small-vc}, concludes the proof of this theorem.
\end{proof}

\subsection{Trees With Few Leaves}

This subsection is dedicated to the proof of Theorem~\ref{thm:pancake}. This proof is achieved through a reduction from the \PN{Pancake Flipping} problem, which is known to be \CC{NP}-hard~\cite{BFR15}:
as input we receive a permutation $\pi = \pi_1, \dots, \pi_n$ of the set~$[n]$ together with a positive integer~$k$.
The question is whether $\pi$ can be sorted using $k$~prefix reversals, i.e., the operation of replacing a prefix of arbitrary length with its reversal.
Note that we can assume that $k$ is at most $\frac{18}{11} n$, otherwise it is trivially a yes-instance since every permutation of length~$n$ can be sorted using at most $\frac{18}{11} n$ prefix reversals~\cite{CFMMSSV09}.

\thmNpTrees*
\begin{proof}
Let $\langle \pi, k\rangle$ be an instance of \PN{Pancake Flipping} where $\pi$ is a permutation of length~$n$ and set $n^+ \coloneqq n + 2$, $L \coloneqq 3n^+k$.
Let us first describe the high level idea of the reduction before getting into all the technical details.
We start with a description of a graph~$G'$ that forms the backbone of the final graph in the \MAPFShort instance.
The graph $G'$ consists of three paths, referred to as \emph{$A$-path}, \emph{$B$-path} and \emph{$C$-path}, with their endpoints connected to a central vertex~$v^\star$.
The $A$-path is of length~$n$, whereas both $B$- and $C$-paths are of length $L$.
Let us denote the vertices of the $A$-path as $v^A_0,\dots, v^A_n$, and the vertices of the $B$-path and $C$-path as $v^B_0, \dots, v^B_L$ and $v^C_0, \dots, v^C_L$ respectively, always starting from the neighbors of $v^\star$ outwards.

There are exactly $n$~\emph{primary agents} $a_1, \dots, a_n$ that start on the vertices of the $A$-path in the order given by~$\pi$, i.e., we set $s_0(a_i) = v^A_{\pi_i}$ for each $i \in [n]$.
Their final destinations lie also on the $A$-path, this time in the increasing order of their indices, i.e., $t(a_i) = v^A_i$ for each $i \in [n]$.
By the inclusion of many additional \emph{auxiliary agents}, we will enforce that any optimum schedule has makespan~$L$ and moreover, it can be divided into exactly $k$~rounds with each round consisting of three distinct phases --- \emph{push}, \emph{reverse} and \emph{pop}.
Each phase lasts exactly $n^+$ time steps.

In the push phase, the vertex~$v^C_0$ is blocked and arbitrary prefix of primary agents with respect to the distance from~$v^\star$ can move from the $A$-path to the $B$-path.
In the reverse phase, the vertex~$v^A_0$ is blocked and thus, the $A$-path is separated from the $B$- and $C$-paths.
Primary agents can move from the $B$-path to the $C$-path, effectively reversing their order with respect to the distance from~$v^\star$.
Finally in the pop phase, the vertex~$v^B_0$ is blocked, separating the $B$-path from the rest of the graph.
All the primary agents located on the $C$-path now return to the $A$-path in reversed order.
Moreover, auxiliary agents enforce that no primary agent can reside on the $C$-path during the push phase and similarly, that no primary agent can wait on the $B$-path during the pop phase.
In this way, the primary agents can alter their order along the $A$-path in each round by performing a prefix reversal.
Consequently, the primary agents can rearrange themselves into their final destinations on the $A$-path if and only if $\pi$ can be sorted with $k$~prefix reversals.

\paragraph{Reduction.}
Let us now describe the reduction in full detail, starting with the construction of the graph.
The final graph~$G$ is obtained by adding to~$G'$ four \emph{auxiliary paths} of length~$2L$, each connected with an edge to  one of the vertices $v^A_0$, $v^B_0$ or $v^C_0$.
Namely, we add
\begin{enumerate}
\item a path with vertices $u^A_{-L}, \dots, u^A_0, \dots, u^A_L$ that is connected by an edge $u^A_{-1} v^A_0$,
\item a path with vertices $w^A_{-L}, \dots, w^A_0, \dots, w^A_L$ that is connected by an edge $v^A_0 w^A_1$,
\item a path with vertices $u^B_{-L}, \dots, u^B_0, \dots, u^B_L$ that is connected by an edge $u^B_{-1} v^B_0$, and
\item a path with vertices $u^C_{-L}, \dots, u^C_0, \dots, u^C_L$ that is connected by an edge $u^C_{-1} v^C_0$.
\end{enumerate}

Observe that $G$ is indeed a tree with exactly 11~leaves.
See Fig.~\ref{fig:pancake}.

\begin{figure}[!t]
\centering

\begin{tikzpicture}[scale=0.6, inner sep=0.6mm]

    \node[draw, circle, line width=1pt, color=red, fill=white](va1) at (-1,2)[label=left: { $v_n^A$}] {};
    \node[draw, circle, line width=1pt, color=red, fill=white](va2) at (3,2)[] {};
    \node[draw, circle, line width=1pt, color=red, fill=white](va3) at (5,2)[label=above right: { $v_0^A$}] {};
    \node[draw, circle, line width=0.5pt, fill=white](va4) at (0,0)[label=left: { $u_L^A$}] {};
    \node[draw, circle, line width=0.5pt, fill=white](va5) at (1,0)[] {};
    \node[draw, circle, line width=0.5pt, fill=white](va6) at (4,0)[] {};
    \node[draw, circle, line width=0.5pt, fill=white](va7) at (5,0)[label=right: { $u_{-1}^A$}] {};
    \node[draw, circle, line width=0.5pt, fill=white](va8) at (5,-1)[] {};
    \node[draw, circle, line width=0.5pt, fill=white](va9) at (5,-4)[] {};
    \node[draw, circle, line width=0.5pt, fill=white](va10) at (5,-5)[label=right: { $u_{-L}^A$}] {};

    \draw[thick middle dotted line] (va1) -- (va2);
    \draw[-, line width=1pt, color=red]  (va2) -- (va3);
    \draw[-, line width=0.5pt]  (va3) -- (va7);
    \draw[-, line width=0.5pt]  (va4) -- (va5);
    \draw[middle dotted line] (va5) -- (va6);
    \draw[-, line width=0.5pt]  (va6) -- (va7);
    \draw[-, line width=0.5pt]  (va7) -- (va8);
    \draw[middle dotted line] (va8) -- (va9);
    \draw[-, line width=0.5pt]  (va9) -- (va10);

    \node[draw, circle, line width=0.5pt, fill=white](w1) at (0,4)[label=left: { $w_{-L}^A$}] {};
    \node[draw, circle, line width=0.5pt, fill=white](w2) at (1,4)[] {};
    \node[draw, circle, line width=0.5pt, fill=white](w3) at (4,4)[] {};
    \node[draw, circle, line width=0.5pt, fill=white](w4) at (5,4)[label=right: { $w_{1}^A$}] {};
    \node[draw, circle, line width=0.5pt, fill=white](w5) at (5,5)[] {};
    \node[draw, circle, line width=0.5pt, fill=white](w6) at (5,8)[] {};
    \node[draw, circle, line width=0.5pt, fill=white](w7) at (5,9)[label=right: { $w_{L}^A$}] {};

    \draw[-, line width=0.5pt]  (w1) -- (w2);
    \draw[middle dotted line] (w2) -- (w3);
    \draw[-, line width=0.5pt]  (w3) -- (w4);
    \draw[-, line width=0.5pt]  (w4) -- (va3);
    \draw[-, line width=0.5pt]  (w4) -- (w5);
    \draw[middle dotted line] (w5) -- (w6);
    \draw[-, line width=0.5pt]  (w6) -- (w7);

    \node[draw, circle, line width=1pt, color=red, fill=white](vb1) at (8,11)[label=right: { $v_L^B$}] {};
    \node[draw, circle, line width=1pt, color=red, fill=white](vb2) at (8,7)[] {};
    \node[draw, circle, line width=1pt, color=red, fill=white](vb3) at (8,5)[label=left: { $v_0^B$}] {};
    \node[draw, circle, line width=0.5pt, fill=white](vb4) at (10,10)[label=right: { $u_L^B$}] {};
    \node[draw, circle, line width=0.5pt, fill=white](vb5) at (10,9)[] {};
    \node[draw, circle, line width=0.5pt, fill=white](vb6) at (10,6)[] {};
    \node[draw, circle, line width=0.5pt, fill=white](vb7) at (10,5)[label=below: { $u_{-1}^B$}] {};
    \node[draw, circle, line width=0.5pt, fill=white](vb8) at (11,5)[] {};
    \node[draw, circle, line width=0.5pt, fill=white](vb9) at (14,5)[] {};
    \node[draw, circle, line width=0.5pt, fill=white](vb10) at (15,5)[label=right: { $u_{-L}^B$}] {};

    \draw[thick middle dotted line] (vb1) -- (vb2);
    \draw[-, line width=1pt, color=red]  (vb2) -- (vb3);
    \draw[-, line width=0.5pt]  (vb3) -- (vb7);
    \draw[-, line width=0.5pt]  (vb4) -- (vb5);
    \draw[middle dotted line] (vb5) -- (vb6);
    \draw[-, line width=0.5pt]  (vb6) -- (vb7);
    \draw[-, line width=0.5pt]  (vb7) -- (vb8);
    \draw[middle dotted line] (vb8) -- (vb9);
    \draw[-, line width=0.5pt]  (vb9) -- (vb10);

    \node[draw, circle, line width=1pt, color=red, fill=white](vc1) at (17,2)[label=right: { $v_L^C$}] {};
    \node[draw, circle, line width=1pt, color=red, fill=white](vc2) at (13,2)[] {};
    \node[draw, circle, line width=1pt, color=red, fill=white](vc3) at (11,2)[label=above: { $v_0^C$}] {};
    \node[draw, circle, line width=0.5pt, fill=white](vc4) at (16,0)[label=right: { $u_L^C$}] {};
    \node[draw, circle, line width=0.5pt, fill=white](vc5) at (15,0)[] {};
    \node[draw, circle, line width=0.5pt, fill=white](vc6) at (12,0)[] {};
    \node[draw, circle, line width=0.5pt, fill=white](vc7) at (11,0)[label=left: { $u_{-1}^C$}] {};
    \node[draw, circle, line width=0.5pt, fill=white](vc8) at (11,-1)[] {};
    \node[draw, circle, line width=0.5pt, fill=white](vc9) at (11,-4)[] {};
    \node[draw, circle, line width=0.5pt, fill=white](vc10) at (11,-5)[label=right: { $u_{-L}^C$}] {};

    \draw[thick middle dotted line] (vc1) -- (vc2);
    \draw[-, line width=1pt, color=red]  (vc2) -- (vc3);
    \draw[-, line width=0.5pt]  (vc3) -- (vc7);
    \draw[-, line width=0.5pt]  (vc4) -- (vc5);
    \draw[middle dotted line] (vc5) -- (vc6);
    \draw[-, line width=0.5pt]  (vc6) -- (vc7);
    \draw[-, line width=0.5pt]  (vc7) -- (vc8);
    \draw[middle dotted line] (vc8) -- (vc9);
    \draw[-, line width=0.5pt]  (vc9) -- (vc10);

    \node[draw, circle, line width=1pt, color=red, fill=white](v) at (8,2)[label=below: { $v^\star$}] {};
    \draw[-, line width=1pt, color=red]  (v) -- (va3);
    \draw[-, line width=1pt, color=red]  (v) -- (vb3);
    \draw[-, line width=1pt, color=red]  (v) -- (vc3);

\end{tikzpicture}
\caption{The graph $G$ constructed in the proof of Theorem~\ref{thm:pancake}. The color red (black resp.) is used to indicate the main (auxiliary resp.) vertices and paths.}\label{fig:pancake}
\end{figure}

Now let us define the auxiliary agents.
First, there are $L$~\emph{auxiliary $B$-agents} $\{b^B_i \mid i \in [L]\}$ and $L$~\emph{auxiliary $C$-agents} $\{b^C_i \mid i \in [L]\}$.
These groups of agents start in a single file on vertices $\{u^B_{-i} \mid i \in [L]\}$ and $\{u^C_{-i} \mid i \in [L]\}$ respectively.
Their target destinations lie at distance exactly~$L$ from their starting positions, either on their starting auxiliary paths or on the $B$- or $C$-paths, respectively.
We set~\(s_0(b^B_i) = u^B_{-i}\), \(s_0(b^C_i) = u^C_{-i}\), and
\begin{align*}
  t(b^B_i)   &= \begin{cases}
  v^B_{L-i}  &\text{if } i \in [3 \ell n^+ + 2n^+, 3\ell + 3n^+) \text{ for some } \ell \in \Nzero, \text{ and}\\
  u^B_{L-i}  &\text{otherwise.}
  \end{cases} \\
  t(b^C_i)   &= \begin{cases}
    v^C_{L-i} &\text{if } i \in [3 \ell n^+, 3\ell n^+ + n^+) \text{ for some } \ell \in \Nzero, \text{ and} \\ 
    u^C_{L-i} &\text{otherwise}
  \end{cases}
\end{align*}
where $[a,b)$ denotes the set $\{a, a+1, \dots,b-1\}$.

Finally, there are $L + 2n^+k$~\emph{auxiliary $A$-agents} split into two sets.
First, we have agents $\{b^A_{i,1} \mid i \in [L]\}$ with $s_0(\cdot)$ and $t(\cdot)$ defined as follows
\begin{align*}
  s_0(b^A_{i,1}) &= u^A_{-i}, \\
  t(b^A_{i,1})   &= \begin{cases}
    w^A_{L-i} & \text{if } i \in [3 \ell n^+ + n^+, 3\ell + 2n^+) \text{ for some } \ell \in \Nzero, \text{ and} \\
    u^A_{L-i} &\text{otherwise.}
  \end{cases}
\end{align*}
Second, we add auxiliary $A$-agent $b^A_{i,2}$ for every $i \in [L] \cap \bigcup_{\ell \in \Nzero} [3 \ell n^+, 3\ell + n^+) \cup [3 \ell + 2n^+, 3\ell + 3n^+)$ with the following starting and target positions
\[s_0(b^A_{i,2}) = w^A_{-i}, \quad t(b^A_{i,2}) = w^A_{L-i}.\]
Notice that every vertex~$u^A_{-i}$ for $i \in [L]$ is a starting position of some auxiliary $A$-agent, and every vertex~$w^A_{i-1}$ for $i \in [L]$ is a target position of some auxiliary $A$-agent.
This concludes the description of the instance $\langle G, A, s_0, t, L \rangle$ of \MAPFShort.

Before moving to prove the correctness of the reduction, let us make a few observations about the movement of auxiliary agents in any schedule with optimum makespan.
We first formally define the partition of the schedule into rounds and phases.
The $i$th round takes place in the time interval $[3(i-1) n^+, 3 i n^+)$ and it consists of a push phase in the interval $[3(i-1) n^+, 3(i-1)n^+ + n^+)$, a reverse phase in the interval $[3(i-1) n^+ + n^+, 3(i-1)n^+ + 2n^+)$ and a pop phase in the interval $[3(i-1) n^+ + 2n^+, 3 i n^+)$.
For technical reasons, we consider the time step 0 to be neither part of the first round nor the first push phase.

Recall that the distance of all auxiliary agents from their respective targets is exactly~$L$.
Moreover, the graph~$G$ is a tree and thus, their movement in any schedule of makespan~$L$ is fully predetermined.
It is easy to check that this movement alone creates no conflicts and is therefore feasible on its own.
Let us observe how the auxiliary agents interact with the vertices on the boundary of~$G'$.

\begin{claim}\label{claim:auxiliary-properties}
In any solution to $\langle G, A, s_0, t, L \rangle$ of makespan $L$, the auxiliary agents occupy
\begin{enumerate}
\item the vertices $u^A_{-1}$, $u^B_{-1}$ and $u^C_{-1}$ at time steps $[0,L-1]$,
\item the vertex $w^A_{1}$ at time steps $[2, L]$,
\item the vertex $v^A_0$ precisely during every reverse phase,
\item the vertex $v^B_0$ precisely during every pop phase, and
\item the vertex $v^C_0$ precisely during every push phase.
\end{enumerate}
\end{claim}
\begin{proofclaim}
The first part holds since the agents $\{b^A_{i,1} \mid i \in [L]\}$ pass successively through the vertex~$u^A_{-1}$ with~$b^A_{1,1}$ starting on~$u^A_{-1}$.
The same holds for $u^B_{-1}$ and $u^C_{-1}$ with the auxiliary $B$- and $C$-agents respectively.
The second part holds since the agent~$b^A_{1,2}$ arrives at the vertex~$w^A_1$ in two moves and afterwards, $w^A_1$ is occupied at each turn by either~$b^A_{i,1}$ or~$b^A_{i,2}$.

The last three parts follow directly from the definition of the target function~$t(\cdot)$ for auxiliary agents.
For example, the vertex~$v^B_0$ can possibly be occupied at time step~$i$ only by the auxiliary $B$-agent starting at distance exactly~$i$ from~$v^B_0$, i.e., the agent~$b^B_i$.
However, that happens only when the target destination of~$b^B_i$ lies on the $B$-path, i.e., if and only if $i \in [3 \ell n^+ + 2n^+, 3\ell + 3n^+)$.
\end{proofclaim}

We now show that $\langle G, A, s_0, t, L \rangle$ is a yes-instance of \MAPFShort if and only if $\langle \pi, k \rangle$ is a yes-instance of \PN{Pancake Flipping}.

\paragraph{Correctness (``only if'').}
First, let $\langle \pi, k\rangle$ be a yes-instance of \PN{Pancake Flipping} and let $r_1, \dots, r_k$ be the lengths of prefixes such that their successive reversals sort $\pi$.
As we already argued, the movement of all auxiliary agents is predetermined and thus, it remains to define the movement of primary agents.

Note that we shall not define the functions $s_i(\cdot)$ explicitly, we will rather define them implicitly through describing the movement of primary agents.
Throughout the whole schedule, we maintain the invariant that at the start of each round the primary agents are located at vertices $v^A_1, \dots, v^A_n$.
This clearly holds at the very beginning.

Let us describe the movement of primary agents in the $i$th round.
Let $A_i$ be the set of agents located at vertices $v^A_1, \dots, v^A_{r_i}$ at the start of the $i$th round.
These will be the only primary agents moving in this round.
In the push phase, the agents in $A_i$ immediately start moving in a single file to reach the vertices $v^B_1, \dots, v^B_{r_i}$.
Then they wait here until the push phase ends.
There is enough time to perform this movement since the agents travel distance exactly $r_i+1 \le n +1$ and each push phase lasts $n^+ = n + 2$ time steps, except for the first round when it lasts only $n+1$ steps.
Similarly in the reverse phase, the agents in $A_i$ move in parallel from vertices $v^B_1, \dots, v^B_{r_i}$ to vertices $v^C_1, \dots, v^C_{r_i}$.
Again, there is enough time since each reverse phase lasts exactly $n+2$ time steps.
Finally in the pop phase, the agents in $A_i$ move from vertices $v^C_1, \dots, v^C_{r_i}$ back to $v^A_1, \dots, v^A_{r_i}$.
Again, this can be done since each pop phase lasts exactly $n+2$ time steps.

Notice that the order of the agents in $A_i$ has reversed with respect to the distance from~$v^\star$.
Each round thus reorders the agents with respect to the distance from~$v^\star$ by reversing a prefix of length~$r_i$ and all primary agents arrive at their final destination at the end of the $k$th round.
It is straightforward to check that this movement of primary agents creates no conflict with the auxiliary agents.

\paragraph{Correctness (``if'').}
Now assume that $\langle G, A, s_0, t, L \rangle$ is a yes-instance of \MAPF and fix a particular feasible schedule.
We first prove that the movement of primary agents is quite restricted throughout the schedule.

\begin{claim}\label{claim:primary-agents}
In any solution to $\langle G, A, s_0, t, L \rangle$ of makespan~$L$, all primary agents stay within the graph~$G'$.
Moreover, there are no primary agents on the $B$-path during any pop phase and there are no primary agents on the $C$-path during any push phase.
\end{claim}
\begin{proofclaim}
The primary agents could escape~$G'$ only through vertices~$u^A_{-1}$, $u^B_{-1}$, $u^C_{-1}$ or $w^A_1$.
However, all these vertices are at distance at least~2 from the starting and target position of any primary agent and by Claim~\ref{claim:auxiliary-properties}, they are all blocked by auxiliary agents in the time interval $[2, L-1]$.
Therefore, all primary agents stay within the graph~$G'$ throughout the whole schedule.

Now assume for a contradiction that there is a primary agent~$a_j$ on the $B$-path during a pop phase in the~$i$th round.
By the fourth part of Claim~\ref{claim:auxiliary-properties}, the vertex~$v^B_0$ is occupied by auxiliary agents throughout the whole pop phase and thus, the primary agent~$a_j$ must already be on some vertex~$v^B_p$ for $p \ge 1$ when the pop phase starts.
We set $t = 3(i-1) n^+ + 2n^+$.
At time step~$t$, the auxiliary $B$-agent~$b^B_t$ enters the $B$-path on the vertex~$v^B_0$ and afterwards, it keeps moving along the $B$-path towards its target vertex~$v^B_{L-t}$.
Therefore, the agent~$a_j$ can never return to its final destination as it would have to swap with~$b^B_t$ at some point.
Analogous argument show that there cannot be any primary agent on the $C$-path during a push phase.
\end{proofclaim}

Let us refer to the $A$-path together with the central vertex $v^\star$ as the \emph{extended $A$-path}, with the \emph{extended $B$- and $C$-paths} defined analogously.
It follows by a combination of Claims~\ref{claim:auxiliary-properties} and~\ref{claim:primary-agents} that at the beginning of each round all primary agents are located on the extended $A$-path.
In a push phase, the primary agents can move only on the extended $A$- and $B$-paths.
Suppose that at the end of the $i$th push phase there are $r_i$~primary agents on the extended $B$-path that get separated from the remaining $n-r_i$ primary agents on the $A$-path.
In the following reverse phase, all $r_i$~primary agents from the extended $B$-path must move to the extended $C$-path since they can neither stay on the $B$-path due to Claim~\ref{claim:primary-agents} nor escape to the $A$-path due to the third part of Claim~\ref{claim:auxiliary-properties}.
Finally, these $r_i$~primary agents must vacate the $C$-path and return to the extended $A$-path in the final pop phase of this round by Claim~\ref{claim:primary-agents}.
As a result, the order of the primary agents with respect to the distance from~$v^\star$ changed exactly by a prefix reversal of length~$r_i$.
Since the primary agents are able to rearrange themselves into their final order in $k$~rounds, the permutation~$\pi$ can be sorted by performing successive prefix reversals of lengths $r_1, \dots, r_k$ and $\langle \pi,k \rangle$ is a yes-instance of \PN{Pancake Flipping}.
\end{proof}

Additionally, a minimal modifications to the reduction above yield a hardness result in a more general setting where the agents are semi-anonymous.
The input of a \CMAPF (\CMAPFShort) problem consists of a graph $G$, positive integers~$k$ and $\ell$ and a group of agents~$A_i$ for each $i \in [k]$ with a prescribed starting positions~$S_i$ and target positions~$T_i$.
We are again looking for a feasible schedule of makespan at most~$\ell$ but unlike~\MAPFShort, we do not specify the target of each individual agent.
We only require that the agents in~$A_i$ move from~$S_i$ to~$T_i$ for each $i \in [k]$.

\thmNpColored*
\begin{proof}
We reduce from the \PN{Binary String Prefix Reversal Distance} problem known to be \NP-hard~\cite{HurkensIKKST07}:
the input consists of two binary strings $\alpha = \alpha_1, \dots, \alpha_n$ and $\beta = \beta_1, \dots, \beta_n$ of length~$n$, and a positive integer~$k$.
The question is whether $\alpha$ can be transformed into $\beta$ with $k$ prefix reversals.

Let $\langle \alpha, \beta, k \rangle$ be an instance of \PN{Binary String Prefix Reversal Distance}.
We refer to the reduction of Theorem~\ref{thm:pancake}.
Observe that the constructed graph~$G$ and the starting and target positions of all auxiliary agents depend only on~$n$ and~$k$.
In particular, we do not make any modifications to the graph~$G$.
We shall also use the same set of agents, only split into six different groups.

The primary agents form the first two agent groups~$A_1$ and~$A_2$ where their starting and target positions correspond to the distribution of symbols $0$ and $1$ in the binary strings on input.
Specifically, we set $S_1 = \{i \mid \alpha_i = 0\}$, $T_1 = \{i \mid \beta_i = 0\}$, $S_2 = \{i \mid \alpha_i = 1\}$ and $T_2 = \{i \mid \beta_i = 1\}$.
The third and fourth group of agents consist of the auxiliary $B$-agents and auxiliary $C$-agents, respectively, with $S_3 = \{s_0(b_i^B) \mid i \in [L]\}$, $T_3 = \{t(b_i^B) \mid i \in [L]\}$ and $S_4$, $T_4$ defined analogously for the auxiliary $C$-agents.
The final two groups of agents are formed by the two types of auxiliary $A$-agents.
We set $S_5$ and $T_5$ to contain the starting and target position of each $b^A_{i,1}$ agent, respectively.
And we define analogously $S_6$ and $T_6$ for the agents~$b^A_{i,2}$.

Let us focus on the auxiliary $B$-agents, i.e., the third group~$A_3$.
Observe that within this group, there is a single possible pairing of positions in $S_3$ and $T_3$ such that each pair is at distance at most~$L$ and it exactly corresponds to the individual agents in the~\MAPFShort instance.
The same holds for agent groups $A_4$, $A_5$ and $A_6$.
Therefore, the movement of auxiliary agents in arbitrary schedule of makespan~$L$ is predetermined and in particular, analogues to Claim~\ref{claim:auxiliary-properties} and~\ref{claim:primary-agents} hold.

The rest of the arguments remain unchanged.
We again conclude that the primary agents can rearrange themselves exactly by a prefix reversal in each round.
Therefore, $\langle G, (A_i, S_i, T_i)_{i=1}^6, L)$ is a yes-instance of \CMAPFShort if and only if $\langle \alpha, \beta, k \rangle$ is a yes-instance of \PN{Binary String Prefix Reversal Distance}.
\end{proof}

\section{Efficient Algorithm for Centralized Networks}

In this section, we prove the main positive result of this paper. 
First, we prove that the \MAPF problem can be solved in polynomial time in cliques; this will be necessary later on. 

Throughout this section, there is a notion of \emph{swap} that is going to play a very important role. A swap is the behavior that happens when two neighboring agents have to exchange positions in one turn according to the schedule. Formally, let $\alpha,\beta$ be two agents and $s,s'$ be two placements of agents. We say that $\alpha,\beta$ are \emph{swapping between $s$ and $s'$} (alternatively swapping in $s,s'$) if $s'(\alpha)=u=s(\beta)$ and $s'(\beta)=v=s(\alpha)$; observe that this can only happen if $uv\in E(G')$. Additionally, given a potential solution $s_1, \ldots s,_m$, we will say that a swap is happening in turn $i$ if there exits a pair of agents that are swapping between the placements $s_{i-1}$ and $s_i$. We stress here that we introduce this notion here purely for the sake of exposition, as swapping is not allowed in the version of \MAPF that we consider.

\begin{theorem}\label{thm:clique:polynomial}  
    Let $\mathcal{I} = \langle G, A, s_0, t, \ell \rangle$ be an instance of \MAPF{} and $G=(V,E)$ be a clique. We can decide if $\mathcal{I}$ is a yes-instance in polynomial time. Furthermore, if $\ell \geq 2$ and $|V(G)|\ge 4$, then $\mathcal{I}$ is always a yes-instance.
\end{theorem}

\begin{proof}
    Clearly, if  $|V| < 4$ we can compute an optimal solution of $\mathcal{I}$ in polynomial time. Therefore, we may assume that $|V| \ge 4$.

    First, we check if there exists a feasible solution of makespan $1$. For every $i\in[\ell-1]$, we define a \emph{swapping pair} of agents to be any pair of agents $\alpha, \beta \in A$, with $\alpha \neq \beta$, such that $s_0(\alpha) = t(\beta)$ and $s_0(\beta) = t(\alpha)$. It is easy to see that $\mathcal{I}$ accepts a feasibly solution of makespan $1$ if and only if there exists no swapping pair.
    If such a pair does not exists, then $s_1 =t$ is a feasible solution of makespan $1$.
    Otherwise, there is no feasible solution of makespan $1$.
    
    Next, we focus on the case where there is no feasible solution of makespan $1$.
    Let $p\ge 1$ be the number of swapping pairs that exist and let $(\alpha_i,\beta_i) $, $i \in [p]$, be these pairs. We consider two cases, according to whether $p\ge 2$ or not.

    \medskip
    \noindent\textbf{Case 1.} ($p\ge 2$): Here, we can compute a feasible solution of makespan $2$ as follows.\todo{DK: Again, I believe that we could add a figure for this case.}
    \begin{itemize}
        \item For all $i\in [p-1]$, we set $s_1(\alpha_i) = s_0(\alpha_i)$,
        \item we set $s_1(\alpha_p) = s_0(\beta_1)$,
        \item for all $i \in [p-1]$ we set $s_1(\beta_i)= s_0(\beta_{i+1})$,
        \item we set $s_1(\beta_p) = s_0(\alpha_p)$ and
        \item for every other agent $\alpha\in A$, we set $s_1(\alpha) = s_0(\alpha)$.
    \end{itemize}
    Then, for every agent $\alpha\in A$, we set $s_2(\alpha)=t(\alpha)$. The schedule $s_1,s_2$ is a feasible solution of $\mathcal{I}$ which does not contain swapping pairs. Indeed, the positions of $s_1$ are achieved by having $i+1\ge 3$ agents moving in a cycle.

    \medskip
    \noindent\textbf{Case 2.} ($p=1$): If $|V| > |A|$, then there exists at least one vertex $v\neq s_0(\alpha)$, for all $\alpha\in A$.
    We set $s_1 (\alpha_1)=v$ and $s_1(\alpha) = s_0(\alpha)$ for all $\alpha \in A\setminus \{ \alpha_1\}$.
    Then $s_1,s_2$, where $s_2(\alpha) =t(\alpha)$ for all $\alpha \in A$, is a feasible solution which does not contain any swapping pair.

    Therefore, we may assume that $|V|=|A|$.
    Since $|V|\geq 4$, there exist at least two agents in $A\setminus\{\alpha_1,\beta_1\}$. We select a pair of agents $\alpha_0, \beta_0$ such that, either $s_0(\alpha_0) = t(\alpha_0) $ and $s_0(\beta_0) = t(\beta_0) $, or $s_0(\alpha_0) \neq t(\alpha_0) $ and $s_0(\beta_0) \neq t(\beta_0) $. Notice that such a pair $(\alpha,\beta)$ always exists.
    We deal with each case of the $\alpha_0$ and $\beta_0$ pair separately.

    \medskip
    \noindent\textbf{Case 2 (a).} ($s_0(\alpha_0) = t(\alpha_0) $ and $s_0(\beta_0) = t(\beta_0) $): Consider the following schedule:
    \begin{itemize}
        \item $s_1(\alpha_1) = s_0(\alpha_0)$, $s_1(\alpha_0) = s_0(\beta_0)$ and $s_1(\beta_0 ) = s_0(\alpha_1)$.
        \item For all $\alpha \in A\setminus \{\alpha_0, \alpha_1, \beta_0 \}$, we set $s_1(\alpha)=s_0(\alpha)$.
    \end{itemize}
    We claim that $s_1,s_2$, where $s_2(\alpha) =t(\alpha)$ for all $\alpha \in A$, is a feasible solution which does not contain any swapping pair. Since we are in a clique we just need to prove that there are no swaps between $s_1$ and $s_2$. 
    
    Assume that we have a swapping pair $\alpha,\beta$ and that the swap happens between $s_0$ and $s_1$. Notice that, $\alpha,\beta \in \{ \alpha_0, \alpha_1, \beta_0 \}$. 
    Indeed, if $\alpha \in A\setminus \{\alpha_0, \alpha_1, \beta_0 \}$,  then $s_1(\alpha) = s_0(\alpha) \neq s_0(\beta)$, for any $\beta \in A$; therefore there is no swap. 
    Also, if $\alpha,\beta \in \{\alpha_0, \alpha_1, \beta_0 \}$, then again there is no swap by the construction of $s_1$. Indeed, the agents $\alpha_0,\alpha_1,\beta_0$ are moving in a cycle of length $3$.

    Next, assume we have a swapping pair $\alpha,\beta$ and that the swap happens between $s_1$ and $s_2$.
    Recall that we have assumed that we have exactly one pair of swapping agents, the pair $\alpha_1, \beta_1$. 
    \begin{claim}
        $\alpha,\beta \in \{\alpha_0,\alpha_1,\beta_0,\beta_1\}$
    \end{claim}
    \begin{proofclaim}
        Assume that  $\alpha \notin \{\alpha_0,\alpha_1,\beta_0,\beta_1\}$. 
        By the selection of $\{\alpha_0,\alpha_1,\beta_0,\beta_1\}$, it follows that $t(\alpha) \notin s_0(\{\alpha_0,\alpha_1,\beta_0,\beta_1\}) = s_1(\{\alpha_0,\alpha_1,\beta_0,\beta_1\})$.
        Since we have assumed that $\alpha$ and $\beta$ are a swapping pair, we have that $\beta \notin \{\alpha_0,\alpha_1,\beta_0,\beta_1\}$, as $s_1(\beta) = t(\alpha) \notin s_1(\{\alpha_0,\alpha_1,\beta_0,\beta_1\})$.
        It follows by the construction of $s_1$ that $s_1(\alpha) = s_0(\alpha)$ and $s_1(\beta) = s_0(\beta)$. Therefore, $\alpha$ and $\beta$ are swapping between $s_1$ and $s_2$ if and only if they are swapping between $s_0$ and $t$ (as $s_2=t$). This is a contradiction as only the pair $(\alpha_1,\beta_1)$ is swapping between $s_0$ and $t$. 
    \end{proofclaim}
    
    We proceed with $\alpha,\beta \in \{\alpha_0,\alpha_1,\beta_0,\beta_1\}$.
    By the selection of $\{\alpha_0,\alpha_1,\beta_0,\beta_1\}$, we have that $t(\alpha_0) = s_0(\alpha_0)$, $t(\beta_0)= s_0(\beta_0)$, $t(\alpha_1) = s_0(\beta_1)$ and $t(\beta_1) = s_0(\alpha_1)$. Also, by the construction of $s_1$, we have that $s_1(\beta_1) = s_0(\beta_1)$, $s_1(\alpha_1) = s_0(\alpha_0)$, $s_1(\alpha_0) = s_0(\beta_0)$ and $s_1(\beta_0 ) = s_0(\alpha_1)$.
    Therefore, the agents are moving in a cycle of length $4$.

    \medskip
    \noindent\textbf{Case 2 (b).} ($s_0(\alpha_0) \neq t(\alpha_0) $ and $s_0(\beta_0) \neq t(\beta_0) $): In this case we have that at least one of the following is true: $t(\alpha_0) \neq s_0(\beta_0)$ or $t(\beta_0) \neq s_0(\alpha_0)$. Indeed, if both are true, there exists a second swapping pair, contradicting that $p=1$. W.l.o.g., assume that $t(\alpha_0) \neq s_0(\beta_0)$.
    Consider the following schedule:
    \begin{itemize}
        \item $s_1(\alpha_1) = s_0(\alpha_0)$, $s_1(\beta_1) = s_0(\alpha_1)$ and $s_1(\alpha _0 ) = s_0(\beta_1)$.
        \item For all $\alpha \in A\setminus \{ \alpha_0, \alpha_1, \beta_1 \}$ we are setting $s_1(\alpha)=s_0(\alpha)$.
    \end{itemize}
    Since we are in a clique we just need to prove that there are no swaps between $s_1$ and $s_2$.

    Assume that two agents $\alpha,\beta$ are a swapping pair and that the swap happens between $s_0$ and $s_1$. Notice that $\alpha,\beta \in \{\alpha_0, \alpha_1, \beta_1 \}$. 
    Indeed, if $\alpha \in A\setminus \{\alpha_0, \alpha_1, \beta_1 \}$, then $s_1(\alpha) = s_0(\alpha) \neq s_0(\beta)$, for any $\beta \in A$; therefore, there is no swap. 
    Since $\alpha,\beta \in \{\alpha_0, \alpha_1, \beta_1 \}$, then again there is no swap by construction of $s_1$. Indeed, the agents $\alpha_0, \alpha_1, \beta_1$ are moving in a cycle of length $3$.

    Next, we consider the case where $\alpha,\beta$ are swapping between $s_1$ and $s_2$.
    Recall that $s_1(\{\alpha_0, \alpha_1, \beta_1\}) = s_0(\{\alpha_0, \alpha_1, \beta_1\})$ and $s_1(\gamma) = s_0(\gamma)$, for all $\gamma \notin \{\alpha_0, \alpha_1, \beta_1\}$.
    We will prove that if both $\alpha$ and $\beta$ do not belong in $ \{\alpha_0, \alpha_1, \beta_0, \beta_1\}$, then $\alpha$ and $\beta$ are not swapping.

    Indeed, since $s_1(\alpha) = s_0(\alpha)$, $s_2(\alpha) = t(\alpha)$, $s_1(\beta) = s_0(\beta)$ and $s_2(\beta) = t(\beta)$, we have that $\alpha, \beta$ are swapping between $s_1$ and $s_2$ if and only if they are swapping in $s_0,t$. This is a contradiction as the only pair that swaps in $s_0,t$ is $\alpha_1,\beta_1$.

    Therefore we can assume, w.l.o.g., that $\alpha \in \{\alpha_0, \alpha_1, \beta_0, \beta_1\}$. We proceed by considering each case separately.

    \medskip
    \noindent\textbf{Case $\boldsymbol{\alpha =\alpha_0}$.} We have that $t (\alpha_0) \notin s_1(\{\alpha_1,\beta_0,\beta_1\})$. Indeed, $s_1(\alpha_1)= s_0(\alpha_0) \neq t(\alpha_0)$ (by the selection of $\alpha_0$), $s_1(\beta_0)= s_0(\beta_0) \neq t(\alpha_0)$ (by the selection of $\alpha_0$) and $s_1(\beta_1)= s_0(\alpha_1) = t(\beta_1) \neq t(\alpha_0)$. Since we have assumed that $\alpha$ and $\beta$ swap between $s_1$ and $s_2$, we have that $s_1(\beta) = t(\alpha)$ thus  $\beta \notin \{\alpha_1,\beta_0,\beta_1\} $. Additionally, because we have assumed that $\alpha$ and $\beta$ also swap between $s_1$ and $s_2$, we have that $t(\beta) = s_1(\alpha)=s_1(\alpha_0) = t(\alpha_1)$. Therefore, $\beta=\alpha_1$. This is a contradiction.

    \medskip
    \noindent\textbf{Case $\boldsymbol{\alpha =\beta_0}$.} Here we need to consider two sub-cases, whether $t(\beta_0)=s_0(\alpha_0)$ or $t(\beta_0) \notin s_0(\{\alpha_0,\alpha_1,\beta_1\})$. In the first case, we have no swap because $s_1(\alpha_1) = s_0(\alpha_0) = t(\beta)$ but $t(\alpha_1)=s_0(\beta_1) \neq s_0(\beta_0) = s_1(\beta_0)$. In the second case, $t(\beta_0) = s_1(\gamma)$ and $\beta = \gamma  \notin \{\alpha_0,\alpha_1,\beta_1\}$. As both $\alpha, \beta \notin \{\alpha_0,\alpha_1,\beta_1\}$, we have that  $s_1(\alpha) = s_0(\alpha)$ and $s_1(\beta) = s_0(\beta)$. Therefore, $\alpha$ and  $\beta$ are swapping between $s_1$ and $s_2$ if and only if they are swapping in $s_0,t$. This is a contradiction as the only pair that swaps in $s_0,t$ is $\alpha_1,\beta_1$.

    \medskip
    \noindent\textbf{Case $\boldsymbol{\alpha =\alpha_1}$.} Notice that $t(\alpha_1) = s_0(\beta_1) = s_1(\alpha_0)$. Since $t(\alpha_0) \neq s_0(\alpha_0) =s_1(\alpha_1) $, there is non swap.

    \medskip
    \noindent\textbf{Case $\boldsymbol{\alpha =\beta_1}$.} Notice that $t(\beta_1) = s_1(\beta_1)$. Therefore, there is no swap including $\beta_1$. 

    Therefore, $s_1,s_2$ does not include any swaps.
    This completes the proof of the theorem.
\end{proof}

\subsection{Efficient algorithm for graphs with bounded distance to clique}
This subsection is dedicated to the proof of Theorem~\ref{thm:fpt-distance-to-clique}. Since this proof is quite involved, we break it down in the upcoming lemmas.
Before we start, note that we can assume that the makespan of any optimal solution is strictly larger than $2$, as instances with for $\ell \leq 2$ are easily solvable in polynomial time through matching.

Instead of just deciding whether the given instance of \MAPF{} is a yes-instance or not, we will compute an optimal solution of the given instance (i.e., a feasible solution of minimum makespan) if such a feasible solution exists.
Basically, we will deal with the optimization version instead of the decision version of the problem; thus, we drop the target makespan $\ell$ from the input instance.

Let $\mathcal{I} = \langle G, A, s_0, t \rangle$ be an instance of \MAPF{}. Recall that $\dc(G)$ is the minimum number of vertices that we need to delete from $G$ such that the resulting graph is a clique. That is, there exists a partition $(M,Q)$ of $V(G)$ such that $G[Q]$ is a clique and $|M| = \dc(G)$ and for any other partition $(M',Q')$, either $G[Q']$ is not a clique or $|M'|\ge \dc(G)$. Note also that we can compute $M$ in \FPT{} time parameterized by $\dc(G)$.
In the following, we assume that we have such a partition $(M,Q)$ and will use $\dc$ to refer to the number $\dc(G) = |M|$.

Let us break down the steps required to take during this proof: 
\begin{enumerate}
    \item We prove that if $\mathcal{I}$ admits a feasible solution, then the makespan of any optimal solution is bounded by a computable function of $\dc$; namely, it is bounded by $3(2\dc + 2)^{\dc} + 2$. This is done in subsection~\ref{step1}
    \item We create a new instance $\mathcal{K}=\langle G', A', s'_0, t' \rangle$ of \MAPF{} such that: 
    \begin{itemize}
        \item $G'$ is an induced subgraph of $G$ with a number of vertices that is bounded by a computable function of $\dc$,
        \item $V(G') \supseteq M$ and
        \item $A'$ is an appropriately selected subset of $A$.
    \end{itemize}
    This is done in subsection~\ref{step2}
    \item We prove that if $\mathcal{I}$ admits a feasible solution $s_1,\ldots,s_m$, then $\mathcal{K}$ admits a feasible solution $s'_1,\ldots,s'_m$ such that $|s'_i(A') \cap M| \ge |A|- |Q|$. These solutions are not necessarily optimal. This is done in subsection~\ref{step3}\todo{DK: quantify $i$}
    \item We compute an optimal solution $s'_1,\ldots,s'_{m'}$ of $\mathcal{K}$ such that $|s'_i(A') \cap M| \ge |A|- |Q|$ (if such a solution exists). This is done in \FPT{} time parameterized by \dc{} as $|V(G')|$ is bounded by a computable function of~$\dc$. This is proven in subsection~\ref{step4}
    \item We extend $s'_1,\ldots,s'_{m'}$ into a feasible solution $s_1,\ldots,s_{m'}$ of $\mathcal{I}$. To do so, we create consecutive matchings between the positions of the agents $A\setminus A'$ in the turn $i \in [m'-1]\cup \{0\}$ and $Q\setminus s'_{i+1}(A')$, i.e., the empty positions in the clique in the turn $i+1$. This is done in subsection~\ref{step5}
\end{enumerate}
It is easy to see that the algorithm that follows the above steps will return an optimal solution of $\mathcal{I}$ if it exists. Otherwise, we can conclude that there is no feasible solution for $\mathcal{I}$. Thus, we obtain the main theorem of this subsection.
\thmFptDC*

We now deal with each of the above-mentioned steps. 

\subsubsection{Step 1.}\label{step1} 
We will show that if there exists a feasible solution $s_1,\ldots, s_m$ of $\mathcal{I}$, then there exists a feasible solution $s'_1,\ldots, s'_{f(|M|)}$ of $\mathcal{I}$, where $f\colon \mathbb{N} \rightarrow \mathbb{N}$ is a computable function. Before we do this, we need to introduce some new notation.

We first define a generalization of \MAPF, called \PAMAPF, as follows. The input of the (optimization version of) \PAMAPF{} problem consists of a graph $G=(V,E)$, two sets of agents $A$ and $B$, two functions $s_0\colon A \cup B \rightarrow V$ and $t\colon A \rightarrow V$ and a subset $T$ of $V$ such that $|T| = |B|$. F
For any $a , a' \in A$ and $b \in B$ where $a \neq a'$, we have that $s_0(a) \neq s_0(a')$, $s_0(a) \neq s_0(b)$, $t(a) \neq t(a')$ and $t(a),t(a') \notin T$.
A sequence $s_1,\ldots,s_m$ is a feasible solution of an instance $\langle G, A, B, s_0, t, T \rangle$ of \PAMAPF{} if
\begin{enumerate}
    \item for all $i\in [m]$ we have that $s_i\colon A \cup B \rightarrow V$ are injective,
    \item $s_i(a)$ is a neighbor of $s_{i-1}(a)$ in $G$, for every agent $a \in A \cup B$ and $i \in [1,m]$,
    \item for all $a \in A$ we have that $s_m(a) = t(a)$ and
    \item for the set $S_m = \{s_{m}(b) \mid b \in B \} $ we have that $S_m = T$.
\end{enumerate}
A feasible solution $s_1,\ldots,s_m$ has makespan~$m$. The goal is to find a feasible solution of minimum makespan.

Now, starting from the instance $\mathcal{I}$ of \MAPF{}, we create an instance $\mathcal{J} = \langle G, A', B, s'_0, t', T \rangle$ of \PAMAPF{} as follows. Let $X = \{a \in A \mid s_0(a) \in M \text{ or } t(a) \in M \}$. If $|A \setminus X| < 4$, we set:
\begin{itemize}
    \item $A' = A$, 
    \item $B= \emptyset$,
    \item $t'= t$ and 
    \item $T = \emptyset$.
\end{itemize}
Otherwise, we set:
\begin{itemize}
    \item $A' = X$, 
    \item $B= A\setminus A'$,
    \item $t'\colon A'\rightarrow V$ such that $t'(a) = t(a)$ for all $a \in A'$,
    \item $T = \{ t(b) \mid b \in B \}$.
\end{itemize}

In the first case (i.e., $B= \emptyset$), any feasible solution of $\mathcal{I}$ is a feasible solution of $\mathcal{J}$ and \emph{vice versa}.
Therefore, we may only focus on the second case. Notice that, again, a feasible solution $s_1,\ldots,s_m$  for the instance $\langle G, A, s_0, t, \rangle$ of \MAPF{} is a feasible solution for the instance $\langle G, A', B, s'_0, t', T  \rangle$ of \PAMAPF{}.
Also, any feasible solution $s'_1,\ldots,s'_m$ for the instance $\langle G, A', B, s'_0, t', T  \rangle$ of \PAMAPF{} can be extended into a feasible solution $s'_1,\ldots,s'_{m+2}$ of the instance $\mathcal{I}$ of \MAPF{}. The last statement becomes clear once we observe that we just need to compute a feasible solution of $\langle G[T], B, s'_m, t'' \rangle$ of \MAPF{}, where $t''\colon B \rightarrow T$ is such that $t''(a) = t(a)$ for all $a \in B$.
Since $G[T]$ is a clique and $|B| = |A\setminus A'| \ge 4 $, we know that there exists a feasible solution $s''_1, s''_2$ of $\langle G[T], B, s'_m, t'' \rangle$ (see Theorem~\ref{thm:clique:polynomial}). Then, we set:
\begin{itemize}
    \item $s'_{m+1}(a) = s''_{1}(a)$ and $ s'_{m+2} = s''_{2}(a)$ for all $a \in B$, and  
    \item $s'_{m+1}(a) = s'_{m+2}(a) = s'_{m}(a)$ for all $ a \in  A'$.
\end{itemize}
We claim that the sequence $s'_1,\ldots,s'_{m+2}$ is a feasible solution of $\mathcal{I}$.
Indeed, for all $a \in A'$, we have that $s'_{m+2}(a) = s'_{m}(a) = t'(a)=t(a)$, and for all $a \in B$, we have that $ s'_{m+2}(a) = s''_{2}(a) = t''(a)=t(a)$.
Also, for all $a \in A$ we have that $s'_i(a)$ is a neighbor of $s'_{i-1}(a)$ for all $i \in [m+2]$, because these positions are decided based on feasible solutions.
Finally, we need to show that for all $i \in [m+2]$ and any pair of agents $a,b \in A$, if $a\neq b$, then $s'_i(a)\neq s'_i(b)$.
For any $i \in [m]$, this holds because we have assumed that $s'_1,\ldots,s'_{m+2}$ is a feasible solution of $\mathcal{J}$. So, we just need to deal with $i\in \{m+1,m+2\}$.
Notice that for every $a,b\in A'$, if $a\neq b$ then $s'_{m+j}(a) = s'_{m}(a) \neq s'_{m}(b) = s'_{m+j}(b)$ for any $j\in \{1,2\}$.
Also, if $a,b\in A'$ and $a\neq b$, then $s'_{m+j}(a) = s''_{j}(a) \neq s''_{j}(b) = s'_{m+j}(b)$ for any $j\in \{1,2\}$ as $s''_1,s''_2$ is a feasible solution of $\langle G[T], B, s'_m, t'' \rangle$.
Finally, for any $a\in A'$ and $b \in B$ we have to observe that $s'_{m+j}(a) = t(a) \notin T$ and $s'_{m+j}(b) \in T$.

We are now ready to show that if there exists a feasible solution for $\mathcal{J}$, then there is also one that has makespan $f(|M|)$, for some computable function $f$.
\begin{lemma}\label{lem:1}
    Let $\mathcal{J}= \langle G, A, B, s_0, t, T \rangle$ be the instance of \PAMAPF{} that was constructed above. Assume that $V(G)$ can be partitioned into $(Q, M)$, such that $G[Q]$ is a clique and $|Q| \ge 4$.
    If there exists a feasible solution $s_1,\ldots,s_m$ of $\mathcal{J}$, then there also exists a feasible solution $s'_1,\ldots,s'_{m'}$ of $\mathcal{J}$, where $m'\le 3(|A|+2)^{|M|}$.
\end{lemma}

\begin{proof}
    In the case where $|A|+|B|\le |Q|$, there exists a feasible solution $s_1,\ldots,s_{O(|M|^2)}$. Indeed, we can first move all the agents to the vertices of $Q$, then move the agents which have terminals in $M$ to their terminals and then deal with the agents in the clique. Notice that the first two steps can be done in $O(|M|^2)$ turns (even if we move the agents one by one) and we need at most two turns for the third.

    We now deal with the case where $|A|+|B| > |Q|$.
    Let $s\colon A\cup B \rightarrow V(G)$ be a placement of the agents in the graph. We will say that two placements $s$, $s'$ have the same \emph{placement type} if there exist $A' \subseteq A$ and $B'\subseteq B$ such that:
    \begin{itemize}
        \item $s(\alpha)= s'(\alpha) =v \in M$, for all $\alpha \in A\setminus A'$,
        \item $s(B\setminus B') = s'(B\setminus B') \subseteq M$ and
        \item $s(\alpha) \in Q$ for all $\alpha \in A'\cup B'$.
    \end{itemize}
    Notice that the placement types are defined based on how the agents of $A$ and $B$ are placed in $M$. In order to count the number of different potential placement types, we need to consider the options of each vertex $v \in M$. First, we have $|A|$ options, one for each agent of $A$ that may occupy $v$. If $v$ is not occupied by any agent of $A$, then~$v$ is either empty or being occupied by an agent of $B$. In total, we have $|A|+2$ possibilities. Therefore, there exists at most $(|A|+2)^{|M|}$ different placement types.
    
    Consider a feasible solution $\{s_1\ldots,s_m\}$ of $\mathcal{J}$. We will show that there exists a feasible solution $\{s'_1\ldots,s'_{m'}\}$ where any placement type appears no more than three times. If this is true, then $m'$ is upper bounded by $3(|A|+2)^{|M|}$. Assume that in $\{s_1\ldots,s_m\}$ there exists a placement type that appears more than three times and let $0\le p <q \le m$ such that $s_p$ and $s_q$ are the first and last, respectively, appearances of this placement type.
    
    On a high level, our goal is to replace $s_p,\ldots, s_q$ with a new sequence $s'_1,s'_2,s'_3$ such that the occupied vertices are the same and the same holds for the positions of agents of $A$. Then we will use the $s_{q+1},\ldots, s_{m}$ in order to compute a feasible solution.
    To do so, we first find an optimal solution to the following instance of \PAMAPF: $\mathcal{J}'= \langle G[Q], A', B', s, t', T' \rangle$ where $s(\alpha) = s_p (\alpha) $ for all $\alpha \in A' \cup B'$, and $t'$ gives the terminal positions. Since $G[Q]$ is a clique of size at least $4$, we know that there is always a feasible solution of makespan at most $2$ (Theorem~\ref{thm:clique:polynomial}). Let $s''_1,s''_2,s''_3$ be such a feasible solution (where $s''_1 = s$). For all $i \in [3]$, we set $s'_i (\alpha ) = s''_i (\alpha )$ for all $\alpha \in A'\cup B'$ and $s'_i (\alpha ) = s_{p} (\alpha)$ for all $\alpha \in (A\cup B) \setminus (A'\cup B')$. Notice that $s_p = s'_1$. 
    Therefore, the schedule $s_1,\ldots, s_p,s'_2,s'_3$ satisfies the first two conditions for being a feasible solution of $\mathcal{J}$ and does not include any swaps.

    Notice that for any agent $\alpha \in A$, we have that $s'_3(\alpha) = s_q(\alpha)$. However, the same does not hold for the agents of $B$. We will take advantage of the fact that occupied positions of the agents of $B$ in the placements $s_q$ and $s'_3$ are the same. For each agent $\beta \in B $ we need to find the agent of $\beta' \in B$ with which they have ``switched'' positions according to $s_q$ and $s'_3$. In particular, we set $\beta' = s^{-1}_q (s'_3(\beta))$\footnote{Given a placement $s:A \rightarrow V(G)$, we define $s^{-1} : s(A) \rightarrow A$ to be a function such that $s^{-1}(v) = \alpha$ if and only if $s(\alpha) = v$.}. 
    Now, for each $i \in \{q,\ldots,m\} $ and agent $\beta \in B $, we set $s'_i (\beta) = s_i (\beta' ) $. Finally, we set $s'_i(\alpha) = s_i(\alpha)$ for each $i \in \{q,\ldots,m\} $ and agent $\alpha \in A $.
    
    Now we claim that $s_1,\ldots,s_{p-1}, s'_1,s'_2, s'_3,s'_{q+1}, \ldots , s'_m $, where $s'_1 = s_p $ and $s'_3=s'_q$, is a feasible solution of $\mathcal{J}$.

    \begin{claim}
        For any agent $\alpha \in A' \cup B'$ and two placements $s^1 $, $s^2$ that appear consecutively in $s_1,\ldots,s_{p-1}, s'_1,s'_2, s'_3,s'_{q+1}, \ldots , s'_m $, we have that $s^2(\alpha)\in N[s^1(\alpha)]$.
    \end{claim}
    \begin{proofclaim}
        First, observe that is suffices to prove the claim in two case: when $s^1$, $s^2$  appear consecutively in:
        \begin{enumerate}
            \item $s_1,\ldots,s_{p-1},s'_1,s'_2, s'_3$ and 
            \item $s'_3,s'_{q+1}, \ldots , s'_m$.
        \end{enumerate}
        Since $s_1,\ldots,s_{p-1},s'_1,s'_2, s'_3$ satisfies the first condition of being a feasible solution, we can directly conclude that the statement holds in the first case. Therefore, we focus on the second case. 
        
        Consider any agent $\alpha\in A\cup B$ and let $s^1 =s'_i$ and $s^2 =s'_{i+1}$, for some $i \in [m-1]\setminus[q-1]$, be two consecutive placements in $s'_3,s'_{q+1}, \ldots , s'_m$ (recall that $s'_3,s'_{q}$). 
        By the construction of $s'_i$ and $s'_{i+1}$, $i \in [m-1]\setminus[q-1]$, we have that, for any $\alpha \in A\cup B$, either $s'_i(\alpha) = s_i(\alpha)$ and $s'_{i+1}(\alpha) = s_{i+1}(\alpha)$ or $s'_i(\alpha') = s_i(\alpha')$ and $s'_{i+1}(\alpha') = s_{i+1}(\alpha')$. By the feasibility of $s_1,\ldots,s_m$, we have that $s_{i+1}(\beta) \in N[s_{i}(\beta)]$ for any $\beta \in A\cup B$. Therefore, $s^2(\alpha) \in N[s^1(\alpha)]$.  
    \end{proofclaim}

    Now, we will prove the following claim.

    \begin{claim}
        For any placement $s^1$ that appears in $s_1,\ldots,s_{p-1}, s'_1,s'_2, s'_3,s'_{q+1}, \ldots , s'_m $ and any pair of agents $\alpha \neq \beta$, it holds that $s^1(\alpha)\neq s^1 (\beta)$.
    \end{claim}
    \begin{proofclaim}
    Since $s_1,\ldots,s_{p-1},s'_1,s'_2, s'_3$ satisfies the second condition of being a feasible solution, we can directly conclude that the statement holds for any $s^1$ in $s_1,\ldots,s_{p-1},s'_1,s'_2, s'_3$.
    Now, we need to prove the same for the placements in $s'_{q+1}, \ldots , s'_m $. By the construction of $s'_{q+1}, \ldots , s'_m $, there exists a matching  $N :A\cup B \rightarrow A\cup B $ such that $s'_{i}(\alpha) = s_{i}(N(\alpha))$, for all $i \in [m]\setminus[q-1]$. Since $N(\alpha) \neq N(\beta)$ (and thanks to the feasibility of $s_1\ldots,s_m$), we have $s'_i(\alpha) = s_{i}(N(\alpha)) \neq s_i(N(\beta)) = s'_i(\beta)$.         
    \end{proofclaim}
    
    It remains to prove that $s'_m(\alpha) = t(\alpha)$ and $s'_m(\beta)\in T$, for all $\alpha \in A$ and $\beta\in B$.

    By construction, $s'_m(\alpha) = s_m(\alpha)$ for all $\alpha \in A$. Since we have assumed that $s_1,\ldots, s_m$ is a feasible solution of $\mathcal{J}$, we have that $s'_m(\alpha) = s_m(\alpha) = t(\alpha)$ for all $\alpha \in A$. 
    Consider any agent $\beta \in B$. By the construction of $s'_m$, there exists a matching $N\colon B\rightarrow B$ such that $s'_m(\beta) = s_m(N(\beta))$. Since $N(\beta) \in B$ and $s_1,\ldots, s_m$ is a feasible solution of $\mathcal{J}$, we have that $s_m(N(\beta)) \in T$. Therefore, $s'_m(\beta) = s_m(N(\beta)) \in T$. 
    This completes the proof that $s_1,\ldots,s_{p-1}, s'_1,s'_2, s'_3,s'_{q+1}, \ldots , s'_m $ is a feasible solution of $\mathcal{J}$.

    The previous process can be repeated until we have no position type that appears more than three times.
\end{proof}

    As a direct implication of the previous lemma, we get that if the instance $\mathcal{I}$ of \MAPF{} has a feasible solution, then it also has a feasible solution of makespan $3(2|M|+2)^{|M|}+2$. To do this, it suffices to apply the previous lemma on the instance $\mathcal{J}$ of \PAMAPF{} we created from $\mathcal{I}$. Thus, we have achieved the goal of Step $1$.
    
    \subsubsection{Step $2$.}\label{step2} 
    Recall that in this step we will start from the instance $\mathcal{I} = \langle G, A, s_0, t \rangle$ of \MAPF{} and construct a new instance $\mathcal{K}=\langle G', A', s'_0, t' \rangle$ of \MAPF{} such that: 
    \begin{itemize}
        \item $G'$ is an induced subgraph of $G$ with a number of vertices that is bounded by a computable function of $\dc$,
        \item $V(G') \supseteq M$ and
        \item $A'$ is an appropriately selected subset of $A$.
    \end{itemize}
    Before we describe the construction of $\mathcal{K}$ we need to give some additional definitions. Consider $s_1,\ldots, s_m$, a solution of $\mathcal{I}$ of minimum makespan. If such a solution exists, then $m\le 3(2|M|+2)^{|M|}+2$.
    Furthermore, the number of different agents that occupy a vertex in $M$ in any turn is at most $|M|\big(3(2|M|+2)^{|M|}+3\big)\le (10|M|)^{|M|+1}$ (we also count turn $0$). 
    Let $\kappa= (10|M|)^{|M|+1}$ and notice that $\kappa $ is bounded by a computable function of $|M|$. 

    We define \emph{vertex types} and \emph{agent types}. First, we partition $Q$ into the sets $Q_\tau$, for $\tau \in [2^{|M|}]$, such that, for any $\tau \in [2^{|M|}]$, two vertices $v,u \in Q_\tau$ if and only if $N_G[v] = N_G[u]$. We will say that two vertices $w,y \in Q$ are of the same vertex type if $w,y \in Q_\tau$ for a $\tau \in [2^{|M|}]$. Next, we define agent types based on the vertex types of their starting and terminal positions. We set $A_0\subseteq A$ as the set of agents $\alpha \in A_0$ if and only if $s_0(\alpha) \in M$ or $t(\alpha) \in M$. Notice that all the other agents have starting and terminal positions in $Q$. We set $A_{\sigma, \tau}$, where $(\sigma, \tau) \in [2^{|M|}] \times [2^{|M|}]$, to be the set of agents $\{\alpha \in A \mid s_0(\alpha)\in Q_\sigma \text{ and } t(\alpha)\in Q_\tau \}$. Notice that any agent $\alpha \in A$ belongs in exactly one of the $A_0$, $A_{\sigma,\tau}$, where $(\sigma, \tau) \in [2^{|M|}] \times [2^{|M|}]$. We will say that two agents $\alpha,\beta$ are of the same agent type if $\alpha,\beta \in A_0$ or $\alpha,\beta \in A_{\sigma,\tau}$, for a vertex type $(\sigma, \tau) \in [2^{|M|}] \times [2^{|M|}]$.

    In the construction of $\mathcal{K}$ we have to treat carefully the definition of the agent set $A'$ that will be considered. This is done inductively. We first construct a sequence of subsets of agents $A^0,A^1,\ldots,A^p$ where $p \le 2^{|M|}$ and $|A^p|$ is bounded by a function of $|M|$.
    
    We start with the description of $A^0$. From each agent type $(\sigma, \tau) \in [2^{|M|}] \times [2^{|M|}]$ we select (arbitrarily) $\min\{ \kappa, |A_{\sigma, \tau}| \}$ agents. Let $A^0_{\sigma, \tau}$ be this set.
    The set $A^0$ is defined as the set $A_0 \cup \bigcup_{(\sigma, \tau) \in [2^{|M|}] \times [2^{|M|}]} A_{\sigma, \tau}$. Notice that $|A^0| \le 2|M|+ 2^{|M|} (10|M|)^{|M|+1} \le (20|M|)^{|M|+1} $; this will be used later in order to provide an upper bound the size of $A^p$.

    Next, assume that we have constructed the set $A^i$ for some $i \in [0,p-1]$.
    For each $\tau \in [2^{|M|}]$, we define a set of agents $X^i_\tau$ as follows:
    \begin{itemize}
        \item If $|Q_\tau| > 3|A^i|$, then $X^i_\tau=\emptyset$, otherwise 
        \item $X^i_\tau = \bigcup_{\sigma \in [2^{|M|}]} (A_{\sigma,\tau} \cup A_{\tau, \sigma})$.
    \end{itemize} 
    Now, if for all $\tau \in [2^{|M|}]$ it holds that $X^i_\tau \subseteq A^i$, we stop and we have constructed the sequence.
    Otherwise we select one $\tau\in [2^{|M|}]$ such that $X^i_\tau \setminus A^i \neq \emptyset$, and we set $A^{i+1} = X^i_{\tau} \cup A^i$.

    We need to show that $p \le 2^{|M|}$ and $|A^p|$ is bounded by a function of $|M|$.
    Notice that if the sequence stops with $A^0$, then this is true. Therefore, we first assume that $p \in [2^{|M|}]$.
    
    \begin{claim}
        Assuming that we have constructed the set $A^q$ for an integer $q \in [p]$ then:
        \begin{itemize}
            \item there exist at least $q$ vertex types $\tau \in [2^{|M|}]$ such that $X^{q-1}_\tau \neq \emptyset$ and $X^{q-1}_\tau \subseteq A^q$ and
            \item $|A^q| \le 7 |A^{q-1}| $.
        \end{itemize}
    \end{claim}

    \begin{proofclaim}
    The proof is by induction on $q$.

    \smallskip
    \noindent\textbf{Base of the induction $\boldsymbol{(q=1)}$:} We are constructing $A^1$ only in the case where there exists $\tau \in [2^{|M|}]$ such that $X^{0}_\tau \neq \emptyset$ and $X^{0}_\tau \cap  A^{q-1} \neq \emptyset$. Then, by construction, $A^{1} \supseteq X^{0}_\tau$. Also by construction, $X^{0}_\tau \le 2 |Q_\tau|\le 6 |A^{0}|$, for all $\tau \in [2^{|M|}]$. Since we have added exactly one set $X^{0}_\tau$, i.e., $A^1 = X^0_\tau \cup A^0$, we have $|A^1| \le |A^0| +  6 |A^0| \le 7 |A^0| $.

    \smallskip
    \noindent\textbf{Induction Hypothesis ($\boldsymbol{1\le i <q}$):} We assume that the claim is true for an $1 \le i <q$.

    \smallskip
    \noindent\textbf{Induction Step ($\boldsymbol{i+1}$):} 
    Notice that there exists at least $i$ vertex types $\tau_1,\ldots,\tau_i$ such that $X^{i-1}_{\tau_j} \neq \emptyset$ and $X^{i-1}_{\tau_j} \subseteq A^{i}\subseteq A^{i+1}$.
    Also, since $X^{i-1}_{\tau_j} \neq \emptyset$, we have that $|Q_{\tau}| \le 2|A^i| \le 2|A^{i+1}|$.
    Therefore, $X^{i}_{\tau_j}=X^{i-1}_{\tau_j} \subseteq A^{i+1}$. 
    Assume that there is no other $\sigma \in [2^{|M|}]$ such that $X^{i}_\sigma \neq \emptyset $ and $ X^{i}_\sigma \subseteq A^{i+1}$.
    Then, $\bigcup_{\tau} X^i_\tau = \bigcup_{\tau} X^{i-1}_\tau \subseteq A^i$ and the sequence has ended before constructing $A^{i+1}$. This contradiction concludes the induction.

    For computing the size of $A^{i+1}$, it suffices to observe that, to construct it, we are selecting one of the previous agent types and setting $A^{i+1}  = X^i_\tau \cup A^i$. Since $|X^i_\tau| \le 2|Q_\tau| \le 6|A^{i}|$, we have that $|A^{i+1}| \le 7|A^i| $.\todo{DK: Why not 6 $\to$ 4? Since $|Q_{\tau}| \le 2|A^i| \le 2|A^{i+1}|$.}
    \end{proofclaim}

    Now we prove that $p\le 2^{|M|}$. Assume that at some point we have constructed the set $A^{2^{|M|}}$. By the previous claim, for any vertex type $\tau \in [2^{|M|}]$, we have that $X^{2^{|M|}-1}_{\tau} \neq \emptyset$ and $X^{2^{|M|}-1}_{\tau} \subseteq A^{2^{|M|}}$. Therefore, $\bigcup_{\tau \in [2^{|M|}]} X^{2^{|M|} - 1}_\tau \subseteq A^{2^{|M|}}$. This means that the process stops and  $p\le 2^{|M|}$.

    Finally, since $|A^0|$ is upper bounded by a computable function of $|M|$, $|A^i| < 7 |A^{i-1}|$ and $p \le 2^{|M|}$, we conclude that $|A^p|$ is upper bounded by a computable function of $|M|$. In particular,
    \[
        |A^p| \le 7^{2^{|M|}}|A^0| \le 7^{2^{|M|}} (20|M|)^{|M|+1}< 2^{2^{|M|+3}+ 2|M| \log |M|} < 2^{2^{|M|+4}} \,.
    \]

    We are ready to define the set $A'$ of agents that we will consider in $\mathcal{K}$. If $|A|\le \max\{ 2|A^p|, 100\}$, we set $A' = A$; otherwise $A' = A^p$. Hereafter, let $A'_{\sigma,\tau}=A'\cap A_{\sigma,\tau}$.

    Finally, we create the instance $ \mathcal{K} = \langle G[U], A', s'_0, t' \rangle$ where:
    \begin{itemize}
        \item $U = \bigcap_{\tau \in [2^{|M|}]} U_\tau \cup M $, where $U_\tau  \subseteq Q_\tau $ is such that
        \begin{itemize}
        \item if $|Q_\tau| \le 3|A'|$, then $U_\tau = Q_\tau $;
        \item otherwise, $U_\tau \supseteq (s_0(A') \cup t(A')) \cap Q_\tau$ and $|U_\tau| = 3|A'|$.
        \end{itemize}
        Notice that the selection of the vertices of $U_\tau \setminus (s_0(A') \cup t(A'))$ is arbitrary.
        \item The set $A'$ is defined as explained before.
        \item $s_0'\colon A' \rightarrow U$ is such that $s_0'(\alpha) = s_0(\alpha)$ for all $\alpha \in A'$.
        \item $t'\colon A' \rightarrow U$ is such that $t'(\alpha) = t(\alpha)$ for all $\alpha \in A'$.
    \end{itemize}
    The definition of $s_0'$ and $t'$ is consistent as, for all $\alpha \in A'$, both $s_0 (\alpha)$ and $t(\alpha)$ belong to $U$. Also, for each vertex type $\tau \in [2^{|M|}]$ we have that the selected set of vertices $U$ has size at most $3|A'|$ and $|A'| \le 2|A^p|$. It follows that $|U|$ is upper bounded by a computable function of $|M|$. In particular, $|U| \le |M|+2^{|M|} 3 |A'| \le |M|+ 6 \ 2^{|M|} 2^{2^{|M|+4}} < 2^{2^{|M|+5}}$.
    This completes the construction of $\mathcal{K}$ and the Step $2$. 

    \subsubsection{Step $3$.}\label{step3} 
    Here, we prove that if the instance $\mathcal{I}$ admits a feasible solution of makespan $m$, then the instance $\mathcal{K}$ that was constructed in the previous step admits a feasible solution $s'_1,\ldots,s'_{m}$ such that $|s'_i(A') \cap M| \ge |A|-|Q|$, for all $i \in [m]$.

    \begin{lemma} \label{lemma_distance_clique:K_admits_appropriate_solution}
        Given a feasible solution $s_1,\ldots,s_m$ of $\mathcal{I}$, of makespan $m = \kappa \le \big(3(2|M|+2)^{|M|}+2\big)$, we can compute a feasible solution $s'=s'_1,\ldots,s'_m$ of $\mathcal{K}$ such that $|s'_i(A') \cap M| = |s_i(A) \cap M|\ge |A|-|Q|$, for all $i \in [m]$.
    \end{lemma}

    \begin{proof}
        First, let $A^M_{\sigma,\tau} \subseteq A_{\sigma,\tau}$ be the set of agents such that $\alpha \in A^M_{\sigma,\tau}$ if and only if $s_{i-1}(\alpha) \in M$ for some $i \in [m+1]$. Notice that $|A^M_{\sigma,\tau}| \le \min \{ \kappa, |A_{\sigma,\tau}| \} = |A'_{\sigma,\tau}|$ for all $(\sigma,\tau) \in [2^{|M|}] \times [2^{|M|}]$ (since there are at most $ \kappa$ agents that are going to pass through a vertex of $M$ throughout the schedule).
        
        We continue with some definitions. For any $\tau \in [2^{|M|}]$, if $U_\tau \neq Q_\tau$, we will select a set $U'_\tau \subseteq U_\tau \setminus \big( s'_0(A') \cup t'(A') \big)$ such that $|U'_\tau| = |A'|$.
        Notice that $U_\tau \setminus \big( s'_0(A') \cup t'(A') \big)$ has enough vertices for such a selection and that the selection is arbitrary.
        For the rest of the integers $\tau \in [2^{|M|}]$ (i.e., if $U_\tau \neq Q_\tau$), we set $U'_\tau = \emptyset$.
        We will define an (arbitrary) enumeration of the agents $A'$ and an (arbitrary) enumeration of the vertices of each set $U'_\tau$. 
        From now on, we will assume that $A' = \{\alpha'_1,\ldots,\alpha'_{|A'|}\}$ and $U'_\tau = \{u_1^\tau,\ldots, u_{|A'|}^\tau\}$, for each $\tau \in [2^{|M|}]$ where $U'_\tau \neq \emptyset$.
        
        Now, we will show that there exists a matching between agents of $A'$ and $A$ which we will use to create $s'$. The idea behind this is that we can pair an agent $\alpha' \in A'$ with an agent $\alpha \in A$, which allows us to use the trajectory of $\alpha$ in $G$ to compute the trajectory of $\alpha'$ in $G[U]$.

        To do so, we first define a set of agents $B$ as follows.
        Initially $B$ contains exactly the agents of $A_0$.
        For each agent type $(\sigma,\tau) \in [2^{|M|}]\times [2^{|M|}]$, we construct a set $B_{\sigma,\tau} $ such that  $A^M_{\sigma,\tau} \subseteq B_{\sigma,\tau} \subseteq A_{\sigma,\tau} $ and $|B_{\sigma,\tau}| = |A'_\sigma,\tau|$.
        Notice that, by construction, $|A^M_\sigma,\tau|\le |A'_\sigma,\tau|\le |A_\sigma,\tau|$ and, therefore, we have enough agents. Also, the selection of the agents of $B_{\sigma,\tau} \setminus A^M_{\sigma,\tau}$ is arbitrary. We set $B = \bigcup_{(\sigma,\tau) \in [2^{|M|}]\times [2^{|M|}]} B_{\sigma,\tau}  \cup  A_0$.
        
        We now define the matching we will use as a function $N\colon A' \rightarrow B$ such that:
        \begin{itemize}
            \item $|N(A')| = |A'|$ (i.e., $N$ is a matching),
            \item for all $\alpha \in A_0$ we have $N(\alpha') = \alpha' $,
            \item for all $ A'_{\sigma,\tau} = A_{\sigma,\tau}$, for some $(\sigma,\tau) \in [2^{|M|}]\times [2^{|M|}]$, and $\alpha' \in  A'_{\sigma,\tau}$, we have $N(\alpha') = \alpha' $ and
             for all $ A'_{\sigma,\tau} \neq A_{\sigma,\tau}$, we have $N(A'_{\sigma,\tau}) = B_{\sigma,\tau}$.
        \end{itemize}
        Notice that if $ A'_{\sigma,\tau} \neq A_{\sigma,\tau}$, the matching between the vertices $A'_{\sigma,\tau}$ and $B_{\sigma,\tau}$ is arbitrary.
        Also, such a matching between $A'$ and $B$ always exists. 

        \medskip 

        We are now ready to use the matching $N$ to create a feasible solution $s'=s'_1,\ldots,s'_m$ for $\mathcal{K}$, starting from a feasible solution $s_1,\ldots,s_m$ for $\mathcal{I}$.
        
        Let  $ \alpha_i = N(\alpha'_i) $ for all $\alpha'_i \in A'$.
        To construct $s'$ we proceed as follows:
        \begin{itemize}
            \item First, we set $s'_m(\alpha'_i) = t'(\alpha'_i)$, for each $i \in [|A'|]$.
            \item For each $i \in [m-1]$ and each $p \in |A'|$, if $s_i(\alpha_p) \in M$, then we set $s'_i(\alpha'_p) = s_i(\alpha_p)$. Otherwise, $s_i(\alpha_p) \in Q$.
            \item Then, for each $i \in [m-1]$ and each $p \in |A'|$, let $\tau \in 2^{|M|}$ be the vertex type such that $s'_i(\alpha_p) \in Q_\tau$.
            \begin{itemize}
            \item If $U'_\tau = \emptyset$, then we set $s'_i(\alpha'_p) = s_i(\alpha_p)$.
            \item Otherwise, if $s_i(\alpha_p) \in s_0(A') \cup t(A')$, then $s'_i(\alpha_p) = s_i(\alpha_p)$.
            \item Otherwise, $s'_i(\alpha_p) = u_p^\tau$. 
            \end{itemize}
        \end{itemize}      
        We claim that the schedule $s'$ constructed this way is feasible. To do that, we need to show that $s'$ is collision-free and swap-free. 

        \paragraph*{The schedule $s'$ is collision-free.}
        To start, observe that for each $i\in [0,m]$ and $p \in [|A|]$, we have that either $s'_{i},(\alpha'_p) = s_{i},(\alpha_p)$, or the vertices $s'_{i-1},(\alpha'_p)$ and $s_{i-1},(\alpha_p)$ have the same vertex type $\tau \in [2^{|M|}]$ (i.e., $s'_{i-1},(\alpha'_p), s_{i-1},(\alpha_p) \in Q_{\tau}$, for some $\tau \in [2^{|M|}]$).
        This is indeed clear for all $i \in [m-1]$ by the construction of $s'_i$. To see that this is also true for $i \in \{0,m\}$, recall that $\alpha'_p$ and $\alpha_p$ belong to the same agent type $(\mu, \sigma) \in [2^{|M|}] \times [2^{|M|}]$. That is, $\alpha_p, \alpha'_p \in A_{\mu, \sigma}$ for some $(\mu, \sigma) \in [2^{|M|}] \times [2^{|M|}]$. This means that $s_0(\alpha_p), s'_0(\alpha'_p) \in Q_{\mu}$ and $t(\alpha_p), t'(\alpha'_p) \in Q_{\sigma}$. This covers the case $i=0$.
        For $i=m$ we need to additionally observe that $t(\alpha_p) = s_m(\alpha_p)$ (by the feasibility of $s_1,\ldots,s_m$) and $s'_m(\alpha'_p) = t'(\alpha'_p)$.  

        Thus, for each agent $\alpha'_p\in A'$ and each $i \in [m]$, it holds that $s'_{i}(\alpha'_p) \in N_{G[U]} [s'_{i-1}(\alpha'_p)]$.
        Indeed, since $s'_{i}(\alpha'_p)$ has the same vertex type as $s_{i}(\alpha_p)$, we have that $N_G[s'_{i}(\alpha'_p)] = N_G[s'_{i}(\alpha'_p)]$.
        Similarly, we have that $N_G[s'_{i-1}(\alpha'_p)] = N_G[s'_{i-1}(\alpha'_p)]$.
        Since $s_{i}(\alpha_p) \in N_{G} [s_{i-1}(\alpha_p)]$, we conclude that $s'_{i}(\alpha'_p) \in N_{G} [s'_{i-1}(\alpha'_p)]$. Finally, since $s'_{i}(\alpha'_p) \in U$ and $N_{G[U]} [s'_{i-1}(\alpha'_p)] = U\cap N_{G} [s'_{i-1}(\alpha'_p)]$, we have that $s'_{i}(\alpha'_p) \in N_{G[U]} [s'_{i-1}(\alpha'_p)]$.

        \smallskip 
        
        Next, we prove that in $s'$ there is no turn where two agents occupy the same vertex. Formally, for any $i\in [m]$ and pair of agents $\alpha'_p, \alpha'_q \in A'$, where $p \neq q$, we have that $s'_i(\alpha'_p) \neq s'_i(\alpha'_q)$.
        Indeed, assume that this statement is not true and let $i \in [m]$ and $\alpha'_p, \alpha'_q \in A'$ with $p\neq q$, be such that $s'_i(\alpha'_p) = s'_i(\alpha'_q)=w$.
        Note that the claimed statement does hold for $i=m$. Indeed, $s'_m (\alpha'_p) = t'(\alpha'_p)$ for all $\alpha'_p \in A'$. So, we may assume that $i \in [m-1]$.
        Let $\tau \in [2^{|M|}]$ be the vertex type of $w$. That is, $s'_i(\alpha'_p) = s'_i(\alpha'_q)= w \in Q_{\tau}$. We can assume that $U_\tau = Q_{\tau}$ as, otherwise, by the construction of $s'_i$, we would have that $s'_i(\alpha'_p) = u^\tau_p \neq u^\tau_q = s'_i(\alpha'_q)$ (as $p\neq q$). Since $U_\tau = Q_{\tau}$, it follows by the construction of $s'_i$ that $w=s'_i(\alpha'_p) = s_i(\alpha_p)$ and $w=s'_i(\alpha'_q) =  s_i(\alpha_q)$. This means that $s_i(\alpha_p) = s_i(\alpha_q)$ which is a contradiction to the assumption that $s_1,\ldots,s_m$ is a feasible solution of $\mathcal{I}$.

        \medskip
        
        \paragraph*{The schedule $s'$ is swap-free.} We first claim that if $s_1,\ldots,s_m$ contains no swaps, then the potential swaps of $s'_1,\ldots,s'_m$ may appear only between the turns $s'_0,s'_1$ or $s'_{m-1},s'_m$. Formally, if $s_1,\ldots,s_m$ does not include any swaps then,
        for any $i\in [m-2]$ and agents $\alpha'_p, \alpha'_q \in A'$, where $p \neq q$, we have that $s'_i(\alpha'_p) \neq s'_{i+1}(\alpha'_q)$ or $s'_i(\alpha'_q) \neq s'_{i+1}(\alpha'_p)$. 
        
        Indeed, assume that there exists a turn $i \in [m+1]$ and two agents $\alpha'_p, \alpha'_q \in A'$, with $p\neq q$, such that $s'_{i}(\alpha'_p) = s'_{i+1}(\alpha'_q) = v$ and $s'_{i}(\alpha'_q) = s'_{i+1}(\alpha'_p) = u$. Let $v \in Q_{\tau}$ and $u \in Q_{\sigma}$, for two (not necessarily different) vertex types $\sigma,\tau \in [2^{|M|}]$.
        We claim that $U_\tau =Q_\tau$ and $U_\sigma = Q_\sigma$.
        Assume otherwise and, w.l.o.g., let $U_\tau \neq Q_\tau$. 
        Notice that at least one of $i-1$ and $i$ must be in $[m-1]$.
        Let $j \in \{i-1,i\} \cap [m-1]$. Then either $v = s'_{i-1}(\alpha'_p) = u^\tau_p$ (if $j=i-1$) or $v = s'_i(\alpha'_q) =u^\tau_q$ (if $j=i$). 
        In the first case, we have  $s'_i(\alpha'_q) = v = u^\tau_p$ which is a contradiction since (by the construction of $s'_i$) $\alpha'_q$ either occupies  $u^\tau_q \neq u^\tau_p$ or $t'(\alpha'_q) \notin U_\tau$ (while $u^\tau_p \in U_\tau$).
        In the second case, we have that $s'_{i-1}(\alpha'_p) = v = u^\tau_q$, which is a contradiction since, by construction of $s'_{i}$, $\alpha'_p$ either occupies  $u^\tau_p \neq u^\tau_q$ or $s'_0(\alpha'_q) \notin U_\tau$ (while $u^\tau_q \in U_\tau$).
        So, it holds that $U_\tau =Q_\tau$ and $U_\sigma = Q_\sigma$. Consider two turns $0<i-1<i<m$. By the construction of $s'_{i-1}$ and $s'_{i}$, we have that $s_{i-1}(\alpha_p) = s'_{i-1}(\alpha'_p) = s'_i(\alpha'_q) = s_i(\alpha_q)$ and $s_{i-1}(\alpha_q) = s'_{i-1}(\alpha'_q) = s'_i(\alpha'_p) = s_i(\alpha_p)$. Therefore, this swap appears between $s'_{i-1}$ and $s'_i$ only if there exists a swap between $s_{i-1}$ and $s_i$.

        \smallskip
        
        It remains to show that there are no swaps between $s'_0,s'_1$ or $s'_{m-1},s'_m$.
        Assume that this is not true, and let $\alpha'_p$ and $\alpha'_q$, $p\neq q$, be two agents such that $s'_{i-1}(\alpha'_p) = s'_{i}(\alpha'_q) = w$ and $s'_{i}(\alpha'_p) = s'_{i-1}(\alpha'_q) = w'$ for some $i \in \{1,m\}$. First observe that $w \notin U_{\tau}$ such that $U_\tau \neq Q_\tau$. Indeed, assuming otherwise means that $ s'_{i}(\alpha'_q) = s'_{1}(\alpha'_q) = u^\tau_q$ or $s'_{i-1}(\alpha'_p) =  s'_{m-1}(\alpha'_p) =  u^\tau_p$.
        By the selection of $U'_\tau$, either $w = u^\tau_q \neq s'_0(\alpha'_p) = w$ or $w = u^\tau_p \neq t'(\alpha_q) = s'_m(\alpha'_q)= w$. This is a contradiction. Using the similar arguments we can prove that $w' \notin U_{\tau}$ such that $U_\tau \neq Q_\tau$.
        Therefore, $w,w' \in  M \cup \{ U_\tau \mid \tau \in [2^{|M|}] \text{ and } U_\tau=Q_\tau \}$. 

        The final ingredient we need in order to show that there are no swaps in $s'_1,\dots,s'_m$ is to prove that:
        
        \begin{claim}
            $N(\alpha'_p) = \alpha'_p $ and $N(\alpha'_q) = \alpha'_q $.
        \end{claim} 
        \begin{proofclaim}
            Notice that if $A'=A$ then this is trivially true. Therefore, we may assume that $A'=A^p$.

            First, if $w,w' \in M$, then either $s'_0(\alpha'_p), s'_0(\alpha'_q) \in M$ or $s'_m(\alpha'_q), s'_m(\alpha'_q) \in M$ (depending on whether the swap happens in the first or the last turn). Therefore, $\alpha'_p, \alpha'_q \in A_0$ and thus, $N(\alpha'_p) = \alpha'_p$ and $N(\alpha'_q) = \alpha'_q$ (by the selection of $N$).

            Therefore, we can assume that at least one of $w$ and $w'$ is not in $M$. W.l.o.g., assume that $w \notin M$. Then, there exists a vertex type $\tau \in [2^{|M|}]$ such that  $w \in U_\tau = Q_\tau$. Therefore, either $\alpha'_p \in A'_{\tau,\sigma}$ or $\alpha'_q \in A'_{\sigma,\tau}$ (depending on whether the swap happens in the first or the last turn). Observe that $A'_{\tau,\sigma} = A_{\tau, \sigma}$ and $A'_{\sigma,\tau} = A_{\sigma,\tau}$. Indeed, since $U_\tau = Q_\tau$, we have that $Q_\tau \le 3|A'| = 3|A^p|$. Therefore, by the construction of $A^p$, we have that $A_{\tau, \sigma} \cup A_{\sigma,\tau} \subseteq A^p=A'$. Therefore, by the selection of $N$, we have that:
            \begin{itemize}
                \item either the swap happens in the first turn and $\alpha'_p \in A'_{\tau, \sigma} = A_{\tau,\sigma}$ or
                \item the swap happens in the last turn and $\alpha'_q \in A'_{\sigma,\tau} = A_{\sigma,\tau}$.
            \end{itemize}
            Also, by the selection of $N$ we have that:
            \begin{itemize}
                \item either the swap happens in the first turn and $N(\alpha'_p) = \alpha'_p$ or 
                \item the swap happens in the last turn and $N(\alpha'_q) =\alpha'_q$.
            \end{itemize}
        \end{proofclaim}
        
            It remains to show that, if the swap happens in the first turn then $N(\alpha'_q) = \alpha'_q$ otherwise $N(\alpha'_p) = \alpha'_p$.

            \medskip
            \noindent\textbf{Case a.} (the swap happens in the first turn): If $s'_0(\alpha'_q)=w' \in M$, then $\alpha'_q \in A_0$. Thus, by the selection of $N$, it follows that $N(\alpha'_q) = \alpha'_q$. Therefore, we may assume that there exists a vertex type $\phi\in [2^{|M|}]$ such that $s'_0(\alpha'_q)=w' \in U_{\phi} =Q_{\phi}$. Using similar arguments as before, we can conclude that there also exists a vertex type $\psi \in [2^{|M|}]$ such that $\alpha'_q \in A'_{\phi,\psi} = A_{\phi,\psi}$. Thus, $N(\alpha'_q) = \alpha'_q$.

            \medskip
            \noindent\textbf{Case b.} (the swap happens in the last turn): If $t'(\alpha'_p) = s'_m(\alpha'_p)=w' \in M$, then $\alpha'_p \in A_0$. Thus, by the selection of $N$, we have that $N(\alpha'_p) = \alpha'_p$. Therefore, we may assume that there exists a vertex type $\phi\in [2^{|M|}]$ such that $t'(\alpha'_p) = s'_m(\alpha'_p) = w' \in U_{\phi} =Q_{\phi}$. Using similar arguments as before, we can conclude that there also exists a vertex type $\psi \in [2^{|M|}]$ such that $\alpha'_p \in A'_{\psi,\phi} = A_{\psi,\phi}$. Thus, $N(\alpha'_p) = \alpha'_p$.

            \medskip 
            
            We are finally ready to show that $s'$ contains no swaps. Let us assume that this schedule contains a swap. By the discussion above, this swap is between two agents $\alpha'_p$ and $\alpha'_q$ where $N(\alpha'_p) = \alpha'_p$ and $N(\alpha'_q)=\alpha'_q$, and happens either in the first turn or in the last. We will show that this leads to a contradiction.

            Assume that the swap happens in the first turn.
            Recall that $s'_{0}(\alpha'_p) = s'_{1}(\alpha'_q) = w$ and $s'_{0}(\alpha'_p) = s'_{1}(\alpha'_q) = w'$. We have also proven that $w,w' \in  M \cup \{ U_\tau \mid \tau \in [2^{|M|}] \text{ and } U_\tau=Q_\tau \}$.
            Since $s'_{1}(\alpha'_q) = w$ and $w \in  M \cup \{ U_\tau \mid \tau \in [2^{|M|}] \text{ and } U_\tau=Q_\tau \}$, we have that $s'_1(\alpha'_q) = s_1(N(\alpha'_q)) = s_1(\alpha'_q)$. Similarly, since $s'_{1}(\alpha'_p) = w'$ and $w' \in  M \cup \{ U_\tau \mid \tau \in [2^{|M|}] \text{ and } U_\tau=Q_\tau \}$, we have that $s'_1(\alpha'_p) = s_1(N(\alpha'_p)) = s_1(\alpha'_p)$.
            Also, since $N(\alpha'_p) = \alpha'_p$ and $N(\alpha'_q)=\alpha'_q$, we have that $s_0(\alpha'_p) = s'_0(\alpha'_p)$ and $s_0(\alpha'_q) = s'_0(\alpha'_q)$. Therefore, there exists a swap in the feasible solution $s_1,\ldots,s_m$ of $\mathcal{I}$. This is a contradiction.
    
            We will prove the same happens in the case where the swap happen in the last turn. In this case we have that $s'_{m-1}(\alpha'_p) = s'_{m}(\alpha'_q) = w$ and $s'_{m-1}(\alpha'_p) = s'_{m}(\alpha'_q) = w'$. We have also proven that $w,w' \in  M \cup \{ U_\tau \mid \tau \in [2^{|M|}] \text{ and } U_\tau=Q_\tau \}$.
            Since $s'_{m-1}(\alpha'_p) = w$ and $w \in  M \cup \{ U_\tau \mid \tau \in [2^{|M|}] \text{ and } U_\tau=Q_\tau \}$, we have that $s'_{m-1}(\alpha'_p) = s_{m-1}(N(\alpha'_p)) = s_{m-1}(\alpha'_p)$. Similarly, since $s'_{m-1}(\alpha'_q) = w'$ and $w' \in  M \cup \{ U_\tau \mid \tau \in [2^{|M|}] \text{ and } U_\tau=Q_\tau \}$, we have that $s'_{m-1}(\alpha'_q) = s_{m-1}(N(\alpha'_q)) = s_{m-1}(\alpha'_q)$.
            Also, by the construction of $s'_m$ and $t'$, we have that $s_m(\alpha'_p) = t(\alpha'_p) = t'(\alpha'_p) = s'_m(\alpha'_p)$ and $s_m(\alpha'_q) = t(\alpha'_q) = t'(\alpha'_q) = s'_m(\alpha'_q)$. Therefore, there exists a swap in the feasible solution $s_1,\ldots,s_m$ of $\mathcal{I}$. This is a contradiction.
        
\medskip 

            Therefore,  $s'$ is a feasible solution for $\mathcal{K}$. It remains to prove that $|s'_i(A') \cap M| = |s_i(A) \cap M|$ for all $i \in [m]$. Before doing so, notice that $|s_i(A) \cap M| \ge |A|-|Q|$, as otherwise there is not enough space in $|s_i(A) \cap Q| > |Q|$ which contradicts the feasibility of $s_1,\ldots,s_m$.
    
            \begin{claim}
                 For the feasible solution $s'$ of $\mathcal{K}$, it holds that $|s'_i(A') \cap M| = |s_i(A) \cap M|$ for all $i \in [m]$.
            \end{claim}
    
            \begin{proofclaim}
                Notice that $s'_0(A') = s_0(A')$ and and $s'_m(A')=t'(A')= t(A')$. Therefore, the claim is trivial for the turns $0$ and $m$.
    
                Assume that there exists a turn $i \in [m-1] $ such that $|s'_i(A') \cap M| <k$. Let $A'_i$ be the set of agents such that $s'_{i}(A'_i) = s'_i(A') \cap M$ and $A_i$ the set of agents such that $s_{i}(A_i) = s_i(A) \cap M$. Also, recall that $|A_i| \ge k$ by the selection of $k$. 
                Moreover, by construction, $B \supseteq A_i$ (recall that $A'$ is matched to $B \subseteq A$ by $N$).

                Therefore, it suffices to prove that $|A'\setminus A'_i| \le |B \setminus A_i|\le |B|-k = |A'|-k$. This holds if $N(A'\setminus A'_i) \subseteq B \setminus A_i$.
                To prove this, consider an agent $\alpha' \in A'\setminus A'_i$. Since $s'_i(\alpha')\in Q$, we can conclude that $s_i(N(\alpha')) = s'_i(\alpha') \in Q$ (by the construction of $s'_i$). Therefore, $N(\alpha') \in B\setminus A_i$.
                Thus $N(A'\setminus A'_i) \subseteq B \setminus A_i$.
            \end{proofclaim}        
            
        This concludes the proof of Lemma~\ref{lemma_distance_clique:K_admits_appropriate_solution}, and completes the Step $3$.
    \end{proof}

    \subsubsection{Step $4$.}\label{step4}
    In this step we compute an optimal feasible solution $s'_1,\ldots,s'_{m'}$ of $\mathcal{K}$ such that $|s'_i(A') \cap M| = |s_i(A) \cap M|$ for all $i \in [m]$.

    \begin{lemma} \label{lemma_distance_clique:shortest_path_for_K_solution}
        Given a natural number $k$, we can compute an optimal solution $s'_1,\ldots,s'_{m'}$ of $\mathcal{K}$ such that $|s'_i(A') \cap M| \ge k$ for all $i \in [m]$ in time $O\Big( \big( \frac{|U|!}{(|U|-|A|)!}\big)^2 \Big)$.
    \end{lemma}
 
    \begin{proof}
        Let $pos_i: A\rightarrow U$, $i \in [\frac{|U|!}{(|U|-|A|)!}]$, be all the possible placements of agents of $A'$ on the vertices of $U$. That is, for any $i \in [\frac{|U|!}{(|U|-|A|)!}]$, if $\alpha, \beta \in A$ and $\alpha \neq \beta$, then $pos_i(\alpha) \neq pos_i (\beta)$ (i.e., all agents occupy vertices of $U$ but no two agents occupy the same vertex).
        
        We will create an auxiliary graph $H$ as follows. For each $i \in [\frac{|U|!}{(|U|-|A|)!}]$, create one vertex $v_i$ that represents the placement $pos_i$. Let us call $v_s$ and $v_t$ the vertices which represents the starting and terminal placements $s'_0$ and $t'$, respectively. This completes the definition of the vertex set of $H$.
        For the edge set of $H$, we add an edge between $v_i$ and $v_j$, $i,j \in [\frac{|U|!}{(|U|-|A|)!}]$, if and only if:
        \begin{itemize}
            \item for all $\alpha \in A'$ we have $pos_j (\alpha) \in N_{G[U]}[pos_i (\alpha)]$ and 
            \item there is no pair of agents $\alpha, \beta \in A'$ with $\alpha \neq \beta$ such that $pos_i (\alpha) = pos_j (\beta)$ and $pos_i (\beta) = pos_j (\alpha)$.
        \end{itemize} 

        The following observation allows us to use $H$ in order to find feasible solutions of $\mathcal{K}$.

        \begin{observation}
        Any feasible solution of $\mathcal{K}$ corresponds to a walk in $H$ that starts from $v_s$ and ends on $v_t$. The reverse is also true. That is, any walk in $H$ that starts from $v_s$ and ends to $v_t$ corresponds to a  feasible solution for $\mathcal{K}$.
        \end{observation}

        The previous observation allows us to compute an optimal solution of $\mathcal{K}$ in time $O(|E(H)| + |V(H)|) = O\big( \big( \frac{|U|!}{(|U|-|A|)!}\big)^2 \big)$ by computing a shortest path of $H$ between $v_s$ and $v_t$. Notice that it suffices to delete the vertices $v_i \in V(H)$ which represent placements $pos_i$ such that $|pos_i(A') \cap M| < k$ to guarantee that there is no $i \in [m]$ such that $|s'_i(A') \cap M| < k$, .
    \end{proof}

    Combining Lemmas~\ref{lemma_distance_clique:K_admits_appropriate_solution} and~\ref{lemma_distance_clique:shortest_path_for_K_solution} yields a feasible solution $s'_1,\ldots,s'_{m'}$ for $\mathcal{K}$ with the previous properties and $m'\le m$, if there exists a feasible solution $s_1,\ldots,s_m$ for $\mathcal{I}$. This concludes the Step $4$.

    \subsubsection{Step $5$.}\label{step5}
    Here we prove that we can use the feasible solution $s'_1,\ldots,s'_{m'}$ we computed for $\mathcal{K}$ in Step $4$ in order to compute a feasible solution of $\mathcal{I}$ with the same makespan.

    Recall that, when we constructed the set $A'$ of the instance $\mathcal{K}$ (in Step $2$), we had two cases; either $A'=A$ or not. If we are in the first case, then $s'_1,\ldots,s'_{m'}$ is, trivially, also a feasible solution of $\mathcal{I}$ of the same makespan. So, we may assume that $A'\neq A$. 
    Here, by the construction of $A'$, we know that $|A|> \max\{2|A'|,100\}$. Therefore, we have that $|A\setminus A'| >\max \{|A'|,50\}\ge 50$.

    Now, we will create a (not necessarily feasible) solution $s''_1,\ldots,s''_{m'}$ of $\mathcal{I}$ such that:
    \begin{itemize}
        \item $s''_i(\alpha') = s'_i(\alpha')$ for all $i \in [m']$ and $\alpha' \in A'$,
        \item $s''_i(A\setminus A') \subseteq Q$ for all $i \in [m']$ and
        \item $s''_{m'}(A\setminus A') = t(A\setminus A') $.
    \end{itemize}

    We start by setting $s''_i(\alpha') = s'_i(\alpha')$ for all $i \in [m']$ and $\alpha' \in A'$. Then we deal with the agents $\alpha \in A'\setminus A$ and $i \in [m']$. To do so, we first need to give some definitions. Let $D_0= Q \setminus s_0(A')$ and $D_m=t(A\setminus A')$.
    Then, for each $i\in [m-1]$, we define $D_{i} =  Q \setminus s'_{i}(A')$. The set $D_{i}$ gives us the vertices of $Q$ that are not occupied by any agent of $A'$ during the turn $i$. Also, by the definition of $s'_1,\ldots,s'_{m'}$, we have that $|D_i| \ge |A| -( |A'|- k) $, for any $i \in [m]$, as  $|s'_i(A') \cap M|\ge k$ and $k = \max \{0, |A|-|Q| \}$.

    We now make sure that $s''_1,\ldots,s''_{m'}$ is defined in such a way that $s''_i(A \setminus A) \subseteq D_i$ for all $i \in [m]$. Notice that $s''_0(A \setminus A) = s_0(A \setminus A) \subseteq D_0$. Then, we show how we construct $s''_{i+1}$ starting from $s''_i$ and $D_{i+1}$.
    The idea behind this construction is to create a bipartite graph $H_{i+1}$ where the edges $e\in \mathcal{E}_{i+1}$ of any perfect matching $\mathcal{E}_{i+1}$ of $H_{i+1}$ represent the trajectories of the agents.
    We start with a complete bipartite graph $H_{i+1}'= (X_1,X_2,E)$, where $|X_1| = |X_2| = \min\{ |D_{i}|, |D_{i+1}| \}$. Notice that, by the definition of $D_i$, we have that $|X_1|=|X_2| \ge |A \setminus A' |>50$. Let $S_i$ be a set of vertices such that $s''_i(A\setminus A') \subseteq  S_i \subseteq  D_i $ and $|S_i| = |X_1|$. Each vertex of $X_1$ will represent a unique vertex of $S_1$.
    Similarly, let $S_{i+1}$ be a set of vertices such that $D_{i+1} \supseteq S_{i+1} $ and $|S_{i+1}| = |X_2|$.
    We compute a matching of $H'_{i+1}$. In particular, for any $w\in X_1$, let $w'$ be the vertex of $S_i$ represented by $w$ in $H'$. If $w' =  s'_{i+1}(\alpha')$, for $\alpha' \in A'$ and $s'_{i}(\alpha')=y' \in S_{i+1}$, then we remove the edge $wy$, where $y$ represents $y'$ in $H_{i+1}$. We repeat this for all vertices of $X_1$. It is not hard to see that the removed edges are a matching of $H'_{i+1}$ (as, for each $w$, there is no more than one $y$ and each $y$ can be the pair of only one $w$). Let $H_{i+1}$ be the graph after the removal of this matching. Notice that $H_{i+1}$ is a bipartite graph with $|X_1|=|X_2| >50$ and $\delta(H_{i+1}) = |X_1|-1$. By Hall’s Theorem, we know that there exists a matching $\mathcal{E}_{i+1}$ in $H_{i+1}$. Before we proceed, we will further modify $\mathcal{E}_{i+1}$.
    If there exists a pair of edges $w_1y_1, w_2y_2 \in \mathcal{E}_{i+1}$ such that $w_1$ and $y_2$ represent the same vertex $v_1$ of $V(G)$, and $y_1$ and $w_2$ represent the same vertex $v'_1$ in $V(G)$, then we remove $w_1y_1$ and $w_2y_2$ from $\mathcal{E}_{i+1}$ and we add $w_1y_2$ and $w_2y_1$. We repeat this until $\mathcal{E}_{i+1}$ has no such pair of egdes.

    We will now show that $w_1y_2$ and $w_2y_1$ belong in $E(H_{i+1})$. Indeed, the missing edges represent potential swaps between agents of $A\setminus A'$ and agents of $A$. Thus, if an edge is missing, it means that one of the vertices of this edge represents a vertex occupied by an agent of $A'$. This is a contradiction, as both vertices are occupied by agents of $A\setminus A'$. Additionally, observe that after resolving one such pair, we are not creating any new ones. Indeed, in order to create new pairs with this property, we should have one of the two newly added edges to be included in such a pair. W.l.o.g., assume that $w_1y_2$ is included in a new pair of edges with this property. Then, there exists another edge $w'_1y'_2$ such that $y'_2$ represents the same vertex as $w_1$, i.e., the vertex $v_1$. Since each vertex is represented at most once in each set $X_1$ and $X_2$, we have that $y'_2=y_1$. This leads to a contradiction as $w'_1y'_2 \neq w_2y_1$, and $\mathcal{E}_{i+1}$ is a matching of $H_{i+1}$ (thus no other edge $w'_1y_1$ can belong in $\mathcal{E}_{i+1}$).

    We are now able to define $s''_{i+1}$ using the matching $\mathcal{E}_{i+1}$ of $H_{i+1}$.
    Recall that for all agents $\alpha \in A\setminus A'$, we have that $s''_{i}(\alpha) \in S_i$. Let $s''_{i}(\alpha) = v_{\alpha}$ and let $w_\alpha$ be the vertex that represents $v_\alpha$ in $X_1$. Since $\mathcal{E}_{i+1}$ is a perfect matching of $H_{i+1}$, there exists exactly one edge $w_\alpha y_\alpha\in \mathcal{E}_{i+1}$, where $y_\alpha \in X_2$. Recall that $y_\alpha$ represents a vertex $v'_\alpha \in S_{i+1}$. We set $s''_{i+1}(\alpha) = v'_\alpha$.

    Notice that $|s''_{i+1}(A\setminus A)| = |A\setminus A|$ since it is constructed by a matching. Also, $s''_{i+1}(A\setminus A) \cap s''_{i+1}(A') = \emptyset$ since $s''_{i+1}(A\setminus A) \subseteq D_{i+1}$ and $D_{i+1}= Q\setminus s'_{i+1}(A') = s''_{i+1}(A') $. Therefore, for all agents $\alpha\in A$, we have that $s''_{i+1}(\alpha) \in N[s''_{i}(\alpha)]$ and $s''_{i+1}(\alpha) \neq s''_{i+1}(\beta)$ (for any $\beta\neq \alpha$).

    We proceed by proving that $s''_1,\dots,s''_{m'}$ contains no swap between any pair of agents during turn $i+1$.
    We assume that a swap happens between agents $\alpha$ and $\beta$. Therefore, we have that $s''_{i+1}(\alpha) = s''_{i}(\beta)$ and $s''_{i+1}(\beta) = s''_{i}(\alpha)$. Notice that at least one of the $\alpha$ and $\beta$ does not belong in $A'$. Indeed, if both $\alpha,\beta \in A'$, we have $s'_{i+1}(\alpha) = s''_{i+1}(\alpha) = s''_{i}(\beta)=s'_{i}(\beta)$ and $s'_{i+1}(\beta) = s''_{i+1}(\beta) = s''_{i}(\alpha)= s'_{i}(\alpha)$. Therefore, $s'_1\ldots, s'_{m'}$ includes swaps, which is a contradiction.
    Additionally, we know that at least one of $\alpha$ and $\beta$ does not belong in $A\setminus A'$. If that were the case, then there would exist a pair of edges $w_1y_1, w_2y_2 \in \mathcal{E}_{i+1}$ such that $w_1$ and $y_2$ would represent the same vertex $w'_1$ of $V(G)$, and $y_1$ and $w_2$ would represent the same vertex $w'_2$ in $V(G)$. However, this cannot happen by the construction of $\mathcal{E}_{i+1}$.
    Therefore, we can assume, w.l.o.g., that $\alpha \in A\setminus A'$ and $\beta \in A'$.
    Since $\alpha \in A\setminus A'$, $s''_{i}(\alpha) = v$ and  $s''_{i+1}(\alpha) = v'$, and by the definition of $s''_{i+1}$, there exists an edge $w_1y_1\in \mathcal{E}_{i+1}$ such that the vertex $w_1 \in X_1$ represents $v$ and the vertex $y_1 \in X_2$ represents $v'$. Notice that this edge cannot exist in $\mathcal{E}_{i+1}$. Indeed, since $s'_{i+1}(\beta) = s''_{i+1}(\beta) = v$ and  $s'_{i}(\beta) = s'_{i}(\beta)= v'$, then the edge $w_1y_1$ is one of the edges of the matching we removed from $H'_{i+1}$ in order to construct $H_{i+1}$. Therefore, such an edge cannot belong in the perfect matching $\mathcal{E}_{i+1}$ of $H_{i+1}$. Thus, no swaps are happening between any agents in $s''_1,\ldots,s''_{m'}$.

    Up until now, we have shown that $s''_1,\ldots,s''_{m'}$ is a feasible solution of $\mathcal{I}$. What remains to be done is to check whether $s''_{m'} $ is equal to $t$. Since $s''_{i}(\alpha')= s'_{i}(\alpha')$ for all $\alpha' \in A'$, the only problem we could face is from agents that belong in $A\setminus A'$. Consider the sequence $s''_1,\ldots,s''_{m'-1},t$. Since all agents of $A'$ have $t(\alpha) = s'_{m'}(\alpha') = s''_{m'}(\alpha') $, $s''_{m'-1}(A\setminus A') \subseteq Q$ and $s''_{m'}(A\setminus A') = t(A\setminus A') \subseteq Q$, the only problem in this sequence is that it may include swaps in the last turn. We will modify $s''_{m'-1}(\alpha)$, for every $\alpha \in A\setminus A'$, so that no swaps are happening in turns $m'-1$ and $m'$.

    To this end, assume that there are $p\le |A|/2$ swaps in $s''_1,\ldots,s''_{m'-1},t$ and let $(\alpha_j, \beta_j)$, $j \in [p]$, be the pairs such that $\alpha_j $ and $\beta_j$ are swapping positions in the last turn of $s''_1,\ldots,s''_{m'-1},t$. Notice that since $s''_{m'-1}(\alpha') = s'_{m'-1}(\alpha')$ and $t(\alpha') = s'_{m'}(\alpha')$, we know that in each pair at least one of the $\alpha_j$ and $\beta_j$ must belong in $A\setminus A'$ (otherwise $s'_{1}\ldots,s'_{m'}$ was not a feasible solution of $\mathcal{K}$). W.l.o.g., we assume that $\beta_j \in A\setminus A'$, for every $j\in [p]$. We distinguish cases according to whether $p \ge 4$ or not.

    \medskip
    \noindent\textbf{Case 1.} ($p\ge 4$): Our goal is to create a new positioning $s$ to replace $s''_{m'-1}$ in the sequence $s''_{1},\ldots,s''_{m'-1}, t$. We will construct $s$ such that $s(\alpha) = s''_{m'-1}(\alpha)$ for all $\alpha \in A\setminus \{\beta_1,\ldots,\beta_p\}$, $s(\{\beta_1,\ldots,\beta_p\}) = s''_{m-1}(\{\beta_1,\ldots,\beta_p\})$ and $s''_{1},\ldots,s''_{m'-2}, s, t$ is a feasible solution.

    Before we give $s$, we will reorder the pairs $(\alpha_j, \beta_j)$, for every $j \in [p]$.
    We will say that the ordering $\beta_1,\ldots,\beta_p$ is a \emph{bad ordering} if there exists a $j\in[p]$ such that there exists an agents $\alpha\in A$ with $s''_{m'-2}(\alpha)= s''_{m'-1}(\beta_{(j \mod p)+1})$ and $s''_{m'-1}(\alpha) = s''_{m'-1}(\beta_{j})$; the index $j$ for which this is true will be called \emph{bad}. If there is no bad index, we will say that the ordering is \textit{good}.

    \begin{claim}
        There exists a good ordering $\beta'_1,\ldots,\beta'_p $ of the agents $\beta_1,\ldots,\beta_p$.
    \end{claim}

    \begin{proofclaim}
        Assume that $\beta_1,\ldots,\beta_p$ is a bad ordering and thus, that it has $q\ge 1$ bad indices.
        We will create a new ordering that contains strictly less than $q$ bad indices.

        Let $j$ be a bad index in $\beta_1,\ldots,\beta_p$.
        We create the ordering $\beta'_1,\ldots,\beta'_p$ where $\beta'_j = \beta_{(j\mod p)+1}$, $\beta'_{(j \mod p)+1} = \beta_{j}$ and $\beta'_i=\beta_i$, for all $i \notin \{j,(j\mod p)+1\}$. Basically, we exchanged the orders of $j$ and $(j\mod p)+1$.
        Now, we consider cases based on whether $j$ is a bad index in the new ordering or not.

        \medskip
        \noindent\textbf{Case a.} ($j$ is not a bad index in $\beta'_1,\ldots,\beta'_p$):
        We claim that $\beta'_1,\ldots,\beta'_p$ has less bad indices than $q$.
        Indeed, we have removed $j$ from the bad indices; however we may have created some new bad indices by doing so. Since for any index $i \notin \{j-1,j,(j \mod p)+1\}$ we have that the $\beta'_i = \beta_i $ and $ \beta'_{(i\mod p)+1} =\beta_{(i\mod p)+1}$, it follows that $i$ is a bad index if and only if $i$ was a bad index in the starting ordering. Therefore, we need to check if $j-1$ and $(j \mod p)+1$ are now new bad indices (as we have assumed that $j$ is not a bad index).

        We first focus on $j-1$. Since $j$ was a bad index in $\beta_1,\ldots,\beta_p$, there exists an agent $\alpha \in A$ such that $s''_{m'-2}(\alpha) = s''_{m'-1}(\beta_{(j \mod p)+1}) = s''_{m'-1}(\beta'_{j})$ and $s''_{m'-1}(\alpha) = s''_{m'-1}(\beta_{j}) = s''_{m'-1}(\beta'_{(j \mod p)+1})$. Since $\beta'_{(j \mod p)+1} \neq \beta'_{j -1}$, we have that $j-$ is not a bad index in $\beta'_1,\ldots,\beta'_p$.

        Lastly, we deal with $(j \mod p)+1$. Again, since $j$ was a bad index in $\beta_1,\ldots,\beta_p$, there exists an agent $\alpha \in A$ such that $s''_{m'-2}(\alpha) = s''_{m'-1}(\beta_{(j \mod p)+1}) = s''_{m'-1}(\beta'_{j})$ and $s''_{m'-1}(\alpha) = s''_{m'-1}(\beta_{j}) = s''_{m'-1}(\beta'_{(j \mod p)+1})$.
        Since $s''_{m'-1}(\alpha) = s''_{m'-1}(\beta'_{(j \mod p)+1})$, there is no other agent $\beta$ such that $s''_{m'-1}(\beta) = s''_{m'-1}(\beta'_{(j \mod p)+1})$. Therefore $(j \mod p)+1$ is not a bad index.

        \medskip
        \noindent\textbf{Case b.} ($j$ is a bad index in $\beta'_1,\ldots,\beta'_p$):
        In this case there exists
        an agent $\alpha \in A$ such that
        $s''_{m'-2}(\alpha) = s''_{m'-1}(\beta'_{(j \mod p)+1}) = s''_{m'-1}(\beta_{j})$ and $s''_{m'-1}(\alpha) = s''_{m'-1}(\beta'_{j}) = s''_{m'-1}(\beta_{(j \mod p)+1})$.
        We create a third ordering $\beta''_1,\ldots,\beta''_p$ as follows.
        We set $\beta''_{(j \mod p)+1} = \beta_{(j+1\mod p)+1}$, $\beta''_{(j+1 \mod p)+1} = \beta_{(j \mod p)+1}$ and $\beta'_i=\beta_i$ for all $i \notin \{ (j\mod p)+1, (j\mod p)+1\}$.
        We will prove that $\beta''_1,\ldots,\beta''_p$ has strictly less than $q$ bad indices.

        Notice that for any index $i \notin \{j,(j \mod p)+1,(j+1 \mod p)+1\}$, we have that $\beta''_i = \beta_i$ and $\beta''_{(i \mod p)+1} = \beta''_{(i \mod p)+1}$. Therefore, any $i \notin \{j,(j \mod p)+1,(j+1 \mod p)+1\}$ is a bad index in $\beta''_1,\ldots,\beta''_p$ if and only if $i$ is a bad index in $\beta_1,\ldots,\beta_p$. It remains to prove that for any $i \in \{j,(j \mod p)+1,(j+1 \mod p)+1\}$, the index $i$ is not bad in $\beta''_1,\ldots,\beta''_p$.

        Consider the index $j$. Since $j$ was a bad index in $\beta_1,\ldots,\beta_p$, there exists an agent $\alpha \in A$ such that $s''_{m'-2}(\alpha) = s''_{m'-1}(\beta_{(j \mod p)+1}) = s''_{m'-1}(\beta''_{(j+1 \mod p)+1})$ and $s''_{m'-1}(\alpha) = s''_{m'-1}(\beta_{j}) = s''_{m'-1}(\beta''_{j})$.
        Since $s''_{m'-1}(\alpha) = s''_{m'-1}(\beta''_{j})$, there is no other agent $\beta$ such that $s''_{m'-1}(\beta) = s''_{m'-1}(\beta''_{j})$. Therefore $j$ is not a bad index in $\beta''_1,\ldots,\beta''_p$.

        Consider the index $(j \mod p)+1$. Since $j$ was a bad index in $\beta_1,\ldots,\beta_p$, there exists an agent $\alpha \in A$ such that $s''_{m'-2}(\alpha) = s''_{m'-1}(\beta_{(j \mod p)+1}) = s''_{m'-1}(\beta''_{(j+1 \mod p)+1})$ and $s''_{m'-1}(\alpha) = s''_{m'-1}(\beta_{j}) = s''_{m'-1}(\beta''_{j})$.
        Since $s''_{m'-1}(\alpha) = s''_{m'-1}(\beta''_{j}) \neq s''_{m'-1}(\beta''_{(j \mod p) +1})$, it follows that $(j \mod p)+1$ is not a bad index in $\beta''_1,\ldots,\beta''_p$.

        Finally, we consider the index $(j+1 \mod p)+1$. Since $j$ was a bad index in $\beta'_1,\ldots,\beta'_p$, there exists an agent $\alpha \in A$ such that $s''_{m'-2}(\alpha) = s''_{m'-1}(\beta'_{(j \mod p)+1}) = s''_{m'-1}(\beta_{j})$ and $s''_{m'-1}(\alpha) = s''_{m'-1}(\beta'_{j})= s''_{m'-1}(\beta_{(j \mod p)+1})$.
        If $(j+1 \mod p)+1$ is a bad index in $\beta''_1,\ldots,\beta''_p$, there must exist an agent $\beta \in A$ such that $s''_{m'-2}(\beta) = s''_{m'-1}(\beta''_{(j+2 \mod p)+1}) = s''_{m'-1}(\beta_{(j+2 \mod p)+1})$ and $s''_{m'-1}(\beta) = s''_{m'-1}(\beta''_{(j+1 \mod p)+1})= s''_{m'-1}(\beta_{(j \mod p)+1})$.
        Notice that $s''_{m'-1}(\alpha) = s''_{m'-1}(\beta_{(j \mod p)+1}) = s''_{m'-1}(\beta)$. Therefore, $\alpha = \beta$. This is a contradiction as $s''_{m'-2}(\alpha) = s''_{m'-1}(\beta_{j})$, $s''_{m'-2}(\beta) = s''_{m'-1}(\beta''_{(j+2 \mod p)+1}) = s''_{m'-1}(\beta_{(j+2 \mod p)+1})$ and $p\ge 4$.

        We can repeat this process until no bad indices appear in the ordering. Also, we can compute this ordering in polynomial time.
    \end{proofclaim}

    Thus, we may assume that we have a good ordering $\beta'_1,\ldots,\beta'_p$ and let $(\alpha'_j,\beta'_j)$, $j \in [p]$, be the new ordering of the swaps (that respects this new ordering) that exist in $s''_0,\dots,s''_{m'-1},t$ between $s''_{m'-1}$ and $t$.

    We are creating the promised placement $s$ as follows.
    For any $\alpha \in A \setminus \{\beta'_1,\ldots,\beta'_p \}$, we set $s(\alpha) = s''_{m'-1} (\alpha)$. Then, for any $j \in [p]$, we set $s(\beta'_{j}) = s''_{m'-1} (\beta'_{(j \mod p)+1})$.
    We  claim that $s''_{i},\ldots,s''_{m'-2},s,t$ is a feasible solution of $\mathcal{I}$.
    Since $s(\alpha) = s''_{m'-1} (\alpha)$, for all $\alpha \in A \setminus \{\beta'_1,\ldots,\beta'_p \}$ and $s(\{\beta'_1,\ldots,\beta'_p \}) = s''_{m'-1} (\{\beta'_1,\ldots,\beta'_p \}) \subseteq Q$, we just need to prove that there are no swaps in $s''_{i},\ldots,s''_{m'-2},s,t$. 

    Indeed, we will show that there are no swaps in the turn $m'-1$ (i.e., between $s''_{m'-2}$ and $s$).
    Assume that there exists such a swap.
    Since there are no swaps between $s''_{m'-2}$ and $s''_{m'-1}$, at least one of the agents $\alpha,\beta$ that are swapping must belong in the set $\{ \beta'_{i}\mid i \in [p]\}$.
    Let $\beta'_j$ and $\alpha$ be the agents that swap between $s''_{m'-2}$ and $s$. That means that $s''_{m'-2} (\alpha) = s(\beta'_j) = s''_{m'-2}(\beta'_{(j \mod p)+1})$ and $s''_{m'-2} (\alpha) = s (\alpha) = s''_{m'-2} (\beta'_j)$.
    This contradicts that $\beta'_1,\ldots,\beta'_p$ is a good ordering. Therefore, no swaps are happening in the turn $m'-1$.

    Next, assume that there exists a swap in the lats turn (i.e. between $s$ and $t$) and let $\alpha$ and $\beta$ be the agents that are swapping positions.
    Notice that if $\alpha, \beta \notin \{\alpha'_i,\beta'_i \mid j\in [p]\}$, then there is no swap between them.
    Indeed, in that case we have $s(\alpha)=s''_{m'-1}(\alpha)$ and
    $s(\beta)=s''_{m'-1}(\beta)$. Therefore, we may assume that there exist a pair  $\alpha'_j,\beta'_j$, that represents $\alpha, \beta$ and that at least one of $\alpha$ and $\beta$ is in $\{\alpha'_i,\beta'_i \mid j\in [p]\}$. We assume, w.l.o.g., that $\alpha \in \{\alpha'_i,\beta'_i \mid j\in [p]\}$.
    We consider two cases, $\alpha = \alpha'_j$ or $\alpha = \beta'_j$, for some $j \in [p]$.

    \medskip
    \noindent\textbf{Case 1.a} ($\alpha = \alpha'_j$):
    Recall that between $s''_{m'-1}$ and $t$, the agent $\alpha'_j$ was swapping with $\beta'_j$. Therefore, we have that $s''_{m'-1}(\alpha'_j) = t(\beta'_j)$ and $s''_{m-1}(\beta'_j) = t(\alpha'_j)$.
    Now, notice that by the construction of $s$, we have that $s(\beta'_{j-1 \mod p}) = s''_{m'-1}(\beta'_j) = t(\alpha'_j)$ and $t(\beta'_{j-1 \mod p}) = s''_{m'-1}(\alpha'_{j-1 \mod p} ) =s(\alpha'_{j-1 \mod p} ) \neq  s(\alpha'_{j} )$. Therefore, no swap includes $\alpha$.

    \medskip
    \noindent\textbf{Case 1.b} ($\alpha = \beta'_j$):
    Recall that between $s''_{m'-1}$ and $t$, the agent $\alpha'_j$ was swapping with $\beta'_j$. Therefore, we have that $s''_{m'-1}(\alpha'_j) = t(\beta'_j)$ and $s''_{m'-1}(\beta'_j) = t(\alpha'_j)$.
    By the construction of $s$, we have that $s(\beta'_j)= s''_{m'-1}(\beta'_{(j\mod p)+1}) = t (\alpha'_{(j\mod p)+1})$.
    Additionally, $s(\alpha'_{(j\mod p)+1})= s''_{m'-1}(\alpha'_{(j\mod p)+1}) = t(\alpha'_{(j\mod p)+1)} \neq t(\beta'_j)$.
    Therefore, no swap includes $\alpha$.

    This finishes the analysis of the case where there are $p\geq 4$ swaps in $s_1'',\dots,s_{m'}''$.
    
    \medskip
    \noindent\textbf{Case 2.} ($p < 4$):
    Recall that we are in the case $A'\neq A$. Therefore, we know that $|A\setminus A'| \ge 100$.
    Our goal is to select one agent $\gamma_j$, for each $j\in [p]$, and construct $s$ such that:
    \begin{itemize}
        \item $s(\beta_j) = s''_{m-1}(\gamma_j)$ for all $j \in [p]$,
        \item $s(\gamma_j) = s''_{m-1}(\beta_j)$ for all $j \in [p]$ and
        \item $s(\alpha) = s''_{m-1}(\alpha)$ for all $\alpha \in A\setminus \{\beta_j, \gamma_j \mid j\in[p]\}$.
    \end{itemize}
    We will also prove that for an appropriate $\gamma_j$, $j \in [p]$, the schedule $s''_1\ldots, s''_{m-2}, s, t$ is a feasible solution of $\mathcal{I}$.

    Let $A''$ be the set of agents $A\setminus (A' \cup \{\alpha_j,\beta_j \mid j \in [p]\}$. Notice that there are a least $44$ agents in $A''$.
    For each $j \in [p]$, we select an agent $\gamma_j \in A''$ that satisfies the following conditions:
    \begin{itemize}
        \item there is no agent $\alpha \in A'$ such that $s''_{m-2}(\alpha) = s''_{m-1}(\beta_j)$ and $s''_{m-1}(\alpha) = s''_{m-2}(\gamma_j)$,
        \item there is no agent $\alpha \in A'$ such that $s''_{m-2}(\alpha) = s''_{m-1}(\gamma_j)$ and $s''_{m-1}(\alpha) = s''_{m-2}(\beta_j)$,
        \item $s''_{m-2}(\gamma_j) \notin s''_{m-1}(\{\beta_k\mid k \in[p]\})$, for all $j \in [p]$
        \item $s''_{m-2}(\beta_j) \notin s''_{m-1}(\{\gamma_k\mid k \in[p]\})$, for all $j \in [p]$, and 
        \item $s''_{m-2}(\gamma_j) \notin s''_{m-1}(\{\gamma_k\mid k \in[p]\setminus \{j\} \})$, for all $j \in [p]$.
    \end{itemize}

    Notice that each one of the first $4$ conditions excludes at most $p$ agents from the set $A''$. Therefore, there exists $22$ agents in $A'$ that satisfy the the first $4$ conditions. 
    Notice also that by the last condition, we have that the selection of an agent $\gamma$ to be one of the $\gamma_i$, $i \in [p]$, excludes at most $2(p-1)$ agents to be selected together with $\gamma$. Therefore, we can select arbitrarily $\gamma_1$ (from the set of the $22$ agents) then exclude itself and the other $4$ agents from the set and repeat the process.
    Finally, even if  $p\le 3$, we have $17$ available agents for the selection of $\gamma_2$ and (after the arbitrary selection of $\gamma_2$) we have $17$ available agents for the selection of $\gamma3$. Therefore, we can always compute $\gamma_i$, $i\in [p]$. Also, since we can compute the excluded agents in polynomial time and the selection of $\gamma_i$ is arbitrary, we can compute $\gamma_i$, $i \in [p]$, in polynomial time.
    
    We claim that the selection of $\gamma_j$ is such that the schedule $s''_1\ldots, s''_{m-2}, s, t$ is a feasible solution of $\mathcal{I}$.
    Since $s(\{\beta_j,\gamma_j \mid j\in [p]\}) = s''_{m-1}( \{\beta_j,\gamma_j \mid j\in [p]\} )$ and $s(\alpha) = s''_{m-1}(\alpha)$ for all $\alpha \in A\setminus \{\beta_j,\gamma_j \mid j\in [p]\}$, we just need to prove that there are no swaps in $s''_1\ldots, s''_{m-2}, s, t$.
    Notice that swaps can happen only in the last two turns.
    Additionally, any potential swap must include an agent $\alpha \in \{\beta_j,\gamma_j \mid j\in [p]\}$. Indeed, if there exists a swap between a pair $(\alpha,\beta)$ with $\alpha , \beta \in A\setminus \{\beta_j,\gamma_j \mid j\in [p]\}$, then we have that $s(\alpha) =s''_{m-1}(\alpha)$ and $s(\beta) =s''_{m-1}(\beta)$. Therefore, this swap appears between $s''_{m-1}$ and $t$ in the sequence $s''_1,\ldots,s''_{m-1},t$ and $(\alpha,\beta)$ is one of the pairs $(\alpha_j,\beta_j)$. This means that at least one of these agents belongs in $\{\beta_j \mid j\in [p]\}$. This contradicts the assumption that $\alpha , \beta \in A\setminus \{\beta_j,\gamma_j \mid j\in [p]\}$.

    Therefore, we can assume that at least one of $\alpha$ and $\beta$ belongs in $\{\beta_j,\gamma_j \mid j\in [p]\}$.
    Assume, w.l.o.g.,  that $\beta \in \{\beta_j,\gamma_j \mid j\in [p]\}$. We consider two cases according to whether $\beta \in \{\beta_j \mid j\in [p]\}$ or $\beta \in \{\gamma_j \mid j\in [p]\}$.

    \medskip
    \noindent\textbf{Case 2.a.} ($\beta = \beta_j $ for some $j\in [p]$):
    Since we have assumed that there exists a swap between $\alpha$ and $\beta$ we have that:
    \begin{itemize}
        \item either $s(\alpha) = t(\beta)$ and $s(\beta) = t(\alpha)$, or
        \item  $s''_{m-2}(\alpha) = s(\beta)$ and $s''_{m-2}(\beta) = s(\alpha)$.
    \end{itemize}
    We consider each case separately.

    First, we assume that $s(\alpha) = t(\beta)$ and $s(\beta) = t(\alpha)$.
    By the construction of $s$, we have $s(\beta) = s(\beta_j) = s''_{m-1}(\gamma_j)$. Since we have $t(\alpha) =  s''_{m-1}(\gamma_j)$, we claim that $\alpha \notin \{\alpha_k,\beta_k\mid k \in [p]\}$.
        Indeed, if $\alpha \in \{\alpha_k\mid k \in [p]\}$, then $t(\alpha) \in s''_{m-1}(\{\beta_k\mid k \in [p]\})$ and $s''_{m-1}(\gamma_j)\notin s''_{m-1}(\{\beta_k\mid k \in [p]\})$ (as $\gamma_j \notin \{\beta_k\mid k \in [p]\}$ by definition).
        Additionally, if $\alpha \in \{\beta_k \mid k \in [p]\}$, then
        $t(\alpha) \in s''_{m-1}(\{\alpha_k\mid k \in [p]\})$ and $s''_{m-1}(\gamma_j)\notin s''_{m-1}(\{\alpha_k\mid k \in [p]\})$ (as $\gamma_j \notin \{\alpha_k\mid k \in [p]\}$ by definition).
    Thus, we have that $\alpha \in A\setminus \{\alpha_k,\beta_k\mid k \in [p]\}$. Therefore $t(\beta_j)s=(\alpha) = s''_{m-1}(\alpha)$. This is a contradiction as $t(\beta_k) \in s''_{m-1}(\{ \alpha_k\mid k \in [p]\})$.

    Then, we assume that $s''_{m-2}(\alpha) = s(\beta)$ and $s''_{m-2}(\beta) = s(\alpha)$. Again we have $s(\beta) = s(\beta_j) = s''_{\gamma_j}$.
    Notice that if $\alpha \notin\{ \beta_k,\gamma_k \mid k \in [p]\}$, we have a contradiction as $s''_{m-1}(\alpha) = s(\alpha) = s''_{m-2}(\gamma_j)$ and $s''_{m-1}(\gamma_j) = s(\beta_j) = s''_{m-2}(\alpha)$ which is excluded by the selection of $\gamma_j$. So, we may assume that $\alpha \in\{ \beta_k,\gamma_k \mid k \in [p]\}$.
    If $\alpha =\beta_k$, for $k \neq j$, we have that $s''_{m-2}(\beta_j) = s(\alpha)= s(\beta_k) =  s''_{m-1}(\gamma_k)$. This contradicts the selection of $\gamma_k$.
    Finally, if $\alpha =\gamma_k$, for a $k \neq j$, then $ s''_{m-2}(\gamma_k)= s''_{m-2}(\alpha) = s(\beta_j)= s''_{m-1}(\gamma_j)$. This contradicts the selection of $\gamma_j$ and $\gamma_k$.

    \medskip
    \noindent\textbf{Case 2.b.} ($\beta = \gamma_j $ for some $j\in [p]$):
    Since we have assumed that there exists a swap between $\alpha$ and $\beta$, we have that:
    \begin{itemize}
        \item either $s(\alpha) = t(\beta)$ and $s(\beta) = t(\alpha)$ or
        \item  $s''_{m-2}(\alpha) = s(\beta)$ and $s''_{m-2}(\beta) = s(\alpha)$.
    \end{itemize}
    We consider each case separately.
    First, we assume that $s(\alpha) = t(\beta)$ and $s(\beta) = t(\alpha)$. Here we have that $s(\beta) = s(\gamma_j) = s''_{m-1}(\beta_j)$.
    Recall that $t(\alpha_j) = s''_{m-1}(\beta_j)$ (by the definition of $(\alpha_j,\beta_j)$). Therefore, $\alpha = \alpha_j$. This leads to a contradiction as $s(\alpha) = s(\alpha_j)= s''_{m-1}(\alpha_j) \neq t(\gamma_j) $ (since it is known that $s''_{m-1}(\alpha_j) = t(\beta_j)$).

    Then, we assume that $s''_{m-2}(\alpha) = s(\beta)$ and $s''_{m-2}(\beta) = s(\alpha)$.
    Again, we have that $s(\beta) = s(\gamma_j) = s''_{m-1}(\beta_j)$.
    Notice that, if $\alpha \notin\{ \beta_k,\gamma_k \mid k \in [p]\}$,
    we have a contradiction as $s''_{m-1}(\alpha) = s(\alpha) = s''_{m-2}(\beta_j)$ and $s''_{m-1}(\beta_j) = s(\gamma_j) = s''_{m-2}(\alpha)$ which is excluded by the selection of $\gamma_j$.
    So, we may assume that $\alpha \in\{ \beta_k,\gamma_k \mid k \in [p]\}$.
        If $\alpha =\gamma_k$, for $k \neq j$, we have that $s''_{m-2}(\gamma_k) = s''_{m-2}(\alpha) = s(\beta)= s(\gamma_j) = s''_{m-1}(\beta_j)$ which contradicts the selection of $\gamma_k$.
    Finally, if $\alpha =\beta_k$, for a $k \neq j$, we have that
    $s(\alpha) = s(\beta_k) = s''_{m-1}(\gamma_k)$. Also, $s(\alpha)= s''_{m-2}(\beta) = s''_{m-2}(\gamma_j)$. Therefore $s''_{m-1}(\gamma_k) =  s''_{m-2}(\gamma_j)$, which contradicts the selection of $\gamma_j$ and $\gamma_k$.

    This completes the proof that $s''_{1},\ldots, s''_{m'-2},s,t$ has no swaps.
    Thus, $s''_{1},\ldots, s''_{m'-2},s,t$ is an optimal solution of $\mathcal{I}$. This is also the end of Step $5$, concluding the proofs of all the components used in the proof of Theorem~\ref{thm:fpt-distance-to-clique}.
\section{Conclusion}
In this paper we continued down the path of studying the parameterized complexity of the \MAPF problem, an approach that has recently started gathering interest. Although the intractability of the problem had already been established even for restrictive cases, our hardness results are quite unexpected. Indeed, we stress a graph having bounded vertex cover or max leaf number has a very restricted structure, which leads to FPT algorithms the vast majority of times. Thus, our results in Theorems~\ref{thm:vc-hardness:mapf-vc-hard} and~\ref{thm:pancake} are of great theoretical importance. Indeed, \MAPFShort has the potential to be among the few candidates to give birth to reductions that show hardness when considering these parameters for a plethora of other problems. Moreover, these results provide further motivation towards a more heuristic approach. On the other hand, the FPT algorithm we presented has the potential to be of great practical use, depending on the specific characteristics of the topology on which the \MAPFShort problem has to be solved, which can be the object of a dedicated study.

\begin{acks}
This work was co-funded by the European Union under the project Robotics and advanced industrial production (reg.\ no.\ CZ.02.01.01/00/22\_008/0004590).
JMK was additionally supported by the Grant Agency of the Czech Technical University in Prague, grant No.\ SGS23/205/OHK3/3T/18. FF and NM acknowledge the support by the CTU Global postdoc fellowship program. NM is also partially supported by Charles Univ. projects UNCE 24/SCI/008 and PRIMUS 24/SCI/012, and by the project 25-17221S of GAČR.
DK, JMK, and MO acknowledge the support of the Czech Science Foundation Grant No.\ 22-19557S. 
TAV was partially supported by Charles Univ. project UNCE 24/SCI/008 and partially by the project 22-22997S of GA\v{C}R. 

The preliminary version of our work was presented in AAAI'25~\cite{FKKMOV25}.
\end{acks}

\bibliographystyle{ACM-Reference-Format}
\bibliography{sample-base}

\appendix

\section{Research Methods}

\subsection{Part One}

Lorem ipsum dolor sit amet, consectetur adipiscing elit. Morbi
malesuada, quam in pulvinar varius, metus nunc fermentum urna, id
sollicitudin purus odio sit amet enim. Aliquam ullamcorper eu ipsum
vel mollis. Curabitur quis dictum nisl. Phasellus vel semper risus, et
lacinia dolor. Integer ultricies commodo sem nec semper.

\subsection{Part Two}

Etiam commodo feugiat nisl pulvinar pellentesque. Etiam auctor sodales
ligula, non varius nibh pulvinar semper. Suspendisse nec lectus non
ipsum convallis congue hendrerit vitae sapien. Donec at laoreet
eros. Vivamus non purus placerat, scelerisque diam eu, cursus
ante. Etiam aliquam tortor auctor efficitur mattis.

\section{Online Resources}

Nam id fermentum dui. Suspendisse sagittis tortor a nulla mollis, in
pulvinar ex pretium. Sed interdum orci quis metus euismod, et sagittis
enim maximus. Vestibulum gravida massa ut felis suscipit
congue. Quisque mattis elit a risus ultrices commodo venenatis eget
dui. Etiam sagittis eleifend elementum.

Nam interdum magna at lectus dignissim, ac dignissim lorem
rhoncus. Maecenas eu arcu ac neque placerat aliquam. Nunc pulvinar
massa et mattis lacinia.

\section{Reproducibility Checklist for JAIR}

Select the answers that apply to your research -- one per item. 

\subsection*{All articles:}

\begin{enumerate}
    \item All claims investigated in this work are clearly stated. 
    [yes]
    \item Clear explanations are given how the work reported substantiates the claims. 
    [yes]
    \item Limitations or technical assumptions are stated clearly and explicitly. 
    [yes]
    \item Conceptual outlines and/or pseudo-code descriptions of the AI methods introduced in this work are provided, and important implementation details are discussed. 
    [yes]
    \item 
    Motivation is provided for all design choices, including algorithms, implementation choices, parameters, data sets and experimental protocols beyond metrics.
    [yes]
\end{enumerate}

\subsection*{Articles containing theoretical contributions:}
Does this paper make theoretical contributions? 
[yes] 

If yes, please complete the list below.

\begin{enumerate}
    \item All assumptions and restrictions are stated clearly and formally. 
    [yes]
    \item All novel claims are stated formally (e.g., in theorem statements). 
    [yes]
    \item Proofs of all non-trivial claims are provided in sufficient detail to permit verification by readers with a reasonable degree of expertise (e.g., that expected from a PhD candidate in the same area of AI). [yes]
    \item
    Complex formalism, such as definitions or proofs, is motivated and explained clearly.
    [yes]
    \item 
    The use of mathematical notation and formalism serves the purpose of enhancing clarity and precision; gratuitous use of mathematical formalism (i.e., use that does not enhance clarity or precision) is avoided.
    [yes]
    \item 
    Appropriate citations are given for all non-trivial theoretical tools and techniques. 
    [yes]
\end{enumerate}

\subsection*{Articles reporting on computational experiments:}
Does this paper include computational experiments? [no]

If yes, please complete the list below.
\begin{enumerate}
    \item 
    All source code required for conducting experiments is included in an online appendix 
    or will be made publicly available upon publication of the paper.
    The online appendix follows best practices for source code readability and documentation as well as for long-term accessibility.
    [yes/partially/no]
    \item The source code comes with a license that
    allows free usage for reproducibility purposes.
    [yes/partially/no]
    \item The source code comes with a license that
    allows free usage for research purposes in general.
    [yes/partially/no]
    \item 
    Raw, unaggregated data from all experiments is included in an online appendix 
    or will be made publicly available upon publication of the paper.
    The online appendix follows best practices for long-term accessibility.
    [yes/partially/no]
    \item The unaggregated data comes with a license that
    allows free usage for reproducibility purposes.
    [yes/partially/no]
    \item The unaggregated data comes with a license that
    allows free usage for research purposes in general.
    [yes/partially/no]
    \item If an algorithm depends on randomness, then the method used for generating random numbers and for setting seeds is described in a way sufficient to allow replication of results. 
    [yes/partially/no/NA]
    \item The execution environment for experiments, the computing infrastructure (hardware and software) used for running them, is described, including GPU/CPU makes and models; amount of memory (cache and RAM); make and version of operating system; names and versions of relevant software libraries and frameworks. 
    [yes/partially/no]
    \item 
    The evaluation metrics used in experiments are clearly explained and their choice is explicitly motivated. 
    [yes/partially/no]
    \item 
    The number of algorithm runs used to compute each result is reported. 
    [yes/no]
    \item 
    Reported results have not been ``cherry-picked'' by silently ignoring unsuccessful or unsatisfactory experiments. 
    [yes/partially/no]
    \item 
    Analysis of results goes beyond single-dimensional summaries of performance (e.g., average, median) to include measures of variation, confidence, or other distributional information. 
    [yes/no]
    \item 
    All (hyper-) parameter settings for 
    the algorithms/methods used in experiments have been reported, along with the rationale or method for determining them. 
    [yes/partially/no/NA]
    \item 
    The number and range of (hyper-) parameter settings explored prior to conducting final experiments have been indicated, along with the effort spent on (hyper-) parameter optimisation. 
    [yes/partially/no/NA]
    \item 
    Appropriately chosen statistical hypothesis tests are used to establish statistical significance
    in the presence of noise effects.
    [yes/partially/no/NA]
\end{enumerate}

\subsection*{Articles using data sets:}
Does this work rely on one or more data sets (possibly obtained from a benchmark generator or similar software artifact)? 
[no]

If yes, please complete the list below.
\begin{enumerate}
    \item 
    All newly introduced data sets 
    are included in an online appendix 
    or will be made publicly available upon publication of the paper.
    The online appendix follows best practices for long-term accessibility with a license
    that allows free usage for research purposes.
    [yes/partially/no/NA]
    \item The newly introduced data set comes with a license that
    allows free usage for reproducibility purposes.
    [yes/partially/no]
    \item The newly introduced data set comes with a license that
    allows free usage for research purposes in general.
    [yes/partially/no]
    \item All data sets drawn from the literature or other public sources (potentially including authors' own previously published work) are accompanied by appropriate citations.
    [yes/no/NA]
    \item All data sets drawn from the existing literature (potentially including authors’ own previously published work) are publicly available. [yes/partially/no/NA]
    \item All new data sets and data sets that are not publicly available are described in detail, including relevant statistics, the data collection process and annotation process if relevant.
    [yes/partially/no/NA]
    \item 
    All methods used for preprocessing, augmenting, batching or splitting data sets (e.g., in the context of hold-out or cross-validation)
    are described in detail. [yes/partially/no/NA]
\end{enumerate}

\end{document}